\documentclass[12pt]{article}
\usepackage{amsmath,amssymb,amsthm}
\usepackage{graphicx,psfrag,epsf,caption,subcaption}
\usepackage{enumerate,array,multirow}
\usepackage{arydshln}
\usepackage{url} 
\usepackage{xcolor,hyperref} 
\usepackage{natbib}

\usepackage{authblk} 
\usepackage{pdflscape}

\DeclareMathOperator*{\argmin}{arg\,min}
\newcommand{\bC}{\boldsymbol{C}}
\newcommand{\bX}{\boldsymbol{X}}
\newcommand{\bbeta}{\boldsymbol{\beta}}
\newcommand{\btau}{\boldsymbol{\tau}}
\newcommand{\btheta}{{\boldsymbol{\theta}}}
\newcommand{\bTheta}{\boldsymbol{\Theta}}
\newcommand{\bgamma}{\boldsymbol{\gamma}}

\newcommand{\mI}{\mathbb{I}}

\newcommand{\cI}{\mathcal{I}}
\newcommand{\cA}{\mathcal{A}}

\newcommand{\cS}{\mathcal{S}}

\newcommand{\bbE}{\mathbb{E}}
\newcommand{\bbI}{\mathbb{I}}

\newtheorem{Thm}{Theorem}

\newtheorem{lemma}{Lemma}

\graphicspath{{figs/}}
\newcommand{\blind}{0}
\begin{document}
\def\spacingset#1{\renewcommand{\baselinestretch}%
{#1}\small\normalsize} \spacingset{1}

\title{Partition-Mallows Model and Its Inference \\ for Rank Aggregation}

  \author[1]{Wanchuang Zhu}
  \author[2]{Yingkai Jiang}
  \author[3]{Jun S. Liu}
  \author[1]{Ke Deng}
  \affil[1]{Center for Statistical Science \& Department of Industrial Engineering, Tsinghua University, Beijing 100084, China}
  \affil[2]{Yau Mathematical Sciences Center \& Department of Mathematical Sciences, Tsinghua University, Beijing 100084, China}
  \affil[3]{Department of Statistics, Harvard University, Cambridge, MA 02138, USA}
  \date{}
  \maketitle

\if0\blind

\bigskip
\begin{abstract}
Learning how to aggregate ranking lists has been an active research area for many years and its advances have played a vital role in many applications ranging from bioinformatics to internet commerce. The problem of discerning reliability of rankers based only on the rank data is of great interest to many practitioners, but has received less attention from researchers. By dividing the ranked entities into two disjoint groups, i.e., relevant and  irrelevant/background ones, and incorporating the Mallows model for the relative ranking of relevant entities, we propose a framework for rank aggregation that can not only distinguish quality differences among the rankers but also provide the detailed ranking information for relevant entities. Theoretical properties of the proposed approach are established, and its advantages over existing approaches are demonstrated via simulation studies and real-data applications. Extensions of the proposed method to handle partial ranking lists and conduct covariate-assisted rank aggregation are also discussed. 

\end{abstract}

\noindent%
{\it Keywords:} 
Meta-analysis; Heterogeneous rankers; Mallows model; Partial ranking lists; Covariate-assisted aggregation.
\newpage
\spacingset{1.5} 
\section{Introduction}
\label{sec:intro}

Rank data arise naturally in many fields, such as web searching \citep{renda2003web}, design of recommendation systems \citep{baltrunas2010group} and genomics \citep{BADER20111099}. Many probabilistic models have been proposed for analyzing this type of data, among which the Thurstone model \citep{Thurstone1927}, the Mallows model \citep{mallows1957non} and the Plackett-Luce model \citep{luce1959,Plackett1975} are the most well-known representatives. The Thurstone model assumes that each entity possesses a hidden score and all the scores come from a joint probability distribution. The Mallows model is a location model defined on the permutation space of ordered entities, in which the probability mass of a permuted order is an exponential function of its distance from the true order. The Plackett-Luce model assumes that the preference of entity $E_i$ is associated with a weight $w_i$, and describes a recursive procedure for generating a random ranking list: entities are picked one by one with the probability proportional to their weights in a sequential fashion without replacement, and ranked based on their order of being selected.

Rank aggregation aims to derive a ``better'' aggregated ranking list $\hat\tau$ from multiple ranking lists $\tau_1, \tau_2,\cdots, \tau_m$. It is a classic problem and has been studied in a variety of contexts for decades. Early applications of rank aggregation can be traced back to the 18th-century France, where the idea of rank aggregation was proposed to solve the problem of political elections \citep{de1781memoire}. In the past 30 years, efficient rank aggregation algorithms have played important roles in  many fields, such as web searching \citep{renda2003web}, information retrieval \citep{fagin2003efficient}, design of recommendation systems \citep{baltrunas2010group}, social choice studies \citep{porello2012ontology,soufiani2014statistical}, genomics \citep{BADER20111099} and bioinformatics \citep{2010Integration,chen2016drhp}. 

Some popular approaches for rank aggregation are based on certain summary statistics. These methods simply calculate a summary statistics, such as the mean, median or geometric mean, for each entity $E_i$ based on its rankings across different ranking lists, and obtain the aggregated ranking list based on these summary statistics. 
Optimization-based methods obtain the aggregated ranking by minimizing a user-defined objective function, i.e., let $\hat{\tau} = \arg\min\limits_{\tau} \dfrac{1}{m} \sum\limits_{i=1}^m d\left(\tau, \tau_i\right)$, where distance measurement $d(\cdot,\cdot)$ could be either \textit{Spearman's footrule distance} \citep{diaconis1977spearman} or the \textit{Kendall tau distance} \citep{diaconis1988group}.
More detailed studies on these optimization-based methods can be found in \citet{young1978consistent,young1988condorcet,dwork2001rank}.

In early 2000s, a novel class of Markov chain-based methods have been proposed  \citep{dwork2001rank,2010Integration,Lin2010Space,Deconde2011Combining}, which first use the observed ranking lists to construct a probabilistic transition matrix 
among the entities and then use the magnitudes of the entities' equilibrium probabilities of the resulting Markov chain to rank them.
The boosting-based method \textit{RankBoost} \citep{freund2003efficient}, employs a \textit{feedback function} $\Phi(i,j)$ to construct the final ranking, where $\Phi(i,j)>0$  (or $\leq 0$)  indicates that entity $E_i$ is (or is not) preferred to entity $E_j$.
Some statistical methods utilize aforementioned probabilistic models (such as the Thurstone model) and derive the maximum likelihood estimate (MLE) of the final ranking. More recently, researchers have began to pay attention to rank aggregation methods for pairwise comparison data \citep{rajkumar2014statistical,chen2015spectral,fanjianqing2017Spectral}.

We note that all aforementioned methods assume that the rankers of interest are equally reliable. In practice, however, it is very common that some rankers are more reliable than the others, whereas some are nearly non-informative and may be regarded as ``spam rankers''. 
Such differences in rankers' qualities, if ignored in analysis, may significantly corrupt the rank aggregation and lead to seriously misleading results. 
To the best of our knowledge, the earliest effort to address this critical issue can be traced to \citet{aslam2001models}, which derived an aggregated ranking list by calculating a weighted summation of the observed ranking lists, known as the \textit{Borda Fuse}. 
\citet{2010Integration} extended the objective function of \citet{dwork2001rank} to a weighted fashion. 
Independently, \citet{liu2007supervised} proposed a supervised rank aggregation to determine weights of the rankers by training with some external data. 
Although assigning weights to rankers is an intuitive and simple way to handle quality differences, how to scientifically determine these weights is a critical and unsolved problem in the aforementioned works. 

Recently, \citet{deng2014bayesian} proposed BARD, a Bayesian approach to deal with quality differences among independent rankers
without the need of external information. BARD introduces a partition model, which assumes that all involved entities can be partitioned into two groups: the relevant ones and the background ones. 
A rationale of the approach is that, in many applications, distinguishing relevant entities from background ones has the priority over the construction of a final ranking of all entities. 
Under this setting, BARD decomposes the information in a ranking list into three components: (i) the relative rankings of all background entities, which is assumed to be uniform; (ii) the relative ranking of each relevant entity among all background ones, which takes the form of a truncated power-law; and, (iii) the relative rankings of all relevant entities, which is again uniform.  
The parameter of the truncated power-law distribution, which is ranker-specific, naturally serves as a quality measure for each ranker, as a ranker of a higher quality means a less spread truncated power-law distribution.

\citet{fan2019} proposed a stage-wise data generation process based on an extended Mallows model (EMM) introduced by \citet{Fligner1986Distance}. EMM assumes that each entity comes from a two-components mixture model involving a uniform distribution to model non-informative entities, a modified Mallows model for informative entities and a ranker-specific proportion parameter.
\citet{li2020bayesian} followed the Thurstone model framework to deal with available covariates for the entities as well as different qualities of the rankers.
In their model,  each entity is associated with a Gaussian-distributed latent score and a ranking list is determined by the ranking of these scores. 
The quality of each ranker is determined by the standard deviation parameter in the Gaussian model so that a larger standard deviation indicates a poorer quality ranker. 

Although these recent papers have proposed different ways for learning the quality variation among rankers, they all suffer from some limitations. 
The BARD method \citep{deng2014bayesian} simplifies the problem by assuming that all relevant entities are exchangeable. 
In many applications, however, the observed ranking lists often have a strong ordering information for relevant entities, and simply labeling these entities as ``relevant'' without considering their relative rankings tends to lose too much information and oversimplify the problem. 
\citet{fan2019} does not explicitly measure quality differences by their extended Mallows model. 
Although they mentioned that some of their model parameters can indicate the rankers' qualities,
it is not clear how to properly combine multiple indicators to produce an easily interpretable quality measurement. 
The learning framework of \citet{li2020bayesian} based on Gaussian latent variables appears to be more suitable for incorporating covariates than for handling heterogeneous rankers.

In this paper, we propose a \emph{partition-Mallows model} (PAMA), which combines the partition modeling framework of \citet{deng2014bayesian} with the Mallows model, to accommodate the detailed ordering information among the relevant entities. 
The new framework can not only quantify the quality difference of rankers and distinguish relevant entities from background entities like BARD, but also provide an explicit ranking estimate among the relevant entities in rank aggregation. 
In contrast to the strategy of imposing the Mallows on the full ranking lists, which tends to be sensitive to noises in low-ranking entities, the combination of the partition and Mallows models allows us to focus on highly ranked entities, which typically contain high-quality signals in data, and is thus more robust.
Both simulation studies and real data applications show that the proposed approach is superior to existing methods, e.g., BARD and EMM, for a large class of rank aggregation problems.

The rest of this paper is organized as follows. A brief review of BARD and the Mallows model is presented in Section \ref{sec:overview} as preliminaries. The proposed PAMA model is described in Section \ref{sec:model} with some key theoretical properties established. Statistical inference of the PAMA model, including the Bayesian inference and the pursuit of MLE, is detailed in Section \ref{sec:statinfer}.
Performance of PAMA is evaluated and compared to existing methods via simulations in Section \ref{sec:simulation}. Two real data applications are shown in Section \ref{sec:realdata} to demonstrate strength of the PAMA model in practice. Finally, we conclude the article with a short discussion in Section \ref{sec:discussion}. 

\section{Notations and Preliminaries}
\label{sec:overview}
Let $U = \left\{E_1, E_2, \cdots, E_n \right\}$ be the set of entities to be ranked. 
We use ``$E_i \preceq E_j$" to represent that entity $E_i$ is preferred to entity $E_j$ in a ranking list $\tau$, and denote the position of entity $E_i$ in $\tau$ by $\tau(i)$. 
Note that more preferred entities always have lower rankings. 
Our research interest is to aggregate $m$ observed ranking lists, $\tau_1,\ldots, \tau_m$, presumably constructed by $m$ rankers independently into one consensus ranking list which is supposed to be ``better'' than each individual one.

\subsection{BARD and Its Partition Model}
The partition model in BARD \citep{deng2014bayesian} assumes that $U$ can be partitioned into two non-overlapping subsets: $U=U_R\cup U_B$, with $U_R$ representing the set of relevant entities and $U_B$ for the background ones. 
Let $I=\{I_i\}_{i\in U}$ be the vector of group indicators, where $I_i=\mathbb{I}(E_i \in U_R)$ and $\mathbb{I}(\cdot)$ is the indicator function. 
This formulation makes sense in many applications where people are only concerned about a fixed number of top-ranked entities. 
Under this formulation, the information in a ranking list $\tau_k$ can be equivalently represented by a triplet $(\tau_k^0, \tau_k^{1\mid 0}, \tau_k^1)$, where $\tau_k^0$ denotes relative rankings of all background entities, $\tau_k^{1\mid 0}$ denotes relative rankings of relevant entities among the background entities and $\tau_k^1$ denotes relative rankings of all relevant entities. 

\citet{deng2014bayesian} suggested a three-component model for $\tau_k$ by taking advantage of its equivalent decomposition:
\begin{eqnarray}\label{Eq:RankDecomposition}
P(\tau_k \mid I)=P(\tau_k^0,\tau_k^{1\mid 0},\tau_k^1 \mid I)=P(\tau_k^0 \mid I)\times P(\tau_k^{1\mid 0} \mid I)\times P(\tau_k^1 \mid \tau_k^{1\mid 0},I),
\end{eqnarray}
where both $P(\tau_k^0 \mid I)$ (relative ranking of the background entities) and $P(\tau_k^1 \mid \tau_k^{1\mid 0},I) $ (relative ranking of the relevant entities conditional on their set of positions relative to background entities) are uniform, and the relative ranking of a relevant entity $E_i$ among background ones follows a power-law distribution with parameter $\gamma_k>0$, i.e.,
$$P(\tau_k^{1 \mid 0}(i)=t\mid I ) = q(t\mid \gamma_k, n_0) \propto t^{-\gamma_k}\cdot\mI(1\leq t\leq n_0+1),$$
leading to the following explicit forms for the three terms in equation~(\ref{Eq:RankDecomposition}):
\begin{eqnarray}
\label{eqn:bardtau0}
P(\tau_k^0 \mid I)&=&\frac{1}{n_0!},\\
\label{eqn:bardgamma}
P(\tau_k^{1\mid 0} \mid I)&=&\prod_{i \in U_R} q(\tau_k^{1 \mid 0}(i)\mid \gamma_k, I)=\frac{1}{(B_{\tau_k,I})^{\gamma_k}\times(C_{\gamma_k,n_1})^{n_1}},\\
\label{eqn:bardtau1}
P(\tau_k^1 \mid \tau_k^{1\mid 0},I)&=&\frac{1}{A_{\tau_k,I}}\times\mI\big(\tau_k^1 \in \mathcal{A}_{U_R}(\tau_k^{1\mid0})\big),
\end{eqnarray}
where
$n_1=\sum_{i=1}^n I_i$ and $n_0=n-n_1$ are the counts of relevant and background entities respectively, $B_{\tau_k,I}=\prod_{i \in U_R} \tau_k^{1\mid 0}(i)$,  $C_{\gamma_k,n_1}=\sum_{t=1}^{n_0+1} t^{-\gamma_k}$ is the normalizing constant of the power-law distribution, $\mathcal{A}_{U_R}(\tau_k^{1\mid0})$ is the set of $\tau_k^1$'s that are compatible with  $\tau_k^{1\mid 0}$,  and $A_{\tau_k,I}=\#\{\mathcal{A}_{U_R}(\tau_k^{1\mid0})\}=\prod_{t=1}^{n_0+1}(n_{\tau_k,t}^{1\mid 0}!)$ with $n_{\tau_k,t}^{1\mid 0}= \sum_{i \in U_R} \mI(\tau_k^{1 \mid 0}(i) = t)$. 

Intuitively, this model assumes that each ranker first randomly places all background entities to generate $\tau_k^0$, then ``inserts" each relevant entity independently into the list of background entities according to a truncated power-law distribution to generate $\tau_k^{1\mid 0}$, and finally draws $\tau_k^1$ uniformly from $\mathcal{A}_{U_R}(\tau_k^{1\mid0})$. 
In other words, $\tau_k^0$ serves as a baseline for modeling $\tau_k^{1\mid0}$ and $\tau_k^{1}$. 
It is easy to see from the model that a more reliable ranker should possess a larger $\gamma_k$. With the assumption of independent rankers, we have the full-data likelihood:
\begin{eqnarray}
P(\tau_1,\cdots,\tau_m \mid I,\bgamma)&=&\prod_{k=1}^mP(\tau_k\mid I,\gamma_k)\nonumber\\
&=&[(n_0)!]^{-m}\times\prod_{k=1}^m\frac{\mI\big(\tau_k^1 \in \mathcal{A}_{U_R}(\tau_k^{1\mid0})\big)}{A_{\tau_k,I}\times(B_{\tau_k,I})^{\gamma_k}\times\big(C_{\gamma_k,n_1}\big)^{n_1}},
\end{eqnarray}
where $\bgamma=(\gamma_1,\cdots,\gamma_m)$.  A detailed Bayesian inference procedure for $(I,\bgamma)$ via Markov chain Monte Carlo can be found in \citet{deng2014bayesian}.

\subsection{The Mallows Model} \label{sec:mallows}
 \cite{mallows1957non} proposed the following probability model for a ranking list $\tau$ of $n$ entities:
\begin{equation}\label{Eq:MallowsModel}
\pi(\tau \mid \tau_0, \phi) = \dfrac{1}{Z_n(\phi)}\cdot\exp\{-\phi\cdot  d(\tau,\tau_0)\},
\end{equation}
where $\tau_0$ denotes the true ranking list, $\phi>0$ characterizing the reliability of $\tau$, function $d(\cdot,\cdot)$ is a distance metric between two ranking lists, and 
\begin{equation}\label{eq:Mallow-norm}
    Z_n(\phi)=\sum_{\tau'}\exp\{-\phi\cdot d(\tau', \tau_0)\}=\frac{\prod_{t=2}^n(1-e^{-t\phi})}{(1-e^{-\phi})^{n-1}}
\end{equation}
being the normalizing constant, whose analytic form was derived in  \citet{diaconis1988group}.
Clearly, a larger $\phi$ means that $\tau$ is more stable and concentrates in a tighter neighborhood of $\tau_0$. 
A common choice of $d(\cdot, \cdot)$ is the Kendall tau distance.

The Mallows model under the Kendall tau distance can also be equivalently described by an alternative multistage model, which selects and positions entities one by one in a sequential fashion, where $\phi$ serves as a common parameter that governs the probabilistic behavior of each entity in the stochastic process \citep{mallows1957non}.
Later on, \citet{Fligner1986Distance} extended the Mallows model by allowing $\phi$ to vary at different stages, i.e., introducing a position-specific parameter $\phi_i$ for each position $i$, which leads to a very flexible, in many cases too flexible, framework to model rank data. 
To stabilize the generalized Mallows model by \citet{Fligner1986Distance}, \citet{fan2019} proposed to put a structural constraint on $\phi_i$s of the form $\phi_i=\phi\cdot(1-\alpha^i)$ with $0<\phi <1$ and $0\leq \alpha \leq 1$.
As a probabilistic model for rank data, the Mallows model enjoys great interpretability, model compactness, inference and computation efficiency.
For a comprehensive review of the Mallows model and its extensions, see \citet{Irurozki2014PerMallows} and \citet{fan2019}.

\section{The Partition-Mallows Model} \label{sec:model}
The partition model employed by BARD \citep{deng2014bayesian} tends to oversimplify the problem for scenarios where we care about the detailed rankings of relevant entities. 
To further enhance the partition model of BARD so that it can reflect the detailed rankings of relevant entities, we describe a new partition-Mallows model in this section.

\subsection{The Reverse Partition Model}\label{subsec:RevPar}
To combine the partition model with the Mallows model, a naive strategy is to simply replace the uniform model for the relevant entities, i.e., $P(\tau_k^1 \mid \tau_k^{1|0},I)$ in (\ref{Eq:RankDecomposition}), by the Mallows model, which leads to the updated Equation \eqref{eqn:bardtau1} as below: 
$$P(\tau_k^1 \mid \tau_k^{1\mid 0},I) = \frac{\pi(\tau_k^1)} {Z_{\tau_k,I}} \times \mI \big(\tau_k^1 \in \mathcal{A}_{U_R}(\tau_k^{1\mid0}) \big),$$
where $\pi(\tau_k^1)$ is the Mallows density of $\tau_k^1$ and
$Z_{\tau_k,I}=\sum_{\tau\in \mathcal{A}_{U_R}(\tau_k^{1\mid0})} \pi(\tau)$
is the normalizing constant of the Mallows model with a constraint due to the compatibility of $\tau_k^1$ with respect to $\mathcal{A}_{U_R}(\tau_k^{1\mid0})$. 
Apparently, the calculation of $Z_{\tau_k,I}$, which involves the summation over the whole space of $\mathcal{A}_{U_R}(\tau_k^{1\mid0})$, whose size is $A_{\tau_k,I}=\#\{\mathcal{A}_{U_R}(\tau_k^{1\mid0})\}=\prod_{t=1}^{n_0+1}(n_{\tau_k,t}^{1\mid 0}!)$, is infeasible for most practical cases, rendering such a naive combination of the Mallows model and the partition model impractical. 

To avoid the challenging computation caused by the constraints due to $\mathcal{A}_{U_R}(\tau_k^{1\mid0})$, we rewrite the partition model by switching the roles of $\tau_k^0$ and $\tau_k^1$ in the model: instead of decomposing $\tau_k$ as $(\tau_k^0,\tau_k^{1\mid 0},\tau_k^1)$ conditioning on the group indicators $I$, we decompose $\tau_k$ into an alternative triplet $(\tau_k^1,\tau_k^{0\mid 1},\tau_k^0)$,  where $\tau_k^{0\mid 1}$ denotes the {\it relative reverse rankings} of background entities among the relevant ones.
Formally, we note that $\tau_k^{0\mid 1}(i) \triangleq n_1+2-\tau_{k|\{i\} \cup U_R}(i)$ for any $i\in U_R$, where $\tau_{k|\{i\} \cup U_R}(i)$ denotes the relative ranking of a background entity among the relevant ones. 
In this {\it reverse partition model}, we first order the relevant entities  according to a certain distribution and then use them as a reference system to ``insert'' the background entities. 
Figure~\ref{Tab:decomposition2} illustrates the equivalence between $\tau_k$ and its two alternative presentations, $(\tau_k^0,\tau_k^{1\mid 0},\tau_k^1)$ and $(\tau_k^1,\tau_k^{0\mid 1},\tau_k^0)$. 

Given the group indicator vector $I$, the reverse partition model based on $(\tau_k^1,\tau_k^{0\mid 1},\tau_k^0)$ gives rise to the following distributional form for $\tau_k$:
\begin{eqnarray}\label{Eq:RankDecompositionInverse}
P(\tau_k \mid I)=P(\tau_k^1,\tau_k^{0\mid 1},\tau_k^0 \mid I)=P(\tau_k^1 \mid I)\times P(\tau_k^{0\mid 1} \mid I)\times P(\tau_k^0 \mid \tau_k^{0\mid 1},I),
\end{eqnarray}
which is analogous to  (\ref{Eq:RankDecomposition}) for the original partition model in BARD. Comparing to (\ref{Eq:RankDecomposition}), however, the new form (\ref{Eq:RankDecompositionInverse}) enables us to specify an unconstrained marginal distribution for $\tau_k^1$.
Moreover, due to the symmetry between $\tau_k^{1\mid 0}$ and $\tau_k^{0\mid 1}$, it is highly likely that the power-law distribution, which was shown in \cite{deng2014bayesian} to approximate the distribution of $\tau_k^{1\mid 0}(i)$ well for each $E_i\in U_R$, can also model $\tau_k^{0\mid 1}(i)$ for each $E_i\in U_B$ reasonably well. Detailed numerical validations are shown in Supplementary Material.

If we assume that all relevant entities are exchangeable, all background entities are exchangeable, and the relative reserve ranking of a background entity among the relevant entities follows a power-law distribution, we have
\begin{eqnarray}
\label{eqn:InvBARD_tau1}
P(\tau_k^1 \mid I)&=&\frac{1}{n_1!},\\
\label{eqn:InvBARD_tau01}
P(\tau_k^{0\mid 1} \mid I,\gamma_k)&=&\prod_{i \in U_B} P(\tau_k^{0 \mid 1}(i)\mid I,\gamma_k)=\frac{1}{(B^*_{\tau_k,I})^{\gamma_k}\times (C^*_{\gamma_k,n_1})^{n_0}},\\
\label{eqn:InvBARD_tau0}
P(\tau_k^0 \mid \tau_k^{0\mid 1},I)&=&\frac{1}{A^*_{\tau_k,I}}\times\mI\big(\tau_k^0 \in \mathcal{A}_{U_R}(\tau_k^{0\mid1})\big),
\end{eqnarray}
where $n_1$ and $n_0$ are numbers of relevant and background entities, respectively, $B^*_{\tau_k,I}=\prod_{i \in U_B} \tau_k^{0 \mid 1}(i)$ is the unnormalized part of the power-law, $C^*_{\gamma_k,n_1}=\sum_{t=1}^{n_1+1} t^{-\gamma_k}$ is the normalizing constant, 
$\mathcal{A}_{U_B}(\tau_k^{0\mid1})$ is the set of all $\tau_k^0$ that are compatible with a given $\tau_k^{0\mid 1}$, and $A^*_{\tau_k,I}=\#\{\mathcal{A}_{U_B}(\tau_k^{0\mid1})\}=\prod_{t=1}^{n_1+1}(n_{\tau_k,t}^{0\mid 1}!)$ with $n_{\tau_k,t}^{0\mid 1}= \sum_{i \in U_B} \mI(\tau_k^{0\mid 1}(i) = t)$. Apparently, the likelihood of this reverse-partition model shares the same structure as that of the original partition model in BARD, and thus can be inferred in a similar way.

\subsection{The Partition-Mallows Model} \label{subsec:BMM}
The reverse partition model introduced in section ~\ref{subsec:RevPar} allows us to freely  model $\tau_k^1$  beyond a uniform distribution, which is infeasible for the original partition model in BARD. 
Here we employ the Mallows model for $\tau_k^1$ due to its interpretability, compactness and computability. 
To achieve this, we replace the group indicator vector $I$ in the partition model by a more general indicator vector $\cI=\{\cI_i\}_{i=1}^n$, which takes value in $\Omega_\cI$, the space of all permutations of $\{1,\cdots,n_1,\underbrace{0,\ldots,0}_{n_0}\}$, with $\cI_i=0$ if $E_i\in U_B$, and $\cI_i=k>0$ if $E_i\in U_R$ and is ranked at position $k$ among all relevant entities in $U_R$.
Figure~\ref{Tab:decomposition2} provides an illustrative example of assigning an enhanced indicator vector $\cI$ to a universe of 10 entities with $n_1=5$. 
 
Based on the status of $\cI$, we can define subvectors $\cI^+$ and $\cI^0$, where $\cI^+$ stands for the subvector of $\cI$ containing all positive elements in $\cI$, and $\cI^0$ for the remaining zero elements in $\cI$.
Figure~\ref{Tab:decomposition2} demonstrates the constructions of $\cI$, $\cI^+$ and $\cI^0$, and the equivalence between $\tau_k$, $(\tau_k^0,\tau_k^{1\mid 0},\tau_k^1)$, and $(\tau_k^1,\tau_k^{0\mid 1},\tau_k^0)$ given $\cI$. 
Note that different from the partition model in BARD, in which we allow the number of relevant entities represented by $n_1$ to vary nearby its expected value, the number of relevant entities in the new model, is assumed to be fixed and known for conceptual and computational convenience.
In other words, we have $|U_R|=n_1$ in the new setting.

\begin{figure}[h]
    \centering
    \begin{tabular}{
    p{0.5cm}<{\centering}|p{0.5cm}<{\centering}|p{0.5cm}<{\centering}p{0.5cm}<{\centering}p{0.5cm}<{\centering}|p{0.5cm}<{\centering}p{0.5cm}<{\centering}p{0.5cm}<{\centering} p{0.5cm}<{\centering}p{0.5cm}<{\centering}|p{0.5cm}<{\centering}|p{0.5cm}<{\centering} p{0.5cm}<{\centering}p{0.5cm}<{\centering}|p{0.5cm}<{\centering}|p{0.5cm}<{\centering}}
\cline{1-3}  \cline{5-6} \cline{8-8}\cline{10-12} \cline{14-16}
         $\cI^+$ &$\cI^0$ &$I$& &$\cI$& $U$&  &$\tau_k$ &  &$\tau_k^1$&$\tau_k^{0\mid 1}$&$\tau_k^0$ & & $\tau_k^0$&$\tau_k^{1\mid 0}$&$\tau_k^1$  \\ \cline{1-3} \cline{5-6}  \cline{8-8}\cline{10-12} \cline{14-16}
          1 & - & 1 & & 1& $E_1$    &  & 2 & & 2& - & - &&-&1&2\\
          2 & - & 1 & & 2& $E_2$    &  & 6 & & 4& - & - &&-&3&4\\
          3 & - & 1 & & 3& $E_3$    &  & 4 & & 3& - & -&&-&2&3\\
          4 & - & 1 & & 4& $E_4$    &  & 1 & & 1& - & -&&-&1&1\\
          5 & - & 1 &$\Longleftarrow$& 5& $E_5$    &  & 7 & $\Longleftrightarrow$& 5&-& -&$\Longleftrightarrow$&-&3&5\\
          - & 0 & 0 & & 0& $E_6$    &  & 5 & & -& 3 & 2&&2&-&-\\
          - & 0 & 0 & & 0& $E_7$    &  & 3 & & -& 4 & 1&&1&-&-\\
          - & 0 & 0 & & 0& $E_8$    &  & 8 & & -& 1 & 3&&3&-&-\\
          - & 0 & 0 & & 0& $E_9$    &  & 9 & & -& 1 & 4&&4&-&-\\
          - & 0 & 0 & & 0& $E_{10}$ &  & 10& & -& 1 & 5&&5&-&-\\  \cline{1-3} \cline{5-6}  \cline{8-8}\cline{10-12} \cline{14-16}
    \end{tabular}
    \caption{An illustrative example of construction of $\cI^{+}$, $\cI^0$ and $I$ based on the enhanced indicator vector $\cI$ of $n_1=5$ to a universe of 10 entities, and the decomposition of a ranking list $\tau_k$ into triplet $(\tau_k^1,\tau_k^{0\mid 1},\tau_k^0)$ and $(\tau_k^0,\tau_k^{1\mid 0},\tau_k^1)$ respectively given $\cI$.
    }
    \label{Tab:decomposition2}
\end{figure}

As an analogy of Equations (\ref{Eq:RankDecomposition}) and (\ref{Eq:RankDecompositionInverse}), we have the following decomposition of $\tau_k$ given the enhanced indicator vector $\cI$: 
\begin{eqnarray}\label{Eq:RankDecomposition_BARDM}
P(\tau_k \mid \cI)=P(\tau_k^1,\tau_k^{0\mid 1},\tau_k^0 \mid \cI)=P(\tau_k^1 \mid\cI)\times P(\tau_k^{0\mid 1} \mid\cI)\times P(\tau_k^0 \mid \tau_k^{0\mid 1},\cI).
\end{eqnarray}
Assume that $\tau_k^1\mid\cI$ follows the Mallows model (with parameter $\phi_k$) centered at $\cI^+$:
\begin{eqnarray}\label{Eq:tau1_BARDM}
P(\tau_k^1 \mid \cI, \phi_k)=P(\tau_k^1 \mid\cI^+,\phi_k)=\frac{\exp\{-\phi_k\cdot d_{\tau}(\tau_k^1,\cI^+)\}}
{Z_{n_1}(\phi_k)},
\end{eqnarray}
where $d_{\tau}(\cdot,\cdot)$ denotes Kendall tau distance and
$Z_{n_1}(\phi_k)$ is defined as in \eqref{eq:Mallow-norm}.
Clearly, a larger $\phi_k$ indicates that ranker $\tau_k$ is of a higher quality, as the distribution is more concentrated at the ``true ranking" defined by $\cI^+$. 
Since the relative ranking of background entities are of no interest to us, we still assume that they are randomly ranked. 
Together with the power-law assumption for $\tau_k^{0\mid1}(i)$, we have
\begin{eqnarray}
\label{Eq:tau01_BARDM}
P(\tau_k^{0\mid 1} \mid\cI)&=&P(\tau_k^{0\mid 1} \mid I,\gamma_k)
=\frac{1}{(B^*_{\gamma_k,I})^{\gamma_k}\times (C^*_{\gamma_k,n_1})^{n-n_1}},\\
\label{Eq:tau0_BARDM}
P(\tau_k^0 \mid \tau_k^{0\mid 1},\cI)&=&P(\tau_k^0 \mid \tau_k^{0\mid 1},I)=\frac{1}{A^*_{\tau_k,I}}\times\mI\big(\tau_k^0 \in \mathcal{A}_{U_R}(\tau_k^{0\mid1})\big),
\end{eqnarray}
where notations $A^*_{\tau_k,I}$, $B^*_{\tau_k,I}$ and $C^*_{\gamma_k,n_1}$ are the same as in the reverse-partition model. 
We call the resulting model the \emph{Partition-Mallows model}, abbreviated as PAMA.

Different from the partition and reverse partition models, which quantify the quality of ranker $\tau_k$ with only one parameter $\gamma_k$ in the power-law distribution, the PAMA model contains two quality parameters $\phi_k$ and $\gamma_k$, with the former indicating the ranker's ability of ranking relevant entities and the latter reflecting the ranker's ability in differentiating relevant entities from background ones. 
Intuitively, $\phi_k$ and $\gamma_k$ reflect the quality of ranker $\tau_k$ in two different aspects. 
However, considering that a good ranker is typically strong in both dimensions, it looks quite natural to further simplify the model by assuming
\begin{equation}\label{Eq:phik2phi}
\phi_k=\phi\cdot\gamma_k,
\end{equation}
with $\phi>0$ being a common factor for all rankers. 
This assumption, while reducing the number of free parameters by almost half, captures the natural positive correlation between $\phi_k$ and $\gamma_k$ and serves as a first-order (i.e., linear) approximation to the functional relationship between $\phi_k$ and $\gamma_k$. 
A wide range of numerical studies based on simulated data suggest that the linear approximation showed in \eqref{Eq:phik2phi} works reasonably well for many typical scenarios for rank aggregation. In contrast, the more flexible model with both $\phi_k$ and $\gamma_k$ as free parameters (which is referred to as PAMA$^*$) suffers from unstable performance from time to time.
Detailed evidences to support assumption \eqref{Eq:phik2phi} can be found in Supplementary Material.

Plugging in  \eqref{Eq:phik2phi} into \eqref{Eq:tau1_BARDM}, we have a simplified model for $\tau_k^1$ given $\cI$ as follows:
\begin{eqnarray}\label{Eq:tau1_BARDM_final}
P(\tau_k^1 \mid\cI, \phi,\gamma_k)=P(\tau_k^1 \mid\cI^+,\phi,\gamma_k)=\frac{\exp\{-\phi\cdot\gamma_k\cdot d_{\tau}(\tau_k^1,\cI^+)\}}{Z_{n_1}(\phi\cdot\gamma_k)}.
\end{eqnarray}
Combining \eqref{Eq:tau01_BARDM}, \eqref{Eq:tau0_BARDM} and \eqref{Eq:tau1_BARDM_final}, we get the full likelihood of $\tau_k$:
\begin{eqnarray}\label{Eq:tau_BARDM_final}
P(\tau_k \mid\cI, \phi,\gamma_k)&=&P(\tau_k^1\mid\cI,\phi,\gamma_k)\times P(\tau_k^{0|1}\mid\cI,\gamma_k)\times P(\tau_k^{0}\mid \tau_k^{0|1},\cI)\nonumber\\
&=&\frac{\mI\big(\tau_k^0 \in \mathcal{A}_{U_R}(\tau_k^{0\mid1})\big)}{A^*_{\tau_k,I}\times(B^*_{\tau_k,I})^{\gamma_k}\times(C^*_{\gamma_k,n_1})^{n-n_1}\times (D^*_{\tau_k,\cI})^{\phi\cdot\gamma_k}\times E^*_{\phi,\gamma_k}},
\end{eqnarray}
where $D^*_{\tau_k,\cI}=\exp\{d_{\tau}(\tau_k^1,\cI^+)\}$, $E^*_{\phi,\gamma_k}=Z_{n_1}(\phi\cdot\gamma_k)=\frac{\prod_{t=2}^{n_1}(1-e^{-t\phi\gamma_k})}{(1-e^{-\phi\gamma_k})^{n_1-1}}$, and $A^*_{\tau_k,\cI}$, $B^*_{\tau_k,\cI}$ and $C^*_{\tau_k,n_1}$ keep the same meaning as in the reverse partition model. At last, for the set of observed ranking lists $\btau=(\tau_1,\cdots,\tau_m)$ from $m$ independent rankers, we have the joint likelihood:
\begin{eqnarray}
\label{eqn:like}
P(\btau \mid\cI,\phi,\bgamma)&=&\prod_{k=1}^m P(\tau_k\mid\cI,\phi,\gamma_k).
\end{eqnarray}

\subsection{Model Identifiability and Estimation Consistency}\label{sec:consistency}
Let $\Omega_n$ be the space of all permutations of $\{1,\cdots,n\}$ in which $\tau_k$ takes value, and let $\btheta=(\cI,\phi,\bgamma)$ be the vector of model parameters. The PAMA model  in \eqref{eqn:like},
i.e., $P(\btau \mid \btheta)$, defines a family of probability distributions on $\Omega_{n}^m$ indexed by parameter $\btheta$ taking values in space $\bTheta=\Omega_{\cI} \times \Omega_{\phi} \times \Omega_{\bgamma} $, 
where $\Omega_{\cI}$ is the space of all permutations of $\{1,\cdots,n_1,{\bf 0}_{n_0}\}$, $\Omega_{\phi}=(0,+\infty)$ and $\Omega_{\bgamma}=[0,+\infty)^m$.
We show here that the PAMA model defined in  \eqref{eqn:like} is identifiable and the model parameters can be estimated consistently under mild conditions.

\begin{Thm} \label{thm:identi}
The PAMA model is identifiable, i.e., 
\begin{equation}\label{eq:IdentifiablityCondition}
 \forall\ \btheta_1,\btheta_2\in\bTheta,\ \mbox{if}\  P(\btau\mid\btheta_1)= P(\btau\mid\btheta_2)\ \mbox{for}\ \forall\ \btau\in\Omega_n^m,\ \mbox{then}\ \btheta_1 = \btheta_2.
\end{equation}
\end{Thm}
\begin{proof}
See Supplementary Material.
\end{proof}

To show that parameters in the PAMA model can be estimated consistently, we will first construct a consistent estimator for the indicator vector $\cI$ as $m\rightarrow \infty$ but with the number of ranked entities $n$ fixed, and show later that $\phi$ can also be consistently estimated once $\cI$ is given.
To this end, we define
$\bar\tau(i)=m^{-1}\sum_{k=1}^{m}\tau_k(i)$
to be the average rank of entity $E_i$ across all $m$ rankers, and assume that the ranker-specific quality parameters $\gamma_1,\cdots,\gamma_m$ are i.i.d. samples from a non-atomic probability measure $F(\gamma)$ defined on $[0,\infty)$ with a finite first moment (referred to as condition $\bC_\gamma$ hereinafter).
Then, by the strong law of large numbers we have
\begin{equation}\label{eq:MeanRank}
\bar\tau(i)=\frac{1}{m}\sum_{k=1}^{m}\tau_k(i)\rightarrow \bbE\big[\tau(i)\big]
\ a.s.\ \ \mbox{with} \ m\rightarrow\ \infty,
\end{equation}
since $\{\tau_k(i)\}_{k=1}^m$ are i.i.d. random variables with expectation
$$\bbE\big[\tau(i)\big]=\bbE\Big[\bbE\big[\tau(i)\mid\gamma\big]\Big]=\int\bbE\big[\tau(i)\mid\gamma\big]dF(\gamma),$$
where $\bbE\big[\tau(i)\mid\gamma\big]$ is the conditional mean of $\tau(i)$ given the model parameters $(\cI,\phi,\gamma)$, i.e.,
$$\bbE\big[\tau(i)\mid\gamma\big]=\sum_{t=1}^n t\cdot P\big(\tau(i)=t\mid\cI,\phi,\gamma\big).$$
Clearly, $\bbE\big[\tau(i)\big]$ is a function of $\phi$ given $\cI$ and $F(\gamma)$. We define $e_i(\phi)\triangleq \bbE\big[\tau(i)\big]$ to emphasize $\bbE\big[\tau(i)\big]$'s nature as a continuous function of $\phi$.
Without loss of generality, we suppose that $U_R=\{1,\cdots,n_1\}$ and $U_B=\{n_1+1,\cdots,n\}$, i.e., $\cI=(1,\cdots,n_1,0,\cdots,0)$, hereinafter.  
Then, the partition structure and the Mallows model embedded in the PAMA model lead to the following facts:
\begin{equation}\label{eq:MeanRankRelation}
e_1(\phi) < \cdots < e_{n_1}(\phi)\ \mbox{and}\ e_{n_1+1}(\phi)=\cdots=e_{n}(\phi)=e_0,\ \forall\ \phi\in\Omega_\phi.
\end{equation}
Note that $e_i(\phi)$ degenerates to a constant with respect to $\phi$ (i.e., $e_0$) for all $i>n_1$ because parameter $\phi$ influences only the relative rankings of relevant entities in the Mallows model. The value of $e_0$ is completely determined by $F(\gamma)$. For the BARD model, it is easy to see that
$e_1=\cdots=e_{n_1} \leq e_{n_1+1}=\cdots=e_n.$

\begin{figure}[h]
    \centering
    \includegraphics[width=0.9\textwidth]{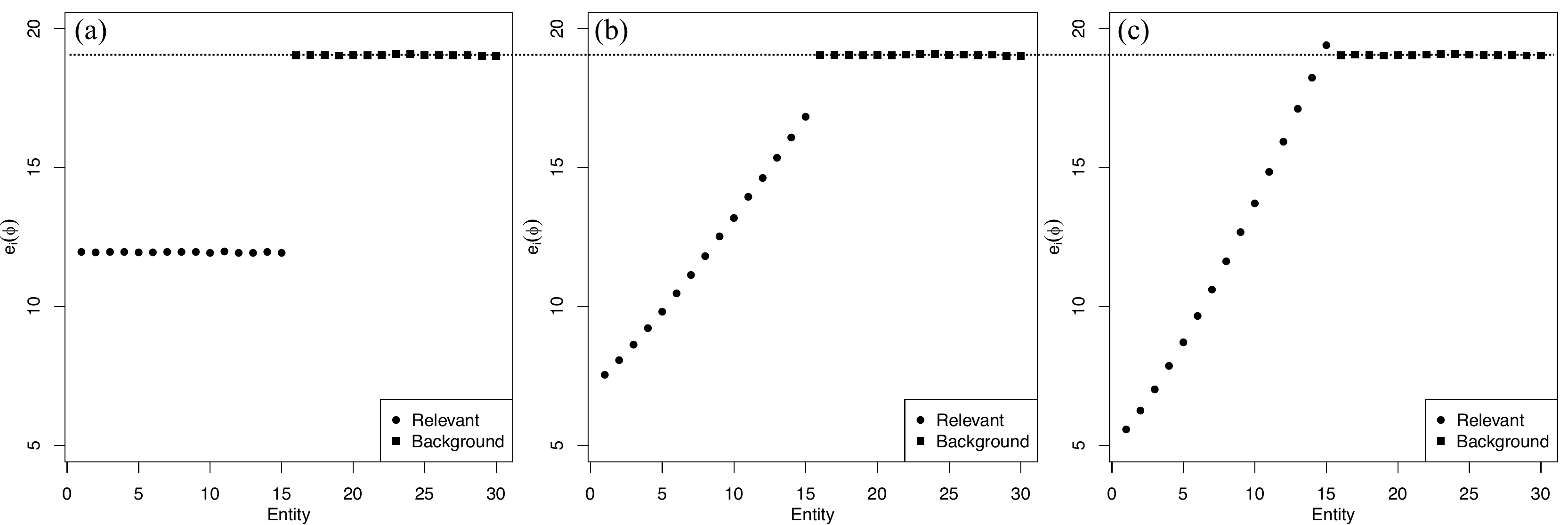}
    \caption{Average ranks of all the entities with fixed $\cI = (1,\cdots,n_1,0,\cdots,0)$, $n=30$, $n_1=15$, $m=100000$ and $F(\gamma)= U(0,2)$. Figures (a), (b) and (c) are the corresponding results for $\phi =0, 0.2\ \mbox{and } 0.4$ respectively.}
   \label{fig:AverageRankOfEntitiesInPAMA}
\end{figure}

Figure \ref{fig:AverageRankOfEntitiesInPAMA} shows some empirical estimates of the $e_i(\phi)$'s based on $m=100,000$ independent samples drawn from PAMA models with $n=30$, $n_1=15$, and $F(\gamma)= U(0,2)$, but three different $\phi$ values: (a) $\phi=0$, which corresponds to the BARD model; (b) $\phi=0.2$; and (c) $\phi=0.5$. One surprise is that in case (c), some relevant entities may have a larger $e_i(\phi)$ (worse ranking) than the average rank of background entities.
Lemma \ref{lem:AverageMean} guarantees that for almost all $\phi\in\Omega_\phi$, $e_0$ is different from $e_i(\phi)$ for $i=1,\cdots,n_1$. The proof of Lemma \ref{lem:AverageMean} can be found in Supplementary Material.

\begin{lemma}\label{lem:AverageMean}
For the PAMA model with condition $\bC_\gamma$, $\exists\ \tilde\Omega_\phi\subset\Omega_\phi$, s.t. $(\Omega_\phi-\tilde\Omega_\phi)$ contains only finite elements and
\begin{equation}\label{eq:e0neqei}
    e_i(\phi)\neq e_0\ \mbox{for}\ i=1,\cdots,n_1,\ \forall\ \phi\in\tilde\Omega_\phi.
\end{equation}
\end{lemma}

The facts demonstrated in \eqref{eq:MeanRankRelation} and \eqref{eq:e0neqei} suggest the following  three-step strategy to estimate $\cI$: (a) find the subset $S_0$ of $n_0=(n-n_1)$ entities from $U$ so that the within-subset variation of the $\bar\tau(i)$'s is the smallest, i.e., 
\begin{equation}
    S_0=\argmin_{S\in U, \ |S|=n_0}  \sum_{i\in S} (e_i-\bar{e}_S)^2 , \ \ \mbox{with} \ \bar{e}_S=n_0^{-1} \sum_{i\in S} e_i,
\end{equation} 
and let $\tilde{U}_B=S_0$ be an estimate of $U_B$; (b) rank the entities in $U\setminus S_0$ by $\bar\tau(i)$ increasingly and use the obtained ranking $\tilde\cI^+$ as an estimate of $\cI^+$; (c) combine the above two steps to obtain the estimate $\tilde\cI$ of $\cI$. This can be achieved  by defining 
$\tilde U_R=U \setminus \tilde U_B$ and $\tilde\cI^+=rank(\{\bar{\tau}(i) :i\in\tilde U_R\}),$ and obtain $\tilde\cI=(\tilde\cI_1,\cdots,\tilde\cI_n)$, with $\tilde\cI_i=\tilde\cI^+_i\cdot\bbI(i\in\tilde U_R).$

Note that $\tilde U_B$ is based on the mean ranks, $\{\bar\tau(i)\}_{i\in U}$, thus is clearly a moment estimator. Although this three-step estimation strategy is neither statistically efficient nor computationally feasible (step (a) is NP-hard), it nevertheless serves as a prototype for developing the consistency theory.
Theorem \ref{thm:consisI} guarantees that $\tilde\cI$ is a consistent estimator of $\cI$ under mild conditions.

\begin{Thm} \label{thm:consisI}
For the PAMA model with condition $\bC_\gamma$,
for almost all $\phi\in\Omega_\phi$, the moment estimator $\tilde\cI$ converges to $\cI$ with probability 1 with $m$ going to infinity.
\end{Thm}

\begin{proof}
Combining fact \eqref{eq:e0neqei} in Lemma \ref{lem:AverageMean} with fact \eqref{eq:MeanRank},
we have for $\forall\ \phi\in\tilde\Omega_\phi$ that
$$e_1(\phi)<\cdots<e_{n_1}(\phi)\ \mbox{and}\ e_i(\phi)\neq e_0\ \mbox{for}\ i=1,\cdots,n_1.$$ 
Moreover, as fact \eqref{eq:MeanRank} tells us that for $\forall\ \epsilon,\delta>0$, $\exists\ M>0$ s.t. for $\forall\ m>M$,
$$P\big(|\bar\tau(i)-e_i(\phi)|<\delta\big)\geq 1-\epsilon,\ i=1,\cdots,n,$$
it is straightforward to see the conclusion of the theorem.
\end{proof}

Theorem \ref{thm:consisI} tells us that estimating $\cI$ is straightforward if the number of independent rankers $m$ goes to infinity: a simple moment method ignoring the quality difference of rankers can provide us a consistent estimate of $\cI$. In a practical problem where only a finite number of rankers are involved, however, more efficient statistical inference of the PAMA model based on Bayesian or frequentist principles becomes more attractive as effectively utilizing the quality information of different rankers is critical.

With $n_0$ and $n_1$ fixed, parameter $\gamma_k$, which governs the power-law distribution for the rank list $\tau_k$, cannot be estimated consistently.
Thus, its distribution $F(\gamma)$ cannot be determined nonparametrically even when the number of rank lists $m$ goes to infinity. We impose a parametric form $F_\psi(\gamma)$ with $\psi$ as the hyper-parameter and
refer to the resulting hierarchical-structured model
as PAMA-H, which has the following marginal likelihood of $(\phi,\psi)$ given $\cI$:
$$L(\phi,\psi\mid\cI)=\int P(\btau\mid\cI,\phi,\bgamma)dF_\psi(\bgamma)
=\prod_{k=1}^m\int P(\tau_k\mid\cI,\phi,\gamma_k)dF_\psi(\gamma_k)=\prod_{k=1}^mL_k(\phi,\psi\mid\cI).$$
We show in Theorem \ref{thm:consisPhiPsi} that the MLE based on the above marginal likelihood is consistent.

\begin{Thm}\label{thm:consisPhiPsi}
Under the PAMA-H model, assume that  $(\phi,\psi)$ belongs to the parameter space $\Omega_\phi\times\Omega_\psi$, and
the true parameter $(\phi_0,\psi_0)$ is an interior point of $\Omega_\phi\times\Omega_\psi$. Let $(\hat{\phi}_\cI,\hat\psi_\cI)$ be the maximizer of $L(\phi,\psi\mid\cI)$.
If  $F_\psi(\gamma)$ has a density function $f_\psi(\gamma)$ that is differentiable and concave with respect to $\psi$, 
then $\lim_{m\rightarrow\infty}(\hat{\phi}_\cI,\hat\psi_\cI)=(\phi_0,\psi_0)$ almost surely.
\end{Thm}

\begin{proof}
See Supplementary Material.
\end{proof}

\section{Inference with the Partition-Mallows Model} \label{sec:statinfer}

\subsection{Maximum Likelihood Estimation}
\label{subsec:MLE}
Under the PAMA model, the MLE of $\btheta= (\cI, \phi,\bgamma)$ is $\hat{\btheta}= \arg\max_{\btheta} l(\btheta)$, where 
\begin{equation} \label{eqn:loglike}
    l(\btheta) = \log P(\tau_1,\tau_2,\cdots,\tau_m\mid\btheta)
\end{equation}
is the logarithm of the likelihood function  (\ref{eqn:like}).
Here, we adopt the \textit{Gauss-Seidel} iterative method in \cite{YANG2018281}, also known as \textit{backfitting} or \textit{cyclic coordinate ascent}, to implement the optimization. Starting from an initial point $\btheta^{(0)}$, the Gauss-Seidel method iteratively updates one coordinate of $\btheta$ at each step with the other coordinates held fixed at their current values.
A Newton-like method is adopted to update $\phi$ and $\gamma_k$. Since $\cI$ is a discrete vector, we find favorable values of $\cI$ by swapping two neighboring entities to check whether $g(\cI\mid\bgamma^{(s+1)}, \phi^{(s+1)})$ increases. More details of the algorithm are provided in Supplementary Material.

With the MLE $\hat\btheta=(\hat\cI,\hat\phi,\hat\bgamma)$, we define $U_R(\hat\cI)=\{i\in U:\ \hat\cI_i>0\}$ and $U_B(\hat\cI)=\{i\in U:\ \hat\cI_i=0\}$, and generate the final aggregated ranking list $\hat\tau$ based on the rules below: (a) set the top-$n_1$ list of $\hat\tau$ as $\hat\tau_{n_1}=sort(i\in U_R(\hat\cI)\ by\ \hat\cI_i\uparrow)$, (b) all entities in $U_B(\hat\cI)$ tie for positions behind. Hereinafter, we refer to this MLE-based rank aggregation procedure under PAMA model as PAMA$_F$.

For the PAMA-H model, a similar procedure can be applied to find the MLE of $\btheta=(\cI,\phi,\psi)$, with the $\bgamma=(\gamma_1,\cdots,\gamma_m)$ being treated as the missing data.
With the MLE $\hat\btheta=(\hat\cI,\hat\phi,\hat\psi)$, we can generate the final aggregated ranking list $\hat\tau$ based on $\hat\cI$ in the same way as in PAMA, and evaluate the quality of ranker $\tau_k$ via the mean or mode of the conditional distribution below:
$$f(\gamma_k\mid\tau_k;\hat\cI,\hat\phi,\hat\psi)\propto f(\gamma_k\mid\hat\psi)\cdot P(\tau_k\mid\hat\cI,\hat\phi,\gamma_k).$$
In this paper, we refer to the above MLE-based rank aggregation procedure under PAMA-H model as PAMA$_{HF}$.
The procedure is detailed in Supplementary Material.

\subsection{Bayesian Inference} \label{subsec:BC}
Since the three model parameters $\cI$, $\phi$ and $\bgamma$ encode ``orthogonal" information of the PAMA model, it is natural to expect that $\cI$, $\phi$ and $\bgamma$ are mutually independent {\it a priori}. We thus specify their joint prior distribution as
$$\pi(\cI,\phi,\bgamma)=\pi(\cI)\cdot\pi(\phi)\cdot\prod_{k=1}^m\pi(\gamma_k).$$
Without much loss, we may restrict the range of $\phi$ and $\gamma_k$'s to a closed interval $[0,b]$ with a large enough $b$. In contrast, $\cI$ is discrete and takes value in the space $\Omega_\cI$ of all permutations of $\{1,\ldots,n_1, \underbrace{0,\ldots,0}_{n_0}\}$.
It is convenient to specify $\pi(\cI)$, $\pi(\phi)$ and $\pi(\gamma_k)$ as uniform, i.e.,
$$\pi(\cI)\sim U(\Omega_\cI),\ \pi(\phi)\sim U[0,b],\ \pi(\gamma_k)\sim U[0,b].$$
Based on our experiences in a large range of simulation studies and real data applications, we find that it works reasonably well to set $b=10$. In Section~\ref{sec:consistency} we also discussed letting $\pi(\gamma_k)$ be of a parametric form, which will be further discussed later.

The posterior distribution can be expressed as
\begin{eqnarray}
&&f(\cI,\phi,\bgamma|\tau_1,\tau_2,\cdots,\tau_m)\nonumber\\
&\propto& \pi(\cI,\phi,\bgamma)\cdot P(\tau_1,\tau_2,\cdots,\tau_m|\cI,\phi,\bgamma)\nonumber\\
&=&\mI\big(\phi\in[0,10])\big)\times\prod_{k=1}^m \Big\{ \frac{\mI\big(\tau_k^0 \in \mathcal{A}_{U_R}(\tau_k^{0\mid1})\big)\times\mI\big(\gamma_k\in[0,10]\big)}{A^*_{\tau_k,I}\times(B^*_{\tau_k,I})^{\gamma_k}\times(C^*_{\gamma_k,n_1})^{n-n_1}\times (D^*_{\tau_k,\cI})^{\phi\cdot\gamma_k}\times E^*_{\phi,\gamma_k}} \Big\},\label{eqn:posterior}
\end{eqnarray}
with the following conditional distributions:
\begin{eqnarray} 
\label{fc:I}
f(\cI\mid\phi,\bgamma) &\propto& \prod_{k=1}^m\frac{\mI\big(\tau_k^0 \in \mathcal{A}_{U_R}(\tau_k^{0\mid1})\big)}{A^*_{\tau_k,I}\times(B^*_{\tau_k,I})^{\gamma_k}\times(D^*_{\tau_k,\cI})^{\phi\cdot\gamma_k}},\\
\label{fc:phi}
f(\phi\mid\cI,\bgamma) &\propto& {\mathbb I}\big(\phi \in [0,10]\big)\times \prod_{k=1}^m \frac{1}{(D^*_{\tau_k,\cI})^{\phi\cdot\gamma_k}\times E^*_{\phi,\gamma_k}},\\
\label{fc:gamma}
f(\gamma_k\mid\cI,\phi,\bgamma_{[-k]})&\propto&\frac{\mI\big(\gamma_k \in [0,10]\big)}{(B^*_{\tau_k,I})^{\gamma_k}\times(C^*_{\gamma_k,n_1})^{n-n_1}\times (D^*_{\tau_k,\cI})^{\phi\cdot\gamma_k}\times E^*_{\phi,\gamma_k}},
\end{eqnarray}
based on which posterior samples of $(\cI,\phi,\bgamma)$ can be obtained by Gibbs sampling, where $\bgamma_{[-k]}=(\gamma_1,\cdots,\gamma_{k-1},\gamma_{k+1},\cdots,\gamma_m)$.

Considering that conditional distributions in (\ref{fc:I})-(\ref{fc:gamma}) are nonstandard, we adopt the Metropolis-Hastings algorithm \citep{Hastings1970} to enable the conditional sampling. To be specific, we choose the proposal distributions for $\phi$ and $\gamma_k$ as 
\begin{eqnarray*}
q(\phi\mid\phi^{(t)};\cI,\bgamma)&\sim&\mathcal{N}(\phi^{(t)},\sigma_{\phi}^2)\\
q(\gamma_k\mid\gamma_k^{(t)};\cI,\phi,\bgamma_{[-k]})&\sim& \mathcal{N}(\gamma_k^{(t)},\sigma_{\gamma_k}^2),
\end{eqnarray*}
where $\sigma_{\phi}^2$ and $\sigma_{\gamma_k}^2$ can be tuned to optimize the mixing rate of the sampler.
Since $\cI$ is a discrete vector, we propose new values of $\cI$ by swapping two randomly selected adjacent entities.
Note that the entity whose ranking is $n_1$ could be swapped with any background entity. 
Due to the homogeneity of background entities, there is no need to swap two background entities. 
Therefore, the number of potential proposals in each step is $\mathcal{O}(n n_1)$. More details about MCMC sampling techniques can be found in \cite{liu2008monte}.

Suppose that $M$ posterior samples $\{(\cI^{(t)},\phi^{(t)},\bgamma^{(t)})\}_{t=1}^M$ are obtained. We calculate the posterior means of different parameters as below:
\begin{eqnarray*}
\bar{\cI}_i&=&\frac{1}{M} \sum_{t=1}^M\Big[\cI_i^{(t)}\cdot I^{(t)}_i+\frac{n_1+1+n}{2}\cdot(1-I^{(t)}_i)\Big],\ i=1,\cdots,n,\\
\bar{\phi} &=& \frac{1}{M} \sum_{t=1}^M \phi^{(t)},\\
\bar{\gamma}_k &=& \frac{1}{M} \sum_{t=1}^M \gamma_k^{(t)}, k=1,\cdots,m.
\end{eqnarray*}
We quantify the quality of ranker $\tau_k$ with $\bar\gamma_k$, and generate the final aggregated ranking list $\hat\tau$ based on the $\bar\cI_i$s as following:
$$\hat\tau=sort(i\in U\ by\ \bar\cI_i \uparrow).$$
Hereinafter, we refer to this MCMC-based Bayesian rank aggregation procedure under the Partition-Mallows model as PAMA$_B$.

The Bayesian inference procedure PAMA$_{HB}$ for the PAMA-H model differs from PAMA$_B$ only by replacing the prior distribution $\prod_{k=1}^m \pi(\gamma_k)$, which is uniform in $[0,b]^m$,
with a hierarchically structured prior $\pi(\psi) \prod_{k=1}^m f_\psi (\gamma_k)$. The conditional distributions needed for Gibbs sampling are almost the same as \eqref{fc:I}-\eqref{fc:gamma}, except an additional one
\begin{eqnarray} 
\label{PAMA-HB:psi}
f(\psi\mid\cI,\phi,\bgamma)&\propto&\pi(\psi)\cdot\prod_{k=1}^mf_\psi(\gamma_k).
\end{eqnarray}
We may specify $f_\psi(\gamma)$ to be an exponential distribution and let $\pi(\psi)$ be a proper conjugate prior to make \eqref{PAMA-HB:psi} easy to sample from.
More details for PAMA$_{HB}$ with $f_\psi(\gamma)$ specified as an exponential distribution is provided in Supplementary Material.

Our simulation studies suggest that the practical performance of PAMA$_B$ and PAMA$_{HB}$ are very similar when $n_0$ and $n_1$ are reasonably large (see Supplementary Material for details).
In contrast, as we will show  in Section~\ref{sec:simulation}, the MLE-based estimates (e.g., PAMA$_F$) typically produce  less accurate results with a shorter computational time compared to  PAMA$_B$.

\subsection{Extension to Partial Ranking Lists}
The proposed Partition-Mallows model can be extended to more general scenarios where partial ranking lists, instead of full ranking lists, are involved in the aggregation. Given the entity set $U$ and a ranking list $\tau_S$ of entities in $S\subseteq U$, we say $\tau_S$ is a \emph{full ranking list} if $S=U$, and a \emph{partial ranking list} if $S\subset U$. 
Suppose $\tau_S$ is a partial ranking list and $\tau_U$ is a full ranking list of $U$. If the projection of $\tau_U$ on $S$ equals to $\tau_S$, we say $\tau_U$ is compatible with $\tau_S$, denotes as $\tau_U\sim\tau_S$. 
Let $\cA(\tau_S)=\{\tau_U:\tau_U\sim\tau_S\}$ be the set of all full lists that are compatible with $\tau_S$. Suppose a partial list $\tau_k$ is involved in the ranking aggregation problem. The probability of $\tau_k$ can be evaluated by:
\begin{eqnarray} \label{eqn:pllike}
P(\tau_k\mid\cI,\phi,\gamma_k)=\sum_{\tau_k^*\sim\tau_k} P(\tau_k^*\mid\cI,\phi,\gamma_k),
\end{eqnarray}
where $P(\tau_k^*\mid\cI,\phi,\gamma_k)$ is the probability of a compatible full list under the PAMA model. 
Clearly, the probability in (\ref{eqn:pllike}) does not have a closed-form representation due to complicated constraints between $\tau_k$ and $\tau_k^*$, and it is very challenging to do statistical inference directly based on this quantity. 
Fortunately, as rank aggregation with partial lists can be treated as a missing data problem, we can resolve the problem via standard methods for missing data inference.

The Bayesian inference can be accomplished by the classic data augmentation strategy~\citep{tanner1987} in a similar way as described in \cite{deng2014bayesian}, which iterates between imputing the missing data conditional on the observed data given the current parameter values, and updating parameter values by sampling from the posterior distribution based on the imputed full data. 
To be specific, we iteratively draw from the following two conditional distributions:
\begin{eqnarray*}
P(\tau_1^\ast,\cdots,\tau_m^\ast\mid\tau_1,\cdots,\tau_m;\cI,\phi,\bgamma)=\prod_{k=1}^m P(\tau_k^\ast\mid\tau_k;\cI,\phi,\gamma_k), \\
f(\cI,\phi,\bgamma\mid\tau_{1}^\ast,\cdots,\tau_m^\ast)\propto \pi(\cI)\times\pi(\bgamma)\times\pi(\phi)\times\prod_{k=1}^m P(\tau_k^\ast\mid\cI,\gamma_{k},\phi).
\end{eqnarray*}

To find the MLE of $\btheta$ for this more challenging scenario, we can use the Monte Carlo EM algorithm (MCEM, \cite{tanner1990}). Let $\tau_k^{(1)},\cdots, \tau_k^{(M)}$ be $M$ independent samples drawn from distribution $P(\tau_k^\ast\mid\tau_k,\cI,\phi,\gamma_k)$. The E-step involves the calculation of the $Q$-function below:
\begin{eqnarray*}
Q(\cI,\bgamma,\phi \mid \cI^{(s)},\bgamma^{(s)},\phi^{(s)}) &=& E \left\{\sum_{k=1}^m \log P(\tau_k^{*} \mid \cI,\bgamma,\phi) \mid \tau_k, \cI^{(s)},\bgamma_k^{(s)},\phi^{(s)} \right\}\nonumber\\
&\approx& \dfrac{1}{M}\sum_{k=1}^m \sum_{t=1}^M \log P(\tau_k^{(t)}\mid \cI,\bgamma_k,\phi).
\end{eqnarray*}
In the M-step, we use the \emph{Gauss-Seidel} method to maximize the above $Q$-function in a similar way as detailed in Supplementary Material.

No matter which method is used, a key step is to draw samples from
\[P(\tau_k^\ast\mid\tau_k;\cI,\phi,\gamma_k)\propto P(\tau_k^\ast\mid\cI,\gamma_k,\phi)\cdot \mI\big(\tau_k^\ast\in\cA(\tau_k)\big).\]
To achieve this goal, we start with $\tau_k^*$ obtained from the previous step of the data augmentation or MCEM algorithms, and conduct several iterations of the following Metropolis step with $P(\tau_k^\ast\mid\tau_k;\cI,\phi,\gamma_k)$ as its target distribution: (a) construct the proposal $\tau_k'$ by randomly selecting two elements in the current full list $\tau_k^*$ and swapping them; (b) accept or reject the proposal according to the Metropolis rule, that is to accept $\tau_k'$ with probability of $\min(1,\frac{P(\tau_k'\mid\cI,\gamma_k,\phi)}{P(\tau_k^{*}\mid\cI,\gamma_k,\phi)})$. Note that the proposed list $\tau_k'$ is automatically rejected if it is incompatible with the observed partial list $\tau_k$.

\subsection{Incorporating Covariates in the Analysis}
In some applications, covariate information for each ranked entity is available to assist rank aggregation.
One of the earliest attempts for incorporating such information in analysing rank data is perhaps the \emph{hidden score model} due to \cite{Thurstone1927}, which has become a standard approach and has many extensions. Briefly, these models assume that there is an unobserved score for each entity that is related to the entity-specific covariates $X_i=(X_{i1},\cdots,X_{ip})^T$ under a regression framework and the observed rankings are determined by these scores plus noises, i.e., 
$$\tau_k=sort(S_{ik}\downarrow,\ E_i\in U),\ \mbox{where}\ S_{ik}=X_i^T\bbeta+\varepsilon_{ik}.$$
Here, $\bbeta$ is the common regression coefficient and $\varepsilon_{ik}\sim N(0,\sigma^2_k)$ is the noise term. Recent progresses along this line are reviewed by \cite{Yu2000Bayesian,Bhowmik2017,li2020bayesian}.

Here, we propose to incorporate covariates into the analysis in a different way. Assuming that covariate $X_i$ provides information on the group assignment instead of the detailed ranking of entity $E_i$, we connect $X_i$ and $\cI_i$, the enhanced indicator of $X_i$, by a logistic regression model:
\begin{equation}\label{eq:BARDM_Logistics}
P(\cI_i\mid X_i)=P(I_i\mid X_i,\boldsymbol{\psi})=\dfrac{\exp\{X_i^T\boldsymbol{\psi}\cdot I_i\}}{1+\exp\{X_i^T\boldsymbol{\psi}\}},~~i=1,\cdots,n,
\end{equation}
where $\boldsymbol{\psi}=(\psi_1,\ldots,\psi_p)^T$ as the regression parameters. Let $\bX=(X_1,\cdots,X_n)$ be the covariate matrix, we can extend the Partition-Mallows model as
\begin{equation}\label{eq:BARDM_Covariate}
P(\tau_1,\cdots,\tau_m, \cI \mid \bX)=P(\cI \mid \bX,\boldsymbol{\psi})\times P(\tau_1,\cdots,\tau_m\mid\cI,\phi,\bgamma),
\end{equation}
where the first term
$$P(\cI \mid \bX,\boldsymbol{\psi})=\prod_{i=1}^n P(I_i \mid X_i,\boldsymbol{\psi})$$
comes from the logistic regression model (\ref{eq:BARDM_Logistics}), and the second term comes from the original Partition-Mallows model. In the extended model, our goal is to infer $(\cI,\phi,\bgamma,\boldsymbol{\psi})$ based on $(\tau_1,\cdots,\tau_m;\bX)$. We can achieve both Bayesian inference and MLE for the extended model in a similar way as described for the Partition-Mallows model. More details are provided in the Supplementary Material.

An alternative way to incorporate covariates is to replace the logistic regression model by a naive Bayes model, which models the conditional distribution of $\bX\mid\cI$ instead of $\cI\mid\bX$, as follows:
\begin{equation}
f(\tau_1,\cdots,\tau_m,\bX\mid\cI)=P(\tau_1,\cdots,\tau_m,\mid\cI,\phi,\bgamma)\times f(\bX\mid\cI),
\end{equation}
where
\begin{eqnarray*}
f(\bX\mid\cI)&=&\prod_{i=1}^nf(X_i\mid \cI_i)=\prod_{i=1}^nf(X_i\mid I_i)=\prod_{i=1}^n\prod_{j=1}^pf(X_{ij}\mid I_i)\\
&=&\prod_{i=1}^n\prod_{j=1}^p\Big\{\big[f_{j}(X_{ij}\mid\psi_{j0})\big]^{1-I_i}\cdot\big[f_{j}(X_{ij}\mid\psi_{j1})\big]^{I_i}\Big\},
\end{eqnarray*}
$f_j$ is pre-specified parametric distribution for covariates $X_j$ with parameter $\psi_{j0}$ for entities in the background group and $\psi_{j1}$ for entities in the relevant group. Since the performances of the two approaches are very similar, in the rest of this paper we use the logistic regression strategy to handle covariates due to its convenient form.

\section{Simulation Study} \label{sec:simulation}
\subsection{Simulation Settings}\label{sec:SimuSetting}
We simulated data from two models: (a) the proposed Partition-Mallows model, referred to as $\mathcal{S}_{PM}$, and (b) the Thurstone hidden score model, referred to as $\mathcal{S}_{HS}$. 
In the $\mathcal{S}_{PM}$ scenario, we specified the true indicator vector as $\cI=(1,\cdots,n_1,0,\cdots,0)$, indicating that the first $n_1$ entities $E_1,\cdots, E_{n_1}$ belong to $U_R$ and the rest belong to the background group $U_B$, and set 
$$\gamma_k=\left\{
   \begin{array}{ll}
    0.1, & \mbox{if } k\leq\frac{m}{2}; \\ 
    a+(k-\frac{m}{2})\times\delta_R, & \mbox{if } k>\frac{m}{2}. \\ 
  \end{array} 
  \right. 
  $$
Clearly, $a>0$ and $\delta_R>0$ control the quality difference and signal strength of the $m$ base rankers in the $\cS_{PM}$ scenario. 
We set $\phi=0.6$ (defined in  \eqref{Eq:phik2phi}), $\delta_R=\frac{2}{m}$, and $a$ with two options: 2.5 and 1.5. 
For easy reference, we denoted the strong signal case with $a=2.5$ and the weak signal case with $a=1.5$ by $\cS_{PM_1}$ and $\cS_{PM_2}$, respectively.

In the $\cS_{HS}$ scenario, we used the Thurstone model to generate the rank lists as $\tau_k = sort(i\in U\ by\ S_{ik} \downarrow),\ \mbox{where}\ S_{ik}\sim N(\mu_{ik},1)$ and
$$\mu_{ik}=\left\{
   \begin{array}{ll}
   0, & \mbox{if } k\leq\frac{m}{2}\ \mbox{or}\ i>n_1; \\ 
    a^*+\frac{b^*-a^*}{m}\times k + (n_1-i)\times\delta_E^*, & \mbox{otherwise}. \\ 
  \end{array}
  \right. $$
In this model, $a^*, b^*$ and $\delta_E^*$ (all positive numbers) control the quality difference and signal strength of the $m$ base rankers.
We also specified two sub-cases:  $\cS_{HS_1}$, the stronger-signal case with $(a^*,b^*,\delta_E^*)=(0.5,2.5,0.2)$; and $\cS_{HS_2}$, the weaker-signal case with $(a^*,b^*,\delta_E^*)=(-0.5,1.5,0.2)$.
Table \ref{Tab:HS1mu} shows the configuration matrix of $\mu_{ik}$ under $\cS_{HS_1}$ when $m=10, n=100$ and $n_1=10$. 
In both scenarios, the first half of rankers are completely non-informative, with the other half providing increasingly strong signals.
\begin{table}[h] \small
    \centering
    \begin{tabular}{c|c|c|c|c|c|c|c|c|c|c}
    \hline 
    & $\mu_1$&$\mu_2$&$\mu_3$&$\mu_4$&$\mu_5$&$\mu_6$&$\mu_7$&$\mu_8$&$\mu_9$&$\mu_{10}$\\  \hline
    $E_1$&0 &0&0&0&0& 3.7&3.9&4.1&4.3&4.5   \\ 
    $E_2$&0 &0&0&0&0  &3.5 &3.7&3.9&4.1&4.3  \\ 
    $E_3$&0 &0&0&0&0  & 3.3&3.5&3.7&3.9&4.1  \\ 
    $E_4$&0 &0&0&0&0  &3.1 &3.3&3.5&3.7&3.9  \\ 
    $E_5$&0 &0&0&0&0  & 2.9&3.1&3.3&3.5&3.7  \\
    $E_6$&0 &0&0&0&0  &2.7 &2.9&3.1&3.3&3.5  \\ 
    $E_7$&0 &0&0&0&0  &2.5 &2.7&2.9&3.1&3.3  \\ 
    $E_8$&0 &0&0&0&0  &2.3 &2.5&2.7&2.9&3.1  \\ 
    $E_9$&0 &0&0&0&0  & 2.1&2.3&2.5&2.7&2.9  \\ 
    $E_{10}$&0 &0&0&0&0  &1.9 &2.1&2.3&2.5&2.7  \\ 
    $E_{11}$&0 &0&0&0&0  &0 &0&0&0&0  \\ 
    $\vdots$&$\vdots$ &$\vdots$&$\vdots$&$\vdots$&$\vdots$  &$\vdots$ &$\vdots$&$\vdots$&$\vdots$&$\vdots$  \\ 
    $E_{100}$&0 &0&0&0&0  &0 &0&0&0&0  \\ \hline 
    \end{tabular}
    \caption{The configuration matrix of the $\mu_{ik}$'s under $\mathcal{S}_{HS_1}$ with $m$=10, $n$=100 and $n_1$=10.}
    \label{Tab:HS1mu}
\end{table}

For each of the four simulation scenarios (i.e., $\cS_{PM_1}$, $\cS_{PM_2}$, $\cS_{HS_1}$ and $\cS_{HS_2}$),
we fixed the true number of relevant entities $n_1=10$, but allowed the number of rankers $m$ and the total number of entities $n$ to vary, resulting in a total of 16 simulation settings (  $\{scenarios: \cS_{PM_1},\cS_{PM_2},\cS_{HS_1},  \cS_{HS_2}\}\times\{m:  10, 20\}\times\{n: 100, 300\}\times\{n_1:  10\}$). Under each setting, we simulated 500 independent data sets to evaluate and compare performances of different  rank aggregation methods.

\subsection{Methods in Comparison and Performance Measures}
Except for the proposed PAMA$_B$ and PAMA$_F$, we considered state-of-the-art methods in several classes, including the Markov chain-based methods MC$_1$, MC$_2$, MC$_3$ in \cite{Lin2010Space} and CEMC in \cite{2010Integration}, the partition-based method BARD in \cite{deng2014bayesian}, and the Mallows model-based methods MM and EMM in \cite{fan2019}.
Classic naive methods based on summary statistics were ignored because they have been shown in previous studies to perform suboptimally, especially in cases where base rankers are heterogeneous in quality. 
The Markov-chain-based methods, MM, and EMM were implemented in \textit{TopKLists}, \textit{PerMallows} and \textit{ExtMallows} packages in R (https://www.r-project.org/), respectively. The code of BARD was provided by its authors.

Let $\tau$ be the underlying true ranking list of all entities, $\tau_R=\{\tau(i):\ E_i\in U_R\}$ be the true relative ranking of relevant entities, $\hat\tau$ be the aggregated ranking obtained from a rank aggregation approach, $\hat\tau_R=\{\hat\tau(i):\ E_i\in U_R\}$ be the relative ranking of relevant entities after aggregation, and $\hat\tau_{n_1}$ be the top-$n_1$ list of $\hat\tau$.
After obtaining the aggregated ranking $\hat\tau$ from a rank aggregation approach, we evaluated its performance by two measurements, namely the \emph{recovery distance} $\kappa_{R}$ and the \textit{coverage} $\rho_R$, defined as below:
\begin{eqnarray*}
\kappa_{R}&\triangleq& d_{\tau}(\hat{\tau}_R,\tau_R) + n_{\hat{\tau}} \times \frac{n+n_1+1}{2},\\
\rho_R&\triangleq&\frac{n_1 -n_{\hat{\tau}} }{n_1},
\end{eqnarray*}
where $d_{\tau}(\hat{\tau}_R,\tau_R)$ denotes the Kendall tau distance between $\hat\tau_R$ and $\tau_R$, and $n_{\hat{\tau}}$ denotes the number of relevant entities who are classified as background entities in $\hat{\tau}$.
The recovery distance $\kappa_R$ considers  detailed rankings of all relevant entities plus mis-classification distances, while the coverage $\rho_R$ cares only about the identification of relevant entities without considering the detailed rankings. In the setting of PAMA, $\frac{n+n_1+1}{2}$ is the average rank of a background entity. The recovery distance increases if some relevant entities are mis-classified as background entities. Clearly, we expect a smaller $\kappa_R$ and a larger $\rho_R$ for a stronger aggregation approach.

\subsection{Simulation Results}
Table~\ref{Tab:recovery} summarizes the performances of the nine competing methods in the 16 different simulation settings, demonstrating the proposed PAMA$_B$ and PAMA$_F$ outperform all the other methods by a significant margin in most settings and PAMA$_B$ uniformly dominates PAMA$_F$. Figure~\ref{Fig:gamma} shows the quality parameter $\bgamma$ learned from the Partition-Mallows model in various simulation scenarios with $m=10$ and $n=100$, confirming that the proposed methods can effectively capture the quality difference among the rankers. The results of $\bgamma$ for other combinations of $(m,n)$ can be found in Supplementary Material which demonstrates consistent performance with Figure \ref{Fig:gamma}.

\begin{table}[htp] \scriptsize
    \centering
    \begin{tabular}{c|cc|ccc|cc|cccc} 
    \hline 
    \multicolumn{3}{c|}{Configuration}&\multicolumn{3}{c|}{Partition-type Models}&\multicolumn{2}{c|}{{Mallows Models}} &\multicolumn{4}{c}{MC-based Models} \\ \cline{1-3} \cline{4-6} \cline{7-8} \cline{9-12}
    $\mathcal{S}$& $n$ &$m$ &PAMA$_F$&PAMA$_B$ &BARD & EMM & MM& MC$_1$ & MC$_2$ & MC$_3$ & CEMC\\  \cline{1-3} \cline{4-6} \cline{7-8} \cline{9-12}
\multirow{8}{*}{$\mathcal{S}_{PM_1}$}&\multirow{2}{*}{100}&\multirow{2}{*}{10}& 24.5 & {\bf 15.2} & 57.1 & 51.7 & 103.2 & 338.4 & 163.1 & 198.6 & 197.8 \\ 
 &&&[0.95] & {\bf[0.97]} & [0.91] & [0.89] & [0.81] & [0.36] & [0.69] & [0.63] & [0.62] \\ \cline{4-12}
 
  &\multirow{2}{*}{100}&\multirow{2}{*}{20}&2.6 & {\bf 0.3} & 42.1 & 22.8 & 44.2  & 466.6 & 88.9 & 121.2 & 114.7 \\ 
 &&&[0.99] & {\bf[1.00]} & [0.95] & [0.95] & [0.93] & [0.11] & [0.82] & [0.78] & [0.77] \\\cline{4-12}
 
 &\multirow{2}{*}{300}&\multirow{2}{*}{10}& 17.4 & {\bf 4.0} & 180.0 & 683.3 & 519.2 & 1268.3 & 997.7 & 1075.8 & 1085.7  \\
&&&[0.99] & {\bf[1.00]} & [0.89] & [0.66] & [0.55]  & [0.17] & [0.34] & [0.29] &[0.28] \\\cline{4-12}

&\multirow{2}{*}{300}&\multirow{2}{*}{20} & 7.1& {\bf 3.2} & 122.3 & 124.4 & 157.1  & 1445.9 & 613.5 & 723.0 & 727.2\\
&&&{\bf [1.00]} & {\bf[1.00]} & [0.93]& [0.92] & [0.90]  & [0.05] & [0.60] & [0.53] & [0.52] \\ 
 \cline{1-12}
 
 \multirow{8}{*}{$\mathcal{S}_{PM_2}$}&\multirow{2}{*}{100}&\multirow{2}{*}{10}&90.0 & {\bf 66.6} & 115.2 & 108.3 & 152.9 & 404.3 & 285.5 & 307.2 & 313.8\\ 
  &&&[0.82] & {\bf [0.86]} & [0.77] & [0.77]& [0.70]  & [0.24] & [0.47] & [0.43] & [0.41]\\\cline{4-12}
  
   &\multirow{2}{*}{100}&\multirow{2}{*}{20}& 26.9 & {\bf 2.4} & 81.5 & 59.8 & 91.5 &  468.1 & 217.3 & 245.2 & 249.5\\
 &&& [0.94] & {\bf[1.00]} & [0.85] & [0.87]  &[0.82] & [0.11] & [0.60] & [0.55] &[0.53] \\\cline{4-12}
 
  & \multirow{2}{*}{300}&\multirow{2}{*}{10}& 81.1 & {\bf 26.8} & 468.4 & 609.8 & 472.1 & 1388.4 & 1294.7 & 1321.5 & 1328.4\\
 &&&[0.95] & {\bf[0.98]} & [0.69] & [0.68] & [0.60] & [0.09] & [0.15] & [0.13] & [0.13] \\\cline{4-12}

  &\multirow{2}{*}{300}&\multirow{2}{*}{20}&77.2 &{\bf 3.4} & 313.6 & 267.5 & 337.0  & 1469.0 & 1205.9 & 1251.8 & 1258.9\\
  &&&[0.95] & {\bf [1.00]} & [0.79] & [0.82] & [0.78]  & [0.04] & [0.21] & [0.18] & [0.18]\\\hline
  
 \multirow{8}{*}{$\mathcal{S}_{HS_1}$}&\multirow{2}{*}{100}&\multirow{2}{*}{10}&24.9 & {\bf 20.6} & 22.9 & 54.9 & 115.9 & 334.7 & 150.9 & 180.3 & 186.0 \\ 
  &&&[0.97] & [0.98] & {\bf[0.99]} & [0.91] & [0.80]  & [0.37] & [0.71] & [0.66] & [0.64] \\ \cline{4-12}
  
    &\multirow{2}{*}{100}&\multirow{2}{*}{20}& 18.7 & 15.6 & 22.8 & {\bf 8.7} & 33.4  & 498.8 & 46.7 & 64.1 & 60.8 \\
 &&&[0.98] & [0.98] & {\bf[1.00]} & {\bf[1.00]} &[0.97] &  [0.05] & [0.92] & [0.89] & [0.89] \\\cline{4-12}
 
  &\multirow{2}{*}{300}&\multirow{2}{*}{10}&172.0 & 159.8 & {\bf 37.9} & 205.5 & 490.6  & 1098.6 & 627.0 & 752.9 & 769.4 \\
 &&&[0.89] & [0.90] & {\bf[0.99]} & [0.87] & [0.68] & [0.28] & [0.59] & [0.50] & [0.49] \\\cline{4-12}

 &\multirow{2}{*}{300}&\multirow{2}{*}{20}&7.4 & {\bf 7.0} & 22.7 & 11.4 & 114.1 & 1402.6 & 237.8 & 319.7 & 322.3\\ 
  &&&{\bf [1.00]} & {\bf[1.00]} & {\bf[1.00]} & {\bf[1.00]} & [0.94]  &[0.08] & [0.84] & [0.79] & [0.79] \\
 \hline
 
\multirow{8}{*}{$\mathcal{S}_{HS_2}$}&\multirow{2}{*}{100}&\multirow{2}{*}{10}&92.6 & 74.0 & {\bf 68.7} & 123.7 & 162.3  & 382.4 & 228.2 & 250.2 & 256.6\\
  &&&[0.83] & [0.86] & {\bf[0.88]} & [0.77] & [0.70]  & [0.27] & [0.56] & [0.52] & [0.50] \\ \cline{4-12}
  
  &\multirow{2}{*}{100}&\multirow{2}{*}{20}&24.4 & { 20.0} & 22.2 & {\bf 12.4} & 38.3 &  500.3 & 87.5 & 103.5 & 102.9 \\
 &&&[0.96] & [0.97] & {\bf[1.00]} & [0.99] &[0.95] &  [0.04] & [0.83] & [0.80] & [0.80] \\\cline{4-12}
 
&\multirow{2}{*}{300}&\multirow{2}{*}{10}& 319.1 & 463.8 & {\bf 245.6} & 516.9 & 683.5  & 1267.9 & 998.0 & 1076.0 & 1085.5 \\
  &&&[0.79] & [0.69] & {\bf [0.84]} & [0.66] & [0.55] & [0.17] & [0.34] & [0.29] & [0.28] \\\cline{4-12}
  
&\multirow{2}{*}{300}&\multirow{2}{*}{20}& 8.7 & {\bf 8.0} & 23.2 & 30.3& 155.5  & 1430.7 & 437.6 & 516.2 & 523.3\\
 &&&{\bf[1.00]} & {\bf[1.00]} & {\bf[1.00]} & [0.99] & [0.91]  & [0.06] & [0.71] & [0.66]& [0.65] \\
 \hline 
    \end{tabular}
    \caption{Average recovery distances [coverages] of different methods based on 500 independent replicates under different simulation scenarios.}
    \label{Tab:recovery}
\end{table}

\begin{figure}[htp]
\centering
\includegraphics[width=0.98\linewidth]{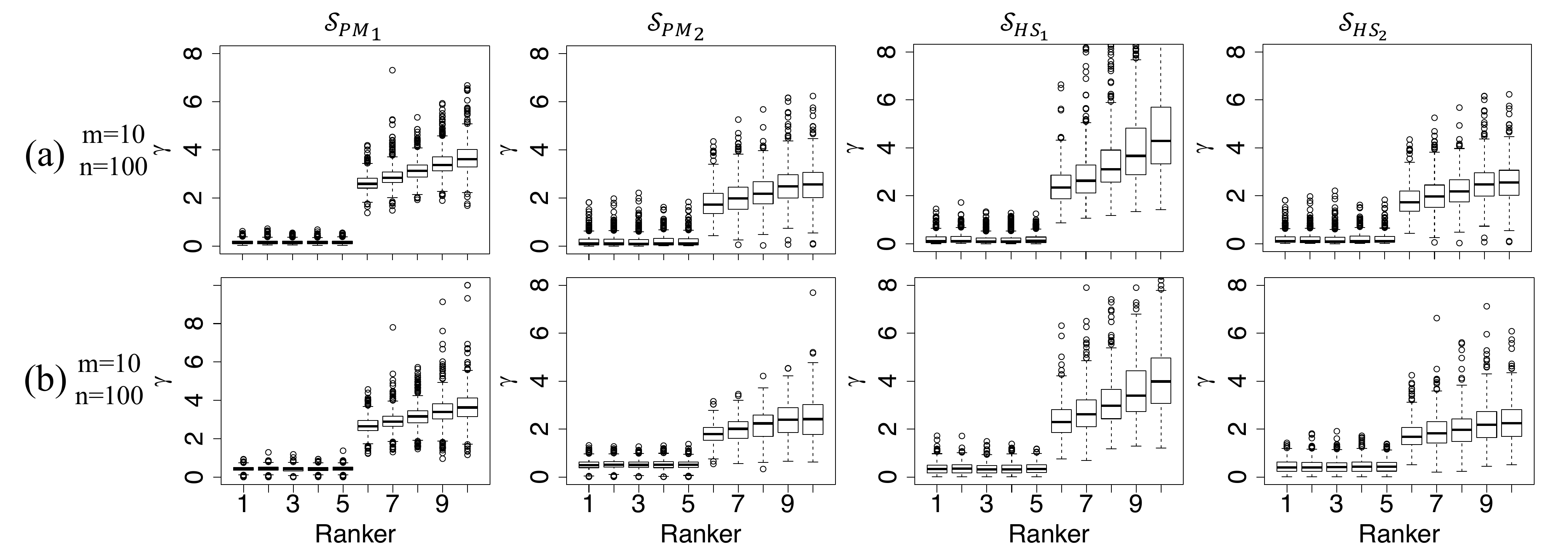}
\caption{(a) The boxplots of $\{\bar\gamma_k\}$ estimated by PAMA$_B$ with $m=10$ and $n=100$. (b) The boxplots of $\{\hat\gamma_k\}$ estimated by PAMA$_F$ with $m=10$ and $n=100$. Each column denotes a scenario setting. The results were obtained from 500 independent replicates.}
\label{Fig:gamma}
\end{figure}

Figure~\ref{Fig:n100m10} (a) shows the boxplots of recovery distances and the coverages of the nine competing methods in the four simulation scenarios with $m=10$, $n=100$, and $n_1=10$. The five methods from the left outperform the other four methods by a significant gap, and the PAMA-based methods generally perform the best.
Figure~\ref{Fig:n100m10} (b) confirms that the methods based on the Partition-Mallows model  enjoys the same capability as BARD in detecting quality differences between informative and non-informative rankers. However, while both BARD and PAMA can further discern quality differences among informative rankers, EMM fails this more subtle task.
Similar figures for other combinations of $(m,n)$ are provided in the Supplementary Material, highlighting consistent results as in Figure~\ref{Fig:n100m10}.

\begin{figure}[htp]
\centering
\includegraphics[width=\linewidth]{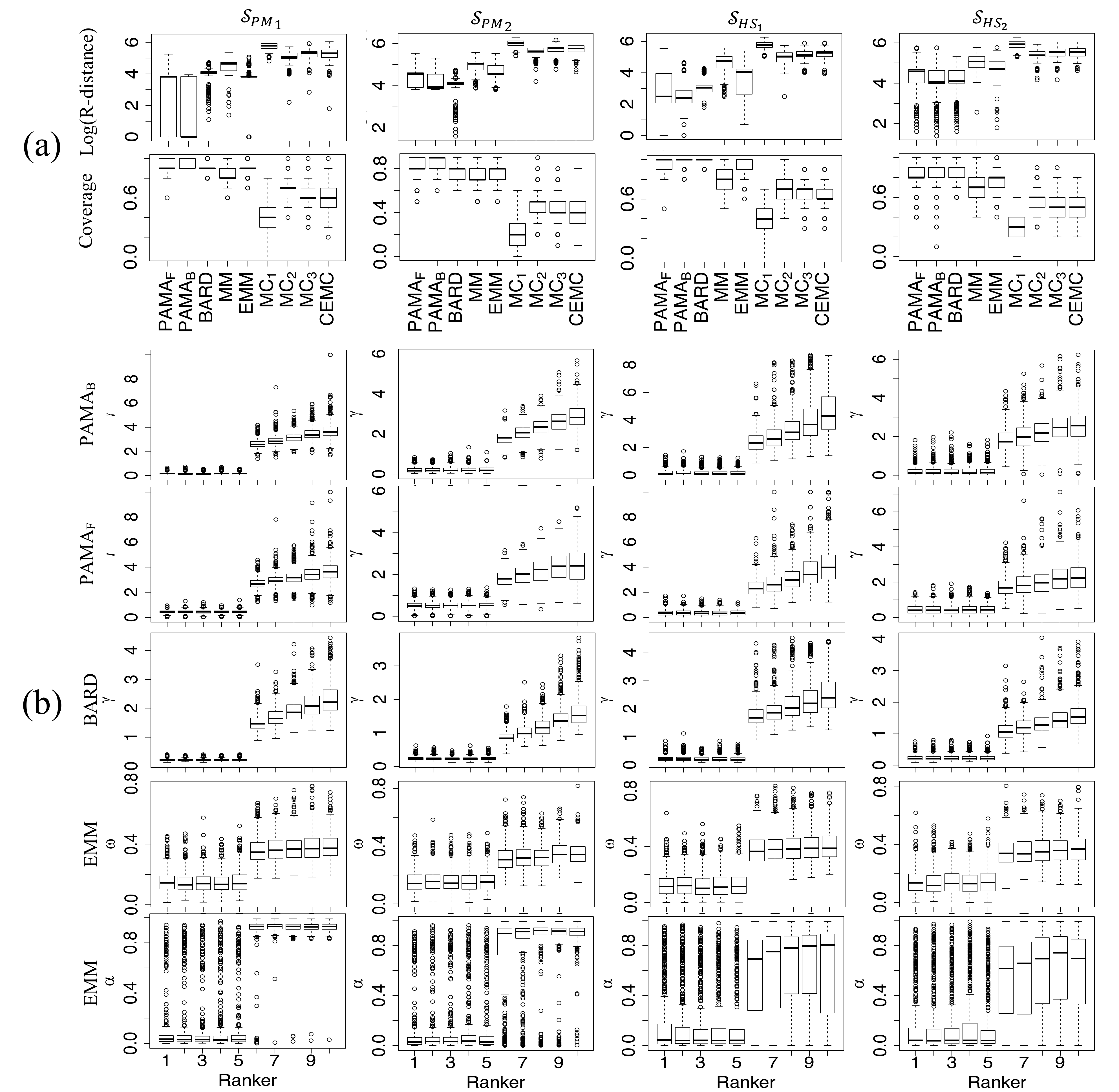} 
\caption{Boxplots of the rank aggregation results of 500 replications obtained from different methods under various scenarios with $m$=10, $n$=100, and $n_1$=10.
(a) Recovery distances in log scale and coverage obtained from nine algorithms. (b) Quality parameters obtained by Partition-type models and EMM.}
\label{Fig:n100m10}
\end{figure}

\subsection{Robustness to the Specification of $n_1$}
We need to specify $n_1$, the number of relevant entities, when applying PAMA$_B$ or PAMA$_F$. In many practical problems, however, there may not be a strong prior information on $n_1$ and there may not even be clear distinctions between relevant and background entities.
To examine robustness of the algorithm with respect to the specification of $n_1$, we design a simulation setting $\cS_{HS_3}$ to mimic the no-clear-cut scenario and investigate how the performance of PAMA is affected by the specification of $n_1$ under this setting. Formally, $\cS_{HS_3}$ assumes that $\tau_k=sort(i\in U\ by\ S_{ik}\downarrow)$, where $S_{ik}\sim N(\mu_{ik},1)$, following the same data generating framework as $\cS_{HS}$ defined in the Section \ref{sec:SimuSetting}, with $\mu_{ik}$ being replaced by 
$$\mu_{ik}=\left\{
   \begin{array}{ll}
   0, & \mbox{if } k\leq\frac{m}{2}, \\ 
    \frac{2\times a^* \times k/m}{1+e^{-b^*\times (70-i)}}, & \mbox{otherwise}, \\ 
  \end{array}
  \right. $$
where $a^*=50$ and $b^*=0.1$.
Different from $\cS_{HS_1}$ and $\cS_{HS_2}$, where $\mu_{ik}$ jumps from 0 to a positive number as $i$ ranges from background to relevant entities, in the $\cS_{HS_3}$ scenario $\mu_{ik}$  increases smoothly as a  function of $i$ for each informative ranker $k$.
In such cases, the concept of ``relevant'' entities is not well-defined.

We simulate 500 independent data sets from $\cS_{HS_3}$ with $n=100$ and $m=10$. For each data set, we try different specifications of $n_1$ ranging from 10 to 50 and compare PAMA to several competing methods based on their performance on recovering the top-$A$ list $[E_1\preceq E_2\preceq\cdots \preceq E_{A}]$, which is still well-defined based on the simulation design. The results summarized in Table~\ref{Tab:misspecifiedd} show that no matter which $n_1$ is specified, the partition-type models consistently outperform all the competing methods in terms of a lower recovery distance from the true top-$n_1$ list of items, i.e., $[E_1\preceq E_2\preceq\cdots \preceq E_{n_1}]$. Figure \ref{fig:consistentd} illustrates in details the average aggregated rankings of the top-10 entities by PAMA as $n_1$ increases, suggesting that PAMA is able to figure out the correct rankings of the top entities effectively.
These results give us confidence that PAMA is robust to misspecification of $n_1$.

\begin{table} \scriptsize
    \centering
    \begin{tabular}{ccc|ccc|cc|cccc} 
    \hline 
    \multicolumn{3}{c|}{Configuration}&\multicolumn{3}{c|}{Partition-type Models}&\multicolumn{2}{c|}{{Mallows Models}} &\multicolumn{4}{c}{MC-based Models} \\ \cline{1-3} \cline{4-6} \cline{7-8} \cline{9-12}
     $n$ &$m$ &$n_1$& PAMA$_F$&PAMA$_B$ &BARD & EMM & MM& MC$_1$ & MC$_2$ & MC$_3$ & CEMC\\  \cline{1-3} \cline{4-6} \cline{7-8} \cline{9-12}
\multirow{2}{*}{100}&\multirow{2}{*}{10} &\multirow{2}{*}{10}&44.8 & {\bf 34.6} & 42.6 & 61.5 & 227.7  & 423.8 & 45.6 & 199.1 & 241.3 \\
&&&[0.90] & [0.93] & {\bf [0.96]} & [0.88] & [0.58]  & [0.20] & [0.92] & [0.63] & [0.54] \\ \cline{4-12}
 \multirow{2}{*}{100}&\multirow{2}{*}{10}&\multirow{2}{*}{20}&39.2 & {\bf 33.9} & 94.2& 107.0 & 308.6  & 764.6 & 52.3 & 268.9 & 372.9  \\ 
 &&& [0.95] & [0.96] & {\bf [0.99]} & [0.90] & [0.75]  & [0.33] & [0.96] & [0.78] & [0.67] \\\cline{4-12}
\multirow{2}{*}{100}&\multirow{2}{*}{10}&\multirow{2}{*}{30}& {\bf 27.5} & 29.2 & 207.4 & 126.2 & 360.6  & 1040.2 & 67.8 & 325.4 & 445.0 \\
&&&{\bf [0.98]} & {\bf [0.98]} & {\bf [0.98]} & [0.93] & [0.83]  & [0.44] & [0.96] & [0.84] & [0.77]  \\\cline{4-12}
\multirow{2}{*}{100}&\multirow{2}{*}{10}& \multirow{2}{*}{40}& {\bf 16.0} & 17.4 & 363.9 & 131.6 & 408.1  & 1274.1 & 83.1 & 382.5 & 486.9 \\ 
 &&&{\bf [0.99]} & {\bf [0.99]} & [0.98] & [0.95] & [0.87]  & [0.54] & [0.97] & [0.87] & [0.83] \\
 \cline{4-12}
 \multirow{2}{*}{100}&\multirow{2}{*}{10}& \multirow{2}{*}{50}&  {\bf 8.8} & 9.3 & 565.3 & 134.6 & 452.2  & 1484.2 & 109.2 & 446.4 & 524.9 \\ 
 &&&{\bf [1.00]} & {\bf [1.00]} & [0.99] & [0.97] & [0.90]  & [0.62] & [0.96] & [0.89] & [0.88] \\
 \cline{1-12}
    \end{tabular}
    \caption{Average recovery distances [coverages] of different methods based on 500 independent replicates under scenario $\mathcal{S}_{HS_{3}}$.
    }
    \label{Tab:misspecifiedd}
\end{table}

\begin{figure}[ht]
    \centering
    \includegraphics[width=0.6\textwidth, height=0.4\textheight]{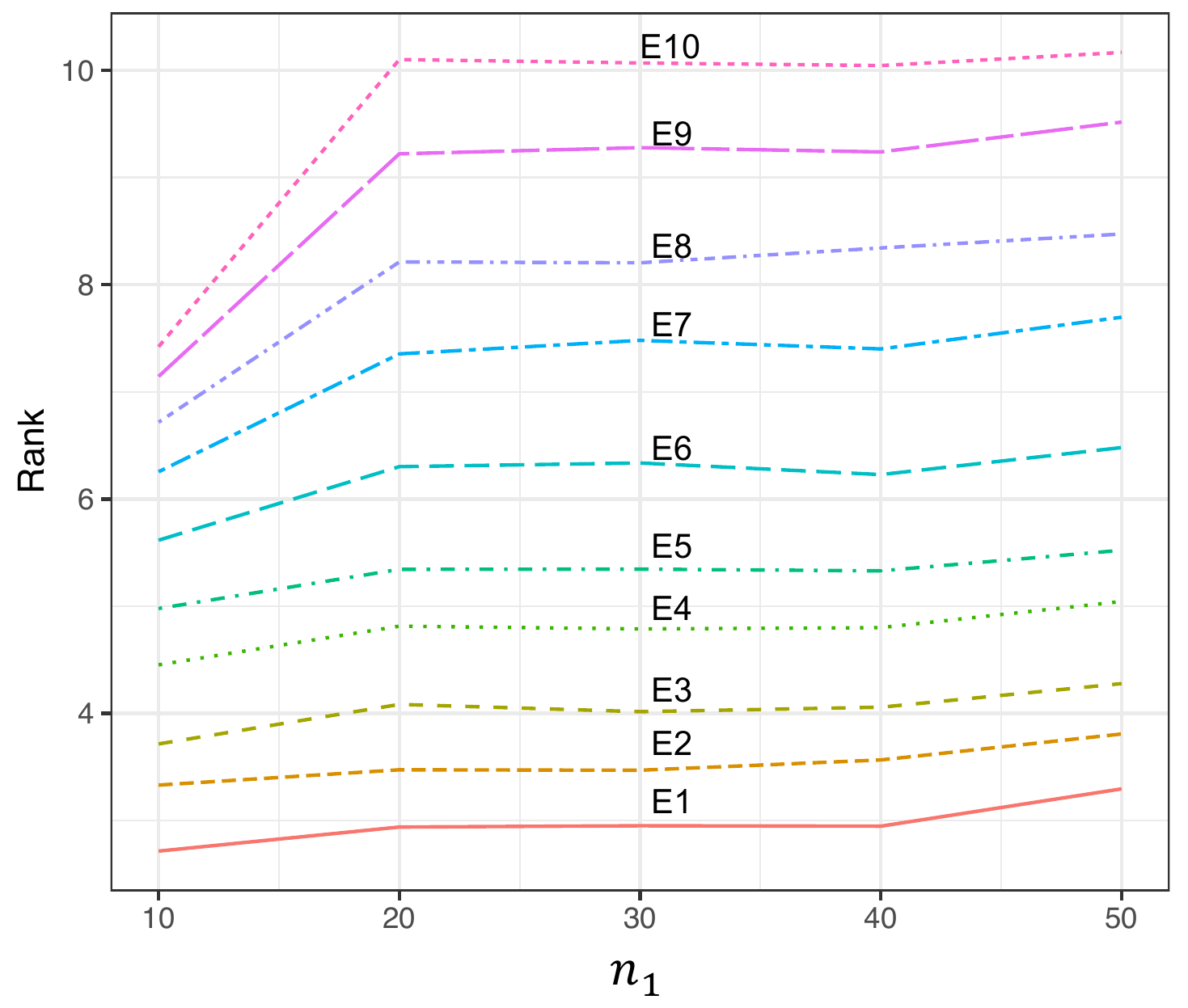}
    \caption{Average aggregated rankings of the top-10 entities by PAMA as $n_1$ increases from 10 to 50 for simulated data sets generated from $\cS_{HS_3}$.}
    \label{fig:consistentd}
\end{figure}

Noticeably, although PAMA and BARD achieve comparable coverage as shown in Table \ref{Tab:misspecifiedd}, PAMA dominates BARD uniformly in terms of a much smaller recovery distance in all cases, suggesting that PAMA is capable of figuring out the detailed ranking of relevant entities that is missed by BARD. In fact, 
since BARD relies only on $\rho_i \triangleq P(I_i = 1 \mid \tau_1, \cdots , \tau_m)$ to rank entities, in cases where the signal to distinguish the relevant and background entities is strong, many $\rho_i$'s are very close to 1, resulting in nearly a ``random" ranking among the top relevant entities.
Theoretically, if all relevant entities are recognized correctly but ranked randomly, the corresponding recovery distance would increase with $n_1$ in an order of $\mathcal{O}({n_1}^2)$, which matches well with the increasing trend of the recovery distance of BARD shown in Table \ref{Tab:misspecifiedd}.

We also tested the model's performance when there is a true $n_1$ but it is mis-specified in our algorithm. We varied $n_1$ as $8, 10$ and $18$, respectively, for setting $\cS_{HS_1}$ with $n$=100 and $m$=10, where the true $n_1$=10 (the first ten entities). Figure \ref{fig:misspecifiedd} shows boxplots of $\cI$ for each mis-specified case. For the visualization purpose, we just show the boxplot of $E_1$ to $E_{20}$. The other entities are of the similar pattern with $E_{20}$. The figure shows a robust behavior of PAMA$_B$ as we vary the specifications of $n_1$. It also shows that the results are slightly better if we specify a $n_1$ that is moderately larger than its true value. The consistent results of other mis-specified cases, such as $5, 12, 15$, can be found in the Supplementary Material.

\begin{figure}
    \centering
    \includegraphics[width=0.7 \linewidth, height=0.6 \linewidth]{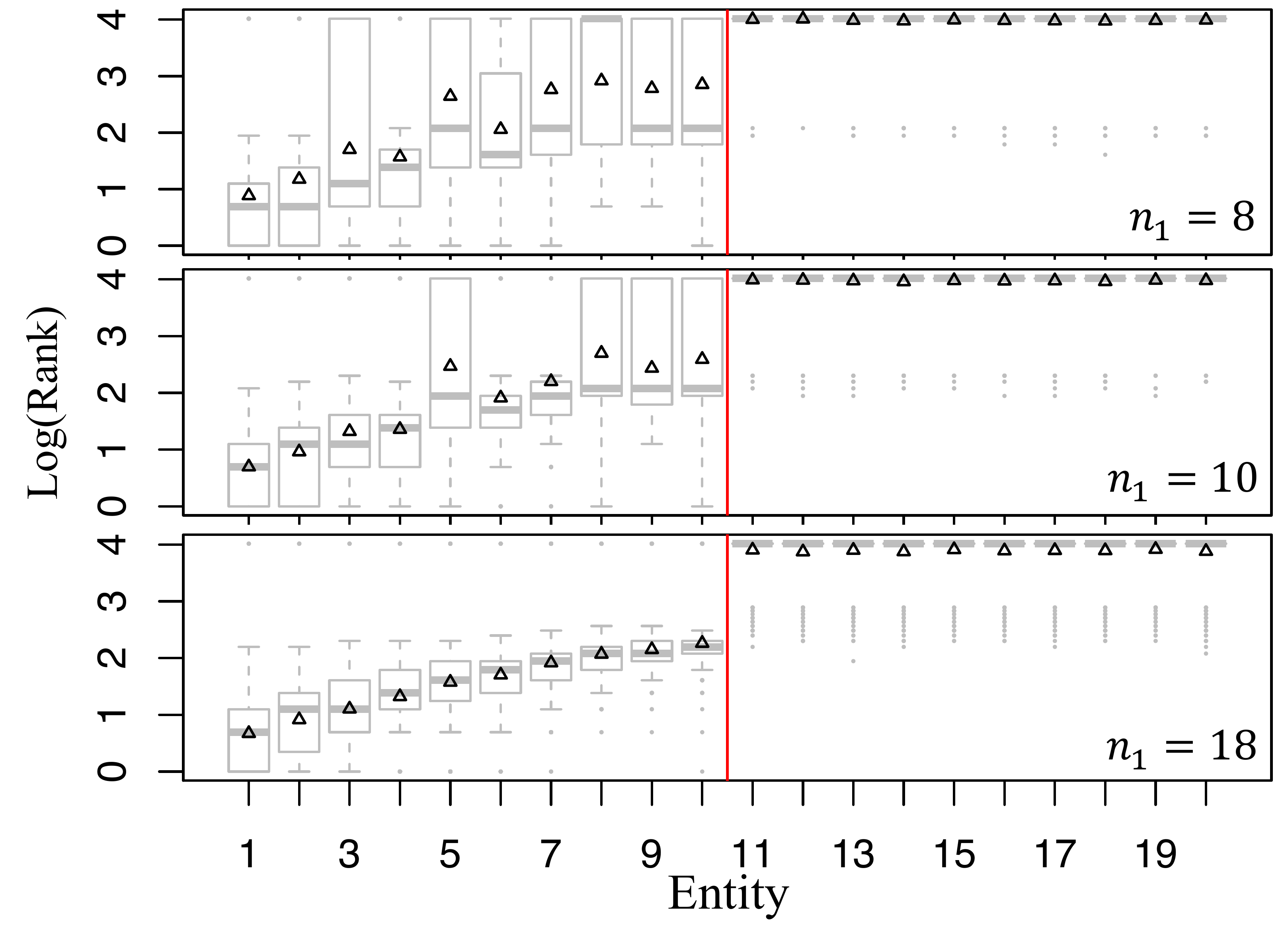}
    \caption{Boxplots of the estimated $\cI$ from  500 replications under the setting of $\cS_{HS_1}$ with $n_1$ being set as 8, 10 and $18$, respectively. The true $n_1$ is $10$.
    The vertical lines separate relevant entities (left) from background ones.
    The Y-axis shows the logarithm of the entities' ranks. The rank of a background entity is replaced by their average $\frac{100+10+1}{2}$. 
    The triangle denotes the mean rank of each entity.}
    \label{fig:misspecifiedd}
\end{figure}
\section{Real Data Applications} \label{sec:realdata}

\subsection{Aggregating Rankings of NBA Teams}
We applied PAMA$_B$ to the NBA-team data analyzed in \cite{deng2014bayesian}, and compared it to competing methods in the literature. The NBA-team data set contains 34 ``predictive" power rankings of the 30 NBA teams in the 2011-2012 season. The 34 ``predictive" rankings were obtained from 6 professional rankers (sports websites) and 28 amateur rankers (college students), and the data quality varies significantly across rankers. More details of the dataset can be found in Table 1 of \cite{deng2014bayesian}.

\begin{figure}[htp]
\centering
  \includegraphics[width=\linewidth]{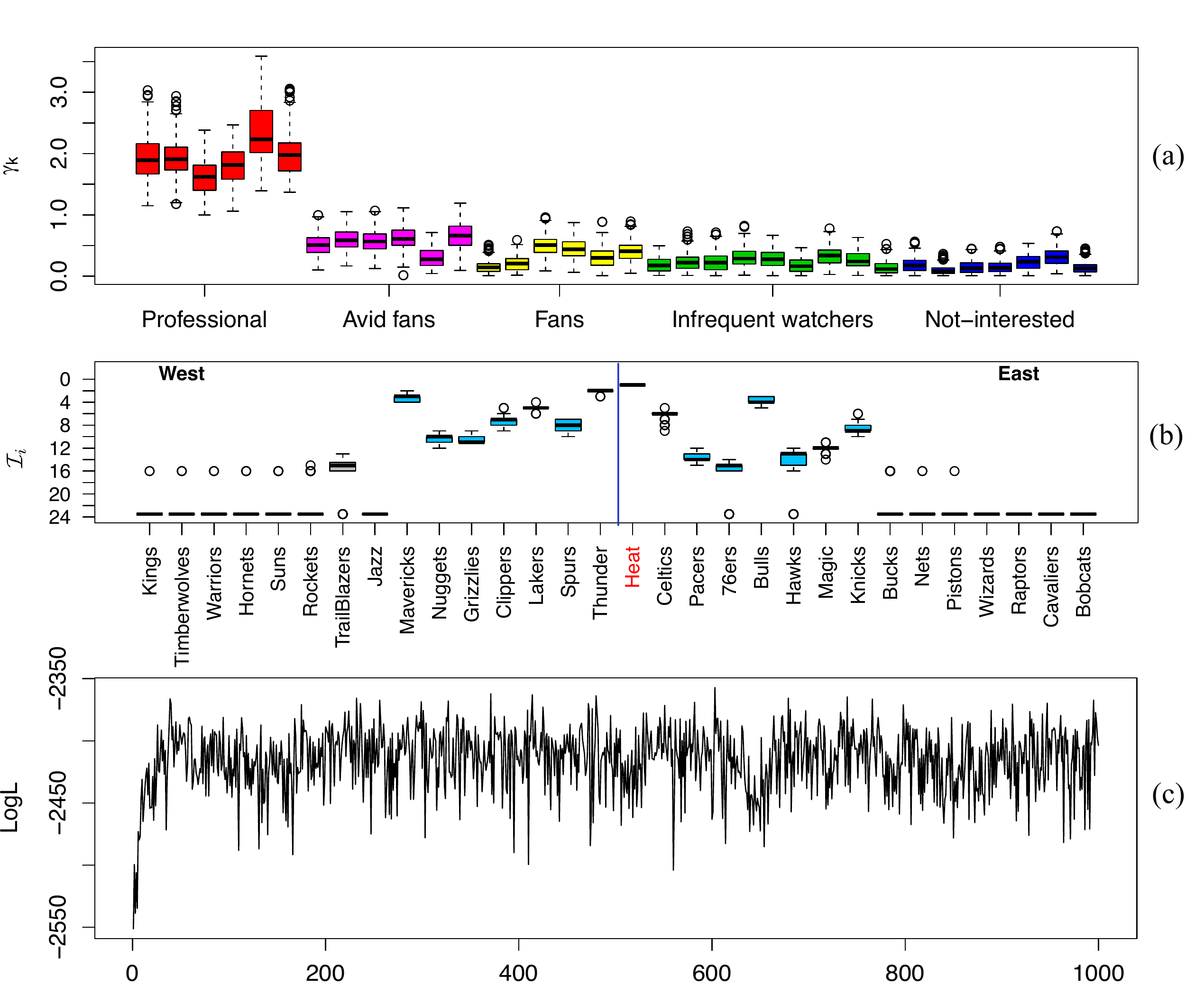}
 \caption{Results from PAMA$_B$ for the NBA-team dataset. (a) boxplots of posterior samples of $\gamma_k$. (b) barplots of $\bar{\cI}_i$ where the vertical line divides the NBA teams in Western and Eastern Conferences. (c) the trace plot of the log-likelihood function.}
\label{Fig:NBA} 
\end{figure}

Figure \ref{Fig:NBA} displays the results obtained by PAMA$_B$ (the partial-list version with $n_1$ specified as 16).
Figure \ref{Fig:NBA} (a) shows the posterior distributions, as boxplots, of the quality parameter of each involved ranker. Figure \ref{Fig:NBA} (b) shows the posterior distribution of the aggregated power ranking of each NBA team. 
All the posterior samples that reports ``0" for the rank of an entity means that the entity is a background one, for visualization purpose we replace ``0'' by the rank of background average rank,  $\frac{n+n_1+1}{2}=\frac{30+16+1}{2}=23.5$. The final set of 16 playoff teams are shown in blue while the champion of that season is shown in red (i.e., Heat). Figure \ref{Fig:NBA} (c) shows the trace plot of the log-likelihood of the PAMA model along the MCMC iteration. 
Comparing the results to Figure 8 of \cite{deng2014bayesian}, we observe the following facts: (1) PAMA$_B$ can successfully discover the quality difference of rankers as BARD; (2) PAMA$_B$ can not only pick up the relevant entities effectively like BARD, but also rank the discovered relevant entities reasonably well, which cannot be achieved by BARD; (3) PAMA$_B$ converges quickly in this application.

We also applied other methods, including MM, EMM and Markov-chain-based methods, to this data set. We found that none of these methods could discern the quality difference of rankers as successfully as PAMA and BARD. Moreover, using the team ranking at the end of the regular season as the surrogate true power ranking of these NBA teams in the reason, we found that PAMA also outperformed BARD and EMM by reporting an aggregated ranking list that is the closest to the truth. 
Table \ref{Tab:NBArankresults} provides the detailed aggregated ranking lists inferred by BARD, EMM and PAMA respectively, as well as their coverage of and Kendall $\tau$ distance from the surrogate truth. 
Note that the Kendall $\tau$ distance is calculated for the eastern teams and western teams separately because the NBA Playoffs proceed at the eastern conference and the western conference in parallel until the NBA final, in which the two conference champions compete for the NBA champion title, making it difficult to validate the rankings between Eastern and Western teams.

\begin{table}[ht]  \footnotesize
 \centering
   \begin{tabular}{|c|c:c|c:c|c:c|c:c|}
    \hline
    Ranking& \multicolumn{2}{c|}{Surrogate truth}&\multicolumn{2}{c|}{BARD}&\multicolumn{2}{c|}{EMM}&\multicolumn{2}{c|}{PAMA}\\ \hline
      & Eastern & Western & Eastern & Western  & Eastern & Western & Eastern & Western \\ \hline
     1&Bulls& Spurs&\emph{Heat}&\emph{Thunder}&Heat&Thunder&Heat&Thunder\\
     2&Heat&Thunder&\emph{Bulls}&\emph{Mavericks}&Bulls&Maverick&Bulls&Maverickss\\
     3&Pacers&Lakers&\emph{Celtics}&\emph{Clippers}&Knicks&Clippers&Celtics&Lakers\\
     4&Celtics&Grizzlies&\emph{Knicks}&\emph{Lakers}&Celtics&Lakers&Knicks&Clippers\\
     5&Hawks&Clippers&\emph{Magic}&\emph{Spurs}&Magic&Spurs&Magic&Spurs\\
     6&Magic&Nuggets&\emph{Pacers}&\emph{Grizzlies}&Pacers&Grizzlies&Hawks&Nuggets\\
     7&Knicks&Mavericks&76ers&\emph{Nuggets}&76ers&Nuggets&Pacers&Grizzlies\\ 
     8&76ers&Jazz&Hawks&Jazz$^*$&Hawks$^*$&Jazz$^*$&76ers&Jazz$^*$\\ \hline
     Kendall $\tau$ &-&-&14.5&10.5&9&10&8&10\\ \hline
     Coverage & \multicolumn{2}{c|}{-}& \multicolumn{2}{c|}{$\frac{15}{16}$}&\multicolumn{2}{c|}{$\frac{14}{16}$}& \multicolumn{2}{c|}{$\frac{15}{16}$} \\ \hline
    \end{tabular}
    \caption{Aggregated power ranking of the NBA teams inferred by BARD, EMM, and PAMA, respectively, and the corresponding coverage of and the Kendall $\tau$ distance from the surrogate true rank based on the performances of these teams in the regular season.
    The teams in italic indicate that they have equal posterior probabilities of being in the relevant group, and the teams with asterisk are those that were misclassified to the background group.}
    \label{Tab:NBArankresults}
\end{table}

\subsection{Aggregating Rankings of NFL Quarterback Players with the Presence of Covariates}
Our next application is targeted at the NFL-player data reported by \cite{li2020bayesian}. The NFL-player data contains 13 predictive power rankings of 24 NFL quarterback players. The rankings were produced by 13 experts based on the performance of the 24 NFL players during the first 12 weeks in the 2014 season. The dataset also contains covariates for each player summarizing the performances of these players during the period, including the \emph{number of games played} (G), \emph{pass completion percentage} (Pct), the \emph{number of passing attempts per game} (Att), \emph{average yards per carry} (Avg), \emph{total receiving yards} (Yds), \emph{average passing yards per attempt} (RAvg), the \emph{touchdown percentage} (TD), the \emph{intercept percentage} (Int), \emph{running attempts per game} (RAtt), \emph{running yards per attempt} (RYds) and the \emph{running first down percentage} (R1st). 
Details of the dataset can be found in Table 2 of \cite{li2020bayesian}. 

\begin{figure}[htp]
\centering
  \includegraphics[width=\linewidth]{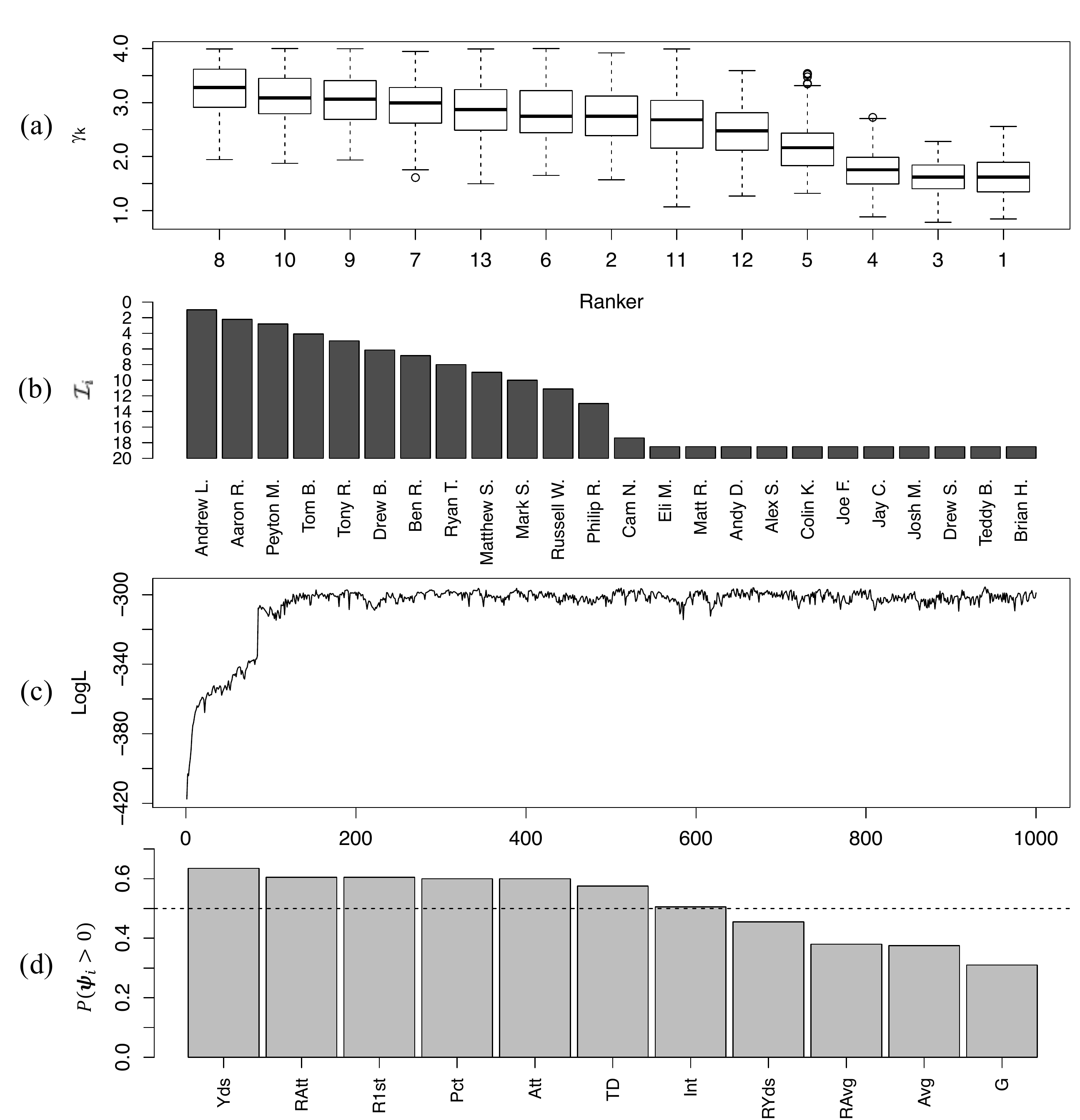}
  \caption{Key results from PAMA$_B$ for the NFL-player dataset. (a) Boxplots of posterior samples of $\gamma_k$. (b) Barplots of $\bar{\cI}_i$. (c) Trace plot of the log-likelihood. (d) Barplots of posterior probabilities for each coefficient to be positive.}
  \label{Fig:NFLcov}
\end{figure}

Here, we set $n_1=12$ in order to find which players are above average.
We analyzed the NFL-player data with PAMA$_B$ (the covaritate-assisted version) and the results are summarized in Figure \ref{Fig:NFLcov}: (a) the posterior boxplots of the quality parameter for all the rankers; (b) the barplot of $\bar{\cI_i}$ for all the NFL players with the descending order; (c) the traceplot of the log-likelihood of the model; and (d), the barplot of probabilities $P({\psi}_j >0)$ and the covariates are rearranged from left to right by decreasing the probability. From Figure \ref{Fig:NFLcov} (a), we observe that rankers 1, 3, 4 and 5 are generally less reliable than the other rankers. In the study of the same dataset in \cite{li2020bayesian}, the authors assumed that the 13 rankers fall into three quality levels, and reported that seven rankers (i.e., 2, 6, 7, 8, 9, 10 and 13) are of a  higher quality than the other six (see Figure 7 of \cite{li2020bayesian}). Interestingly, according to Figure \ref{Fig:NFLcov} (a), the PAMA algorithm suggested exactly the same set of high-quality rankers. In the meantime, ranker 2 is of the lowest quality among the seven high quality rankers in both studies. From Figure \ref{Fig:NFLcov} (b), a consensus ranking list can be obtained. Our result is consistent with that of Figure 6 in \cite{li2020bayesian}.  Figure \ref{Fig:NFLcov} (d) shows that six covariates are more probable to have positive effects.

\begin{table}[htb]
\def\arraystretch{0.8}
    \centering
    \begin{tabular}{|c|c|c|c|c|}
    \hline
         Ranking&Gold standard   & BARD & EMM& PAMA\\ \hline
         1 & Aaron R. & \emph{Andrew L.} & Andrew L.& Andrew L.\\
         2 & Andrew L. & \emph{Aaron R.}& Aaron R.& Aaron R.\\
         3 & Ben R. & \emph{Tom B.} & Tom B.& Tom B.\\
         4 &Drew B. & \emph{Drew B.} &Ben R. & Drew B.\\
         5 &Russell W.  & \emph{Ben R.}  & Drew B.& Ben R.\\
         6 &Matt R.  & \emph{Ryan T.}  & Ryan T. & Ryan T.\\
         7 &Ryan T.  &Russell W. & Russell W. &Russell W.\\
         8 &Tom B.& Philip R.*  & Philip R.*& Philip R.\\
         9 &Eli M.  &Eli M.* & Eli M.* &Eli M.*\\
         10 &Philip R.  &Matt R.* &Matt R.* &Matt R.*\\
         \hline
         R-distance& - &35.5&32&25\\ \hline
         Coverage& -  &0.7 &0.7& 0.8\\ \hline
    \end{tabular}
    \caption{Top players in the aggregated rankings inferred by BARD, EMM and PAMA. The entities in italic indicate that they have equal posterior probabilities of being in the  relevant group, and the players with asterisk are those that were mis-classified to the background group.}
    \label{tab:rankingsofNFL}
\end{table}
Using the Fantasy points of the players (\url{https://fantasy.nfl.com/research/players}) derived at the end of the 2014 NFL season as the surrogate truth, the recovery distance  and coverage of the aggregated rankings by different approaches can be calculated so as to evaluate the performances of different approaches. Note that the Fantasy points of two top NFL players Peyton Manning and Tony Romo are missing for unknown reasons, we excluded them from analysis and only report results for the top 10 positions instead of top 12. Table~\ref{tab:rankingsofNFL} summarizes the results, demonstrating that PAMA outperformed the other two methods.

\section{Conclusion and Discussion} \label{sec:discussion}
The proposed Partition-Mallows model embeds the classic Mallows model into a partition modeling framework developed earlier by \cite{deng2014bayesian}, which is analogous to the well-known ``spike-and-slab" mixture distribution often employed in Bayesian variable selection. Such a nontrivial ``mixture" combines the strengths of both the Mallows model and BARD's partition framework, leading to a stronger rank aggregation method that  can not only learn quality variation of rankers and distinguish relevant entities from background ones effectively, but also provide an accurate ranking estimate of the discovered relevant entities. Compared to other frameworks in the literature for rank aggregation with heterogeneous rankers, the  Partition-Mallows model enjoys more accurate results with better interpretability at the price of a moderate increase of computational burden. We also show that the Partition-Mallows framework can easily handle partial lists and and incorporate covariates in the analysis.

Throughout this work, we assume that the number of relevant entities $n_1$ is known. This is reasonable in many practical problems where a specific $n_1$ can be readily determined according to research demands. Empirically, we found that the ranking results are insensitive to the choice of a wide range of values of $n_1$. If needed, we may also determine $n_1$ according to a model selection criterion, such as AIC or BIC.

In the PAMA model, $\pi(\tau_k^0 \mid \tau_k^{0\mid 1})$ is assumed to be a uniform distribution. If the detailed ranking of the background entities is of interest, we can modify the conditional distribution $\pi(\tau_k^0 \mid \tau_k^{0\mid 1})$ to be the Mallows model or other models. A quality parameter can still be incorporated to control the interaction between relevant entities and background entities. The resulting likelihood function becomes more complicated, but is still tractable.

In practice, the assumption of independent rankers may be violated. In the literature, a few approaches have been proposed to detect and handle dependent rankers. For example, \cite{deng2014bayesian} proposed a  hypothesis-testing-based framework to detect pairs of over-correlated rankers and a hierarchical model to accommodate clusters of dependent rankers; \cite{JohnsonS2019} adopted an extended Dirichlet process and a similar hierarchical model to achieve simultaneous ranker clustering and rank aggregation inference. Similar ideas can be incorporated in the PAMA model as well to deal with non-independent rankers.

\section*{Acknowledgement}
We thank Miss Yuchen Wu for helpful discussions at the early stage of this work and the two reviewers for their insightful comments and suggestions that helped us improve the paper greatly. This research is supported in part by the National Natural Science Foundation of China (Grants 11771242 \& 11931001), Beijig Academy of Artificial Intelligence (Grant BAAI2019ZD0103), and the National Science Foundation of USA (Grants DMS-1903139 and DMS-1712714). The author Wanchuang Zhu is partially supported by the Australian Research Council (Data Analytics for Resources and Environments, Grant IC190100031)

{\section*{Supplementary Materials}}
\appendix
\section{Numerical Supports to Power-Law Model of $\tau_k^{0\mid 1}(i)$ }
\label{AP:Numerical_Validation_of_Power-Law}

Figure~\ref{Fig:test10} illustrates numerical evidences to support the power-law model for $\tau_k^{0\mid 1}(i)$ in different scenarios in a similar spirit of the Figure 2 in \cite{deng2014bayesian}: generating each ranker $\tau_k$ as the order of $\{X_{k,1},\cdots,X_{k,n}\}$, i.e.,
$\tau_k=sort(i\in U\mbox{ by }X_{k,i} \downarrow),$
where $X_{k,i}$ is generated from two different distributions $F_{k,0}$ and $F_{k,1}$ via the following mechanism:
$$X_{k,i}\sim F_{k,0}\cdot\mI(i\in U_B)+F_{k,1}\cdot\mI(i\in U_R),\ \forall\ i\in U.$$
Figure~\ref{Fig:test10} confirms that the log-log plot of $t$ versus $h(t)=P(\tau_k^{0\mid 1}(t)\mid \tau_k^1;I)$ is also nearly linear (when $t$ is not too large), suggesting the power-law model is acceptable as an approximation of the reality.

\begin{figure}[htp]
\centering
\begin{subfigure}{.45\textwidth}
  \centering
  \includegraphics[width=\linewidth]{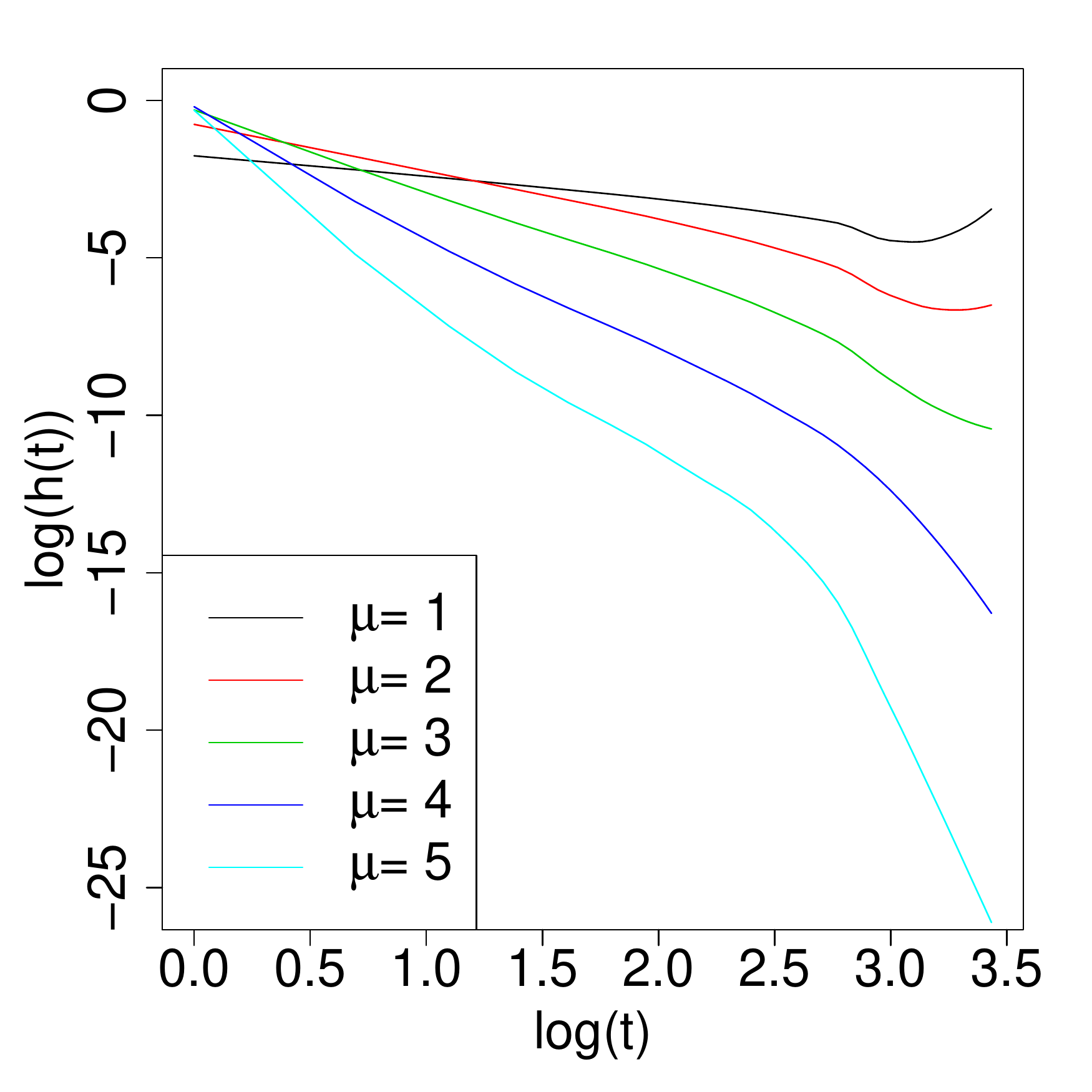}
  \caption{$n_1=30$.}
  \label{Fig:test1030}
\end{subfigure}%
\begin{subfigure}{.45\textwidth}
  \centering
  \includegraphics[width=\linewidth]{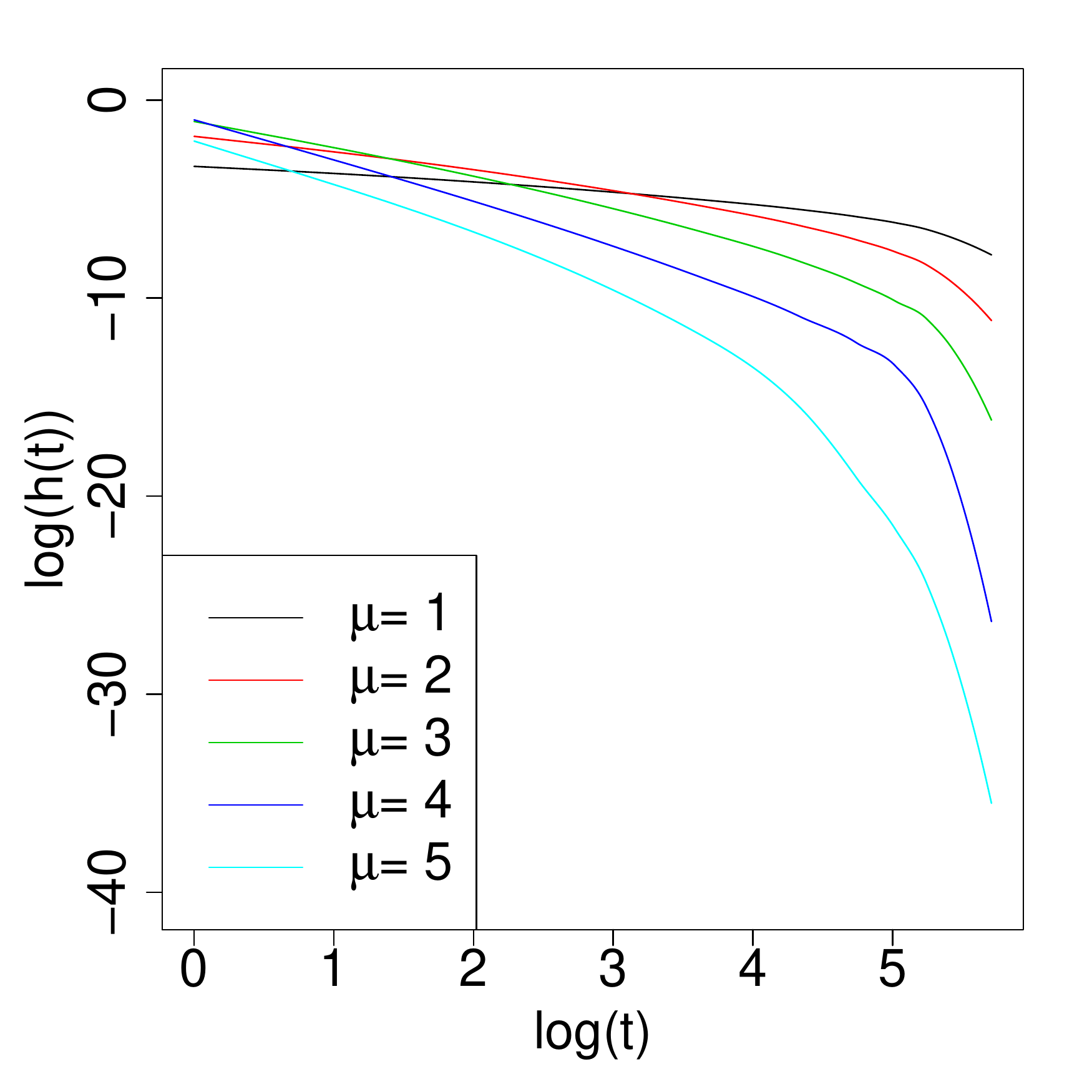}
  \caption{$n_1=300$.}
  \label{Fig:test10300}
\end{subfigure}
\caption{Log-log plots of relative reverse ranking $\tau_k^{0\mid 1}(i)=t$ versus the corresponding probability $h(t)=P(\tau_k^{0\mid 1}(i)=t\mid I)$ under different scenarios. The values of $h(t)$ are calculated by numerical integration.}
\label{Fig:test10}
\end{figure}

\section{Numerical Evidence for the Assumption $\phi_k = \phi \times \gamma_k$} \label{AP:gammaandphi}
In this section, we provide numerical evidence to support the assumption $\phi_k=\phi \times \gamma_k$ by showing that it approximates the reality reasonably well in more general settings. In the literature, the Thurstone hidden score model is widely used to generate ranked data.
Specifying $m=50$, $n=100$, $n_1=10$ and $\cI=(1,\cdots,n_1,0,\cdots,0)$, we explored two simulation settings based on the hidden score model, namely $\cS_{HS}^{'}$ and $\cS_{HS}^{''}$, in which we assume $$\tau_k = sort(i\in U\ by\ S_{ik} \downarrow)\ \mbox{where}\ S_{ik}\sim N(\mu_{ik},1),\ \forall\ 1\leq k\leq m,$$
and specify in the $\cS_{HS}^{'}$ setting
$$\mu_{ik}=\left\{
   \begin{array}{ll}
   -k \times g, &  i>n_1, \\ 
    (n_1-i) \times k \times s, & i \leq n_1; \\
  \end{array}
  \right. $$
and in the $\cS_{HS}^{''}$ setting
$$\mu_{ik}=\left\{
   \begin{array}{ll}
   -k \times g, &  i>n_1, \\ 
   0, & i=n_1,\\
    \mu_{(i+1) k} + U[0, k \times s], & i < n_1, \\ 
  \end{array}
  \right. $$
where $U[a, b]$ stands for a random number drawn from the Uniform distribution on interval $[a,b]$. 
Clearly, the quality of ranker $\tau_k$ increases monotonously with the ranker index $k$ in both settings, and there is a clear gap of mean scores between background and relevant entities.
The two hyper-parameters $(s,g)$ control the signal strength and ranker heterogeneity in the simulated data.

For each data set simulated from the above settings, given the true ranking list $\cI$, we can always fit a Mallows model for the rankings of relevant entities in each $\tau_k$ to get an estimated $\hat\phi_k$ and a power-law model for the relative rankings of the background entities among the relevant ones in each $\tau_k$ to get an estimated $\hat\gamma_k$, leading to a set of pairwise estimates $\{(\hat\phi_k,\hat\gamma_k)\}_{1\leq k\leq m}$.
Figure \ref{fig:relationshipphi} provides a graphical illustration of the estimated parameters $\{(\hat\phi_k,\hat\gamma_k)\}_{1\leq k\leq m}$ for typical simulated datasets from the $\cS_{HS}^{'}$ setting with different specifications of $(s,g)$, and Figure \ref{fig:relationshipphiS3} demonstrates the counterpart for the $\cS_{HS}^{''}$ setting.
From the figures, we can see clearly a strong linear trend between $\hat\phi_k$ and $\hat\gamma_k$ in all cases, suggesting that the presumed assumption approximates the reality very well.

\begin{figure}[!h]
    \centering
    \includegraphics[width=0.8\textwidth]{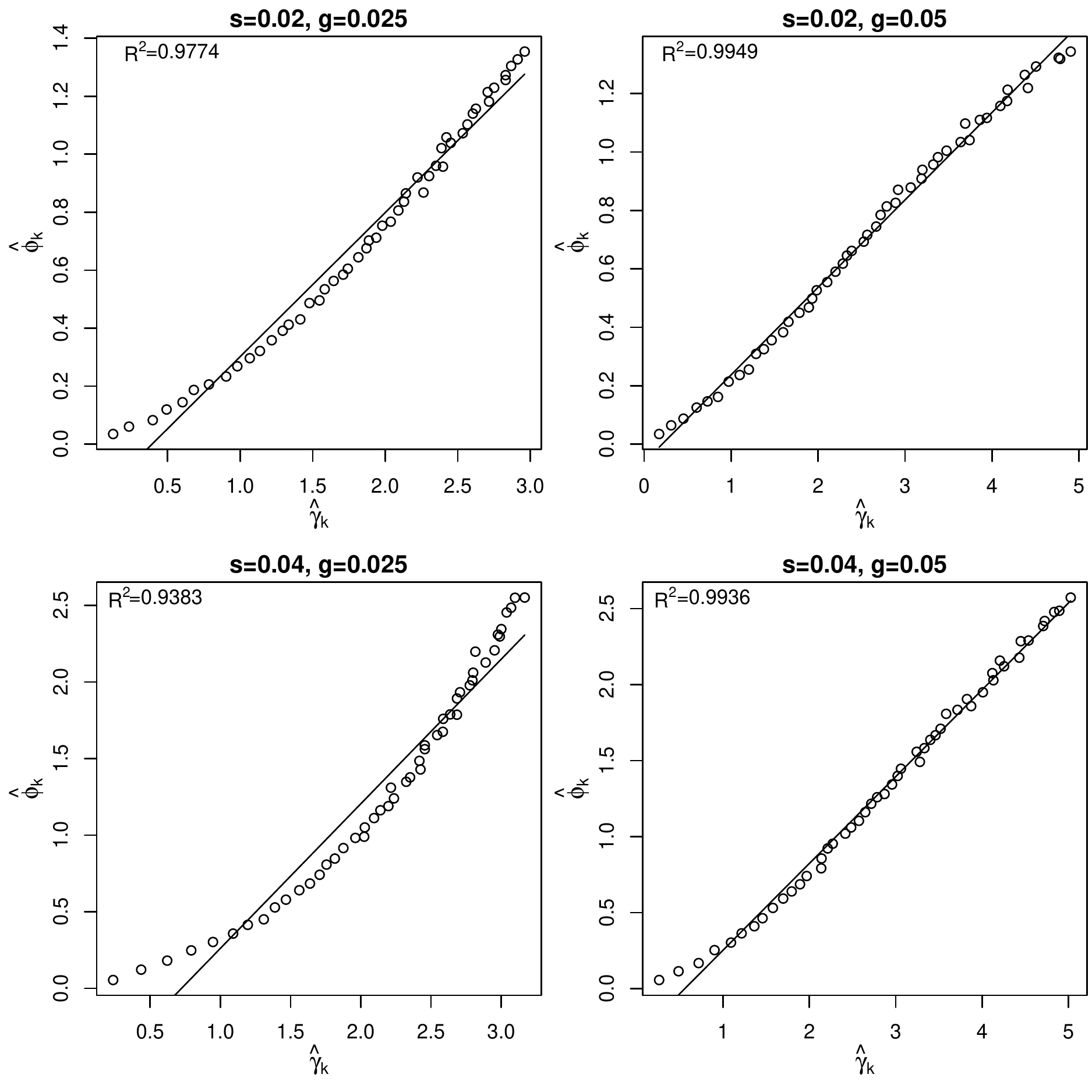}
    \caption{The relationship between $\hat\gamma_k$ and $\hat\phi_k$ under the scenario $\cS_{HS}^{'}$ with different specifications of $(s,g)$.}
    \label{fig:relationshipphi}
\end{figure}

\begin{figure}[!h]
    \centering
    \includegraphics[width=0.8\textwidth]{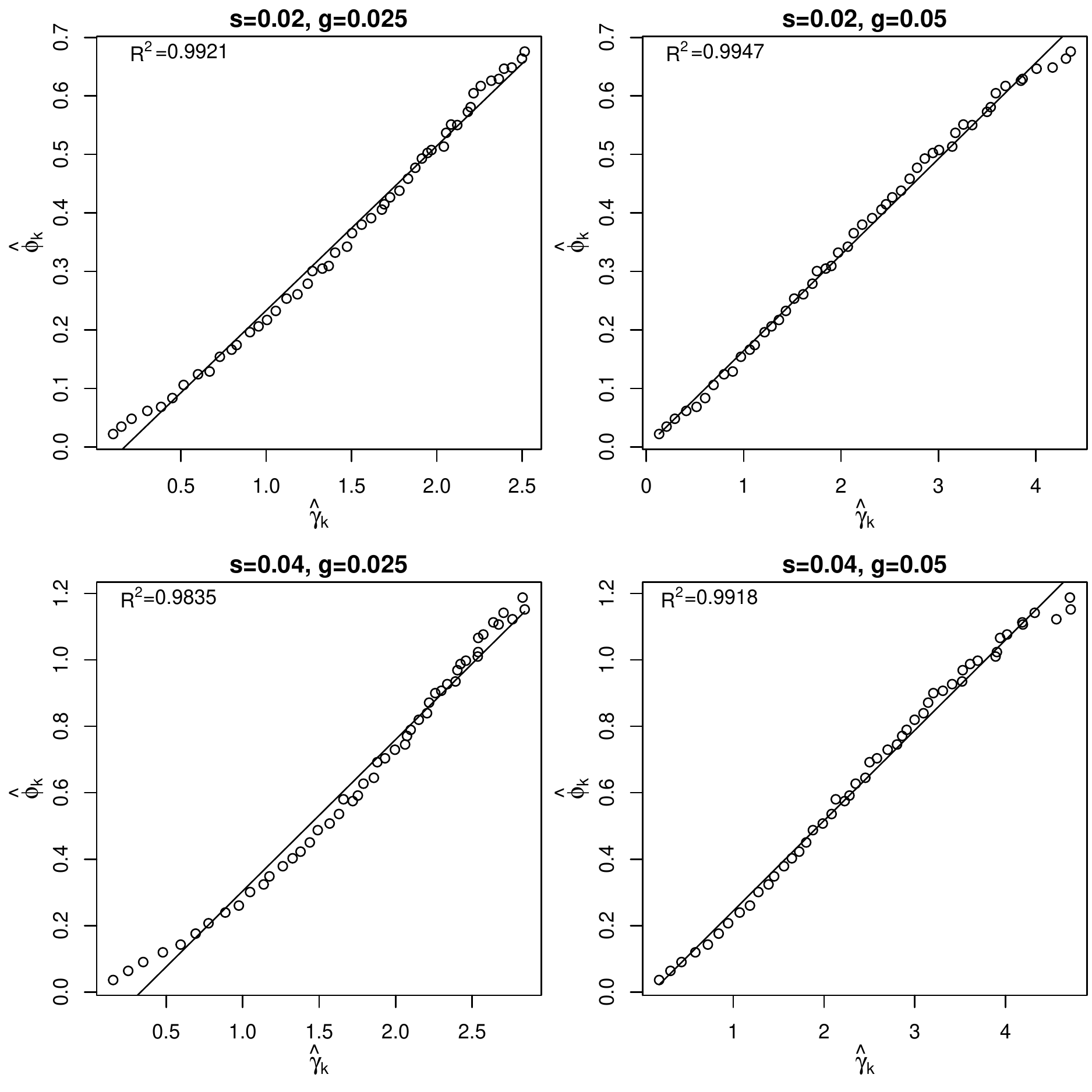}
    \caption{The relationship between $\hat\gamma_k$ and $\hat\phi_k$ under the scenario $\cS_{HS}^{''}$ with different parameters.}
    \label{fig:relationshipphiS3}
\end{figure}

\section{Proof of Theorem 1} \label{proof:ident}
\begin{proof}
We start with the degenerated special case where $m=1$. In this special case, the parameter vector $\btheta$ degenerates to a 3-dimensional vector $\btheta=(\cI,\phi,\gamma)$, and the PAMA model $P(\btau\mid\btheta)$ defined in Equation (19) degenerates to a simpler form $P(\tau\mid\btheta)$ as defined in Equation (18) with $\tau\in\Omega_n$.
To prove the identifiability of the degenerated PAMA model, we need to show that for any two proper parameter vectors $\btheta_1=(\cI_1,\phi_1,\gamma_1)$ and $\btheta_2=(\cI_1,\phi_1,\gamma_1)$ from the parameter space $\bTheta$, we always have:
$$\mbox{if}\ P(\tau\mid\cI_1,\phi_1,\gamma_1)=P(\tau\mid\cI_2,\phi_2,\gamma_2)\ \mbox{for}\ \forall\ \tau\in\Omega_n,\ \mbox{then}\ (\cI_1,\phi_1,\gamma_1)=(\cI_2,\phi_2,\gamma_2).$$
Here, we choose to prove the above conclusion by proving its equivalent counterpart: 
$$for\ \forall\  (\cI_1,\phi_1,\gamma_1)\neq(\cI_2,\phi_2,\gamma_2),\ \exists\ \tau\in\Omega_n,\ s.t.\  P(\tau\mid\cI_1,\phi_1,\gamma_1)\neq P(\tau\mid\cI_2,\phi_2,\gamma_2).$$
Apparently, the condition $(\cI_1,\phi_1,\gamma_1)\neq(\cI_2,\phi_2,\gamma_2)$ implies two possible scenarios: $$(i)\ (\phi_1,\gamma_1)\neq (\phi_2,\gamma_2),\ \mbox{or}\ (ii)\ (\phi_1,\gamma_1)= (\phi_2,\gamma_2)\ \mbox{with}\ \cI_1\neq\cI_2.$$
Below we discuss the two scenarios separately.

$(i)$ The scenario where $(\phi_1,\gamma_1)\neq (\phi_2,\gamma_2)$. 
Let 
\begin{eqnarray}\label{eq:smallesttau}
\tau^*_{\btheta}=\arg\min_{\tau\in\Omega_n} P (\tau\mid\btheta)\ \mbox{and}\ \tau^{**}_{\btheta}=\arg\min_{\tau\in\Omega_n\setminus\tau^*_{\btheta}} P (\tau\mid\btheta)
\end{eqnarray}
be the rankings with the smallest and second smallest sampling probability among the ranking space $\Omega_n$ given model parameter $\btheta$.
According to the likelihood function (18), it is easy to check that the solutions of the optimization problems in \eqref{eq:smallesttau} depend on $\cI$ only and have the general form below:
$$\tau^*_{\btheta}=inv(\cI)\ \mbox{and}\  \tau^{**}_{\btheta}=swap_\cI(\tau^*_{\btheta}),$$
where the operator $inv(\cI)$ locates the relevant entities defined by $\cI$ to the last $n_1$ positions in the ranking list with their internal order reversed and the background entities randomly to the other $(n-d)$ open positions, and the operator $swap_\cI(\tau^*_{\btheta})$ swaps the positions of either two adjacent relevant entities  defined by $\cI$ or the last relevant entity and an arbitrary background entity in $\tau^*_{\btheta}$.
For instance, assume that $n=20$, $n_1=10$ and $\cI= (1,2,\cdots,n_1,0,\cdots,0)$, we have:
\begin{eqnarray*}
\tau^*_{\btheta}&=&\ inv(\cI)\ =\  \big(20,19,\cdots,13,12,11,perm(10,\cdots,2,1)\big),\\
\tau^{**}_{\btheta}&=&swap_\cI(\tau^*_{\btheta})=\big(20,19,\cdots,13,11,12,perm(10,\cdots,2,1)\big),\\
or\ \tau^{**}_{\btheta}&=&swap_\cI(\tau^*_{\btheta})=\big(20,19,\cdots,13,12,10,perm(11,9,\cdots,2,1)\big),
\end{eqnarray*}
where $perm(S)$ stands for a random permutation of the input sequence $S$.
Note that although both $\tau^*_{\btheta}$ and $\tau^{**}_{\btheta}$ have a lot of variations, the different variations correspond to the same sampling probability below:
\begin{eqnarray*}
p^*_{\btheta} &\triangleq& P(\tau^*_{\btheta}\mid\btheta)=\dfrac{1}{(n-n_1)!}\times\dfrac{(n_1+1)^{-(n-n_1)\gamma}}{(C^*_{\gamma,n_1})^{n-n_1}}\times\dfrac{\exp\left\{-\frac{n_1(n_1-1)}{2} \phi \gamma\right\} }{Z(\phi \gamma)}=h(\phi,\gamma),\\
p^{**}_{\btheta}&\triangleq& P(\tau^{**}_{\btheta}\mid\btheta)=h(\phi,\gamma)\times \min\{(n-n_1)\times (1+1/n_1)^{\gamma},\exp\{\phi\gamma/2\} \},
\end{eqnarray*}
where $C^*_{\gamma,n_1}=\sum_{t=1}^{n_1+1} t^{-\gamma}$.
Further, define 
$$r^*_\btheta=\frac{p^{**}_{\btheta}}{p^{*}_{\btheta}}=\min\big\{(n-n_1)\times (1+1/n_1)^{\gamma},\exp\{\phi\gamma/2\} \big\}=g(\phi,\gamma).$$
It can be showed via proof by contradiction that the following two equations can not hold simultaneously for $(\phi_1,\gamma_1)\neq (\phi_2,\gamma_2)$: $$h(\phi_1,\gamma_1) = h(\phi_2,\gamma_2)\ \mbox{and}\ g(\phi_1,\gamma_1) = g(\phi_2,\gamma_2).$$
Because if $g(\phi_1,\gamma_1) = g(\phi_2,\gamma_2)$, we would have either $\phi_1 \gamma_1 = \phi_2 \gamma_2$ or $\gamma_1 = \gamma_2$ based on the definition of function $g(\phi,\gamma)$, which leads to $(\phi_1,\gamma_1)= (\phi_2,\gamma_2)$ as function $h(\phi,\gamma)$ is monotonically decreasing with respect to $\phi$ and $\gamma$. 
Apparently, the above fact indicates that the following two equations can not hold simultaneously:
\begin{eqnarray*}
P(\tau^*_{\btheta_1}\mid\btheta_1)&=&P(\tau^*_{\btheta_2}\mid\btheta_2)\ \mbox{and}\  P(\tau^{**}_{\btheta_1}\mid\btheta_1)=P(\tau^{**}_{\btheta_2}\mid\btheta_2),
\end{eqnarray*}
which means that the two distributions $P(\tau\mid\btheta_1)$ and $P(\tau\mid\btheta_2)$ are not identical.
Therefore, there must exist $\tau\in\Omega_n$ such that $P(\tau\mid\btheta_1)\neq P(\tau\mid\btheta_2)$.

$(ii)$ The scenario where $(\phi_1,\gamma_1)= (\phi_2,\gamma_2)$ but $\cI_1\neq\cI_2$. As $\tau^*_{\btheta_1}$ and $\tau^*_{\btheta_2}$ are the minima of $P(\tau\mid\btheta_1)$ and $P(\tau\mid\btheta_2)$, and $\tau^*_{\btheta_1}\neq\tau^*_{\btheta_2}$ in this case, we have  $$P(\tau^*_{\btheta_1}\mid\btheta_1) = h(\phi_1,\gamma_1) = h(\phi_2,\gamma_2)= P(\tau^*_{\btheta_2}\mid\btheta_2) <P(\tau^*_{\btheta_1}\mid\btheta_2).$$ Similarly, we have $$ P(\tau^*_{\btheta_2} \mid\btheta_1)>P(\tau^*_{\btheta_2}\mid\btheta_2).$$
Therefore, there exists $\tau\in\Omega_n$ (e.g., $\tau^{*}_{\btheta_1}$ or $\tau^{*}_{\btheta_2}$) such that $P(\tau\mid\btheta_1)\neq P(\tau\mid\btheta_2)$.
Combining the above two scenarios, we conclude that fact (20) holds for $m=1$.

Next, we prove the more general cases where $m>1$ by mathematical induction: assuming that (20) holds for all $m\leq M$, we will prove that it also holds for $m=M+1$.
Considering that for any $m>1$, $P(\btau\mid\btheta_1)=P(\btau\mid\btheta_2)$ implies:
$$\prod_{k=1}^mP(\tau_k\mid\cI_1,\phi_1,\gamma_{1k})= \prod_{k=1}^mP(\tau_k\mid\cI_2,\phi_2,\gamma_{2k}),$$
where $\btau = (\tau_1,\cdots,\tau_m)$ and $\btheta_t=(\cI_t,\phi_t,\bgamma_t)$ with $\bgamma_t=\{\gamma_{tk}\}_{1\leq k\leq m}$, the condition in Theorem 1 for $m=M+1$ suggests that
\begin{equation}\label{eq:IdentifiablityConditionForM}
    \prod_{k=1}^{M+1}P(\tau_k\mid\cI_1,\phi_1,\gamma_{1k})= \prod_{k=1}^{M+1}P(\tau_k\mid\cI_2,\phi_2,\gamma_{2k}),\ \forall\ \btau\in\Omega_n^{M+1}.
\end{equation}
Define $\btau_{[-j]}=(\tau_1,\cdots,\tau_{j-1},\tau_{j+1},\cdots,\tau_{M+1})$, $\bgamma_{t[-j]}=\{\gamma_{tk}\}_{k\neq j}$, $\btheta_{t[-j]}=(\cI_t,\phi_t,\bgamma_{t[-j]})$ and
$$P(\btau_{[-j]}\mid\btheta_{t[-j]})=\prod_{k\neq j}P(\tau_k\mid\cI_t,\phi_t,\gamma_{tk}).$$
For $\forall\ 1\leq j\leq M+1$, equation \eqref{eq:IdentifiablityConditionForM} can be expressed alternatively as below:
\begin{equation*}
    P(\btau_{[-j]}\mid\btheta_{1[-j]})\cdot P(\tau_j\mid\cI_1,\phi_1,\gamma_{1j})=P(\btau_{[-j]}\mid\btheta_{2[-j]})\cdot P(\tau_j\mid\cI_2,\phi_2,\gamma_{2j}),\ \forall\ \btau\in\Omega_n^{M+1}.
\end{equation*}
For any fixed $\btau_{[-j]}$, summing over all equations of the above form for all $\btau$ that is compatible with $\btau_{[-j]}$, we have
$$P(\btau_{[-j]}\mid\btheta_{1[-j]})\cdot\sum_{\tau_j\in\Omega_n}P(\tau_j\mid\cI_1,\phi_1,\gamma_{1k})=P(\btau_{[-j]}\mid\btheta_{2[-j]})\cdot\sum_{\tau_j\in\Omega_n}P(\tau_j\mid\cI_2,\phi_2,\gamma_{2k})$$
Considering that 
$$\sum_{\tau_j\in\Omega_n}P(\tau_j\mid\cI_1,\phi_1,\gamma_{1k})=\sum_{\tau_j\in\Omega_n}P(\tau_j\mid\cI_2,\phi_2,\gamma_{2k})=1,$$
we have 
$$P(\btau_{[-j]}\mid\btheta_{1[-j]})=P(\btau_{[-j]}\mid\btheta_{2[-j]}),\ \forall\ \btau_{[-j]}\in\Omega_n^M,$$
which indicates that
$$\btheta_{1[-j]}=\btheta_{2[-j]},\ \forall\ j=1,\cdots,M+1.$$
Thus, we have $\btheta_1=\btheta_2$, and the proof is complete.
\end{proof}

\section{Proof of Lemma 1} \label{proof:lemma1}

\begin{proof}
The function $e_i(\phi)$ is given as
\begin{eqnarray}
e_i(\phi) &=& \int\bbE\big[\tau(i)\mid\gamma\big]dF(\gamma) =\int\sum_{j=1}^n jP(\tau(i)=j \mid \cI,\phi,\gamma) dF(\gamma),
\end{eqnarray}
where 
\begin{eqnarray*}   
&&P(\tau(i)=j \mid \cI,\phi,\gamma)=\sum_{k= \max(1,j-n+n_1)}^{\min(n_1, j)} P(\tau^1(i)=k)P\left(\sum_{e\in U_B} \mathbb{I} (\tau^{0|1}(e)\geq n_1+2-k) = j-k\right),\\
&&P(\tau^1(i)=k)=\dfrac{e^{-\phi\gamma (i-k)}}{\sum_{l=0}^{n_1-1} e^{-\phi\gamma l}}  ,~~~k\leq n_1,\\
&&P\left(\sum_{e\in U_B}\mathbb{I}(\tau^{0|1}(e) \geq n_1+2-k)=j-k\right) = \begin{pmatrix} n-n_1\\j-k \end{pmatrix} \left(\dfrac{\sum_{l\geq n_1+2-k} l^{-\gamma}}{\sum_{l=1}^{n_1+1}l^{-\gamma}}\right)^{j-k}\left(\dfrac{\sum_{l\leq n_1+1-k} l^{-\gamma}}{\sum_{l=1}^{n_1+1}l^{-\gamma}}\right)^{n^{'}},\\
&&n^{'}= n-n_1-j+k.
\end{eqnarray*} 

As $P\left(\sum_{e\in U_B}\mathbb{I}(\tau^{0|1}(e) \geq n_1+2-k)=j-k\right)$ is a constant with respect to $\phi$, we denote $g(\gamma) \triangleq P\left(\sum_{e\in U_B}\mathbb{I}(\tau^{0|1}(e) \geq n_1+2-k)=j-k\right)$. Thus,
\begin{eqnarray}\label{eqn:relevantmean}
e_i(\phi) &=&\int\sum_{j=1}^n jP(\tau(i)=j \mid \cI,\phi,\gamma) dF(\gamma)\nonumber\\
&=& \int\sum_{j=1}^n j \sum_{k= \max(1,j-n+n_1)}^{\min(n_1, j)} \dfrac{e^{-\phi\gamma (i-k)}}{\sum_{l=0}^{n_1-1} e^{-\phi\gamma l}} g(\gamma) dF(\gamma).
\end{eqnarray}

From Equation \eqref{eqn:relevantmean}, we can see that $e_i(\phi)$ is a continuous function of $\phi$ for $i\in U_R$, not an infinite oscillation function,
and degenerates to a constant $e_0$ for $i\in U_B$. Thus, equation $e_i(\phi)=e_0$ has only finite solutions in $[0,\phi_{max})$.
Let $\cS_i$ be the solutions of equation $e_i(\phi)=e_0$ for $i\leq n_1$ and $\cS=\cup_{i=1}^d\cS_i$. We have $\tilde\Omega_\phi=\Omega_\phi-\cS$.
\end{proof}

\section{Proof of Theorem 3}
\label{AP:consisPhiPsi}
\begin{proof}
Based on the classic theory for MLE consistency \citep{awald1949}, to prove that $(\hat\phi_\cI,\hat\psi_\cI)$ are consistent, we only to show that the following regularity conditions hold for the PAMA-H model: 
(1) probability measure $L_k(\phi,\psi)=\int P(\tau_k\mid\cI,\phi,\gamma_k)dF_\psi(\gamma_k)$ keeps the same support for all $(\phi,\psi)\in\Omega_\phi\times\Omega_\psi$, 
(2) the true parameter $(\phi_0,\psi_0)$ is an interior point of the parameter space $\Omega_\phi\times\Omega_\psi$, 
(3) the log-likelihood $l(\phi,\psi\mid\cI)$ is differentiable with respect to $\phi$ and $\psi$, 
and (4) the MLE $(\hat{\phi},\hat\psi)$ is the unique solution of the score equations $\frac{\delta l(\phi,\psi)}{\delta \phi}=0$ and $\frac{\delta l(\phi,\psi)}{\delta \psi}=0$.

It is easy to check that $l(\phi,\psi\mid\cI)$ satisfies regular conditions (1)-(3). The regular condition (4) is satisfied by showing that $l(\phi,\psi\mid\cI)$ is a concave function with respect to $\phi$ and $\psi$ respectively.
The log-likelihood of PAMA-H is
\begin{eqnarray*}
l(\phi,\psi\mid\cI)&=&\sum_{k=1}^m l_k(\phi,\psi\mid\cI)=\sum_{k=1}^m\log\left(\int P(\tau_k\mid\cI,\phi,\gamma_k)dF_\psi(\gamma_k)\right)\\
&=&\sum_{k=1}^m \log \left ( \int_0^\infty \frac{\exp \{ -\phi \gamma_k d_{\tau} (\tau_k, \cI^+)\} (1-e^{-\phi \gamma_k})^{n_1  -1}}{\prod_{t=2}^{n_1} (1-e^{-t \phi \gamma_k}) } \frac{f_\psi(\gamma_k)}{G(\gamma_k)}d\gamma_k \right)\\
&=&\sum_{k=1}^m \log \left ( \int_0^\infty R(\phi,\gamma_k) \frac{f_\psi(\gamma_k)}{G(\gamma_k)}d\gamma_k \right)= \sum_{k=1}^m \log (H_k(\phi,\psi)),
\end{eqnarray*}
where $H_k(\phi,\psi) = \int_0^\infty R(\phi,\gamma_k) \frac{f_\psi(\gamma_k)}{G(\gamma_k)}d\gamma_k$, $R(\phi,\gamma_k)=\frac{\exp \{ -\phi \gamma_k d_{\tau} (\tau_k, \cI^+)\} (1-e^{-\phi \gamma_k})^{n_1  -1}}{\prod_{t=2}^{n_1} (1-e^{-t \phi \gamma_k}) }$ and $G(\gamma_k) = A^*_{\tau_k,I} (B^*_{\tau_k,I})^{\gamma_k} (C^*_{\gamma_k,n_1})^{n_0}$. According to properties of the Mallows model, it is easy to check $\frac{\partial^2 R(\phi,\gamma_k)}{\partial ^2 \phi} <0$ and thus $R(\phi,\gamma_k)$ is a concave function with respect to $\phi$. As $\frac{f_\psi(\gamma_k)}{G(\gamma_k)}>0$, then $H_k(\phi,\psi)$ preserves convexity of $R(\phi,\gamma_k)$ according to the properties of a concave function. 
Obviously, $l(\phi,\psi\mid\cI)$ is a 
composite function of $H_k(\phi,\psi)$ with composite functions being logarithm and summation. Such composite function preserves convexity of $H_k(\phi,\psi)$. Therefore, $l(\phi,\psi\mid\cI)$ is a concave function with respect to $\phi$. Similarly, given that $f_\psi(\gamma_k)$ is a concave function of $\psi$, it is easy to show $l(\phi,\psi\mid\cI)$ is a concave function with respect to $\psi$. Therefore, regular condition (4) is satisfied.
Thus, the proof is complete.

\end{proof}

\section{Details of the Gauss-Seidel Iterative Optimization} \label{AP:mle}
Let $\btheta^{(s)}=(\cI^{(s)},\phi^{(s)},\bgamma^{(s)})$ be the maximizer obtained at cycle $s$. The \textit{Gauss-Seidel} method works as follows:
\begin{enumerate}
\item Update $\gamma_k,k=1,\cdots,m$. Define $g(\gamma_k)$ as partial function of log-likelihood with respect to $\gamma_k$ given the other parameters. Newton-like method is adopted to update $\gamma_k$ from $\gamma_k^{(s)}$ to $\gamma_k^{(s+1)}$. 
\item Update $\phi$. Define $g(\phi)$ as partial function of log-likelihood with respect to $\phi$ given other parameters. Similarly, Newton-like method is adopted to update $\phi$ from $\phi^{(s)}$ to $\phi^{(s+1)}$. 
\item Update $\cI$. Let $g(\cI)$ denote the log-likelihood as a function of $\cI$ with other parameters fixed. We randomly select two neighboring entities and swap their rankings to check whether $g(\cI\mid\bgamma^{(s+1)}, \phi^{(s+1)})$ increases.
\end{enumerate}

The procedure starts from a random guess of all the parameters and then repeat Steps 1-3 until the log-likelihood converges. The convergence of likelihood is achieved when the difference of log-likelihood between any two consecutive iterations is less than 0.1. The difficulty of the procedure lies in the update of $\cI$. Practically, the starting point of $\cI$ can be some quick estimation of $\cI$, such as an estimate from the Mallows model. 

In each cycle of the Gauss-Seidel update for finding the MLE, a {partial function of likelihood} has to be defined. Suppose the current cyclic index is $s$,  the detailed computation of $\btheta^{(s+1)}$ is given below. 

\subsection{Update $\gamma_k$} \label{APsub:gamma}
Define $g(\gamma_k)=\log(f(\gamma_k| \bgamma_{[1:k-1]}^{(s+1)},\bgamma_{[k+1:m]}^{(s)},\cI^{(s)},\phi^{(s)}))$. Thus, for $\gamma_k \in [0,10]$,
\begin{small}
\begin{eqnarray}
\nonumber
g(\gamma_k) &=& \log\left(\frac{1}{(B^*_{\tau_k,I})^{\gamma_k}\times(C^*_{\gamma_k,n_1})^{n-n_1}\times (D^*_{\tau_k,\cI})^{\phi\cdot\gamma_k}\times E^*_{\phi,\gamma_k}}\right) \\
&=& -\gamma_k \times \log(B^*_{\tau_k,I}) -(n-n_1)\times \log(C^*_{\gamma_k,n_1})-\phi^{(s)} \times\gamma_k \times \log(D^*_{\tau_k,\cI}) -\log(E^*_{\phi,\gamma_k}).
\end{eqnarray}
\end{small}
The updating equation is given as
\begin{eqnarray}
\gamma_k^{(s+1)}= \gamma_k^{(s)} - \alpha_{\gamma_k} \frac{g'(\gamma_k)}{g''(\gamma_k)},
\end{eqnarray}
where $\alpha_{\gamma_k}$ is step length parameter which can be tuned to control convergence of $g(\gamma_k)$. And
\begin{eqnarray}
g'(\gamma_k)&=& - \log(B^*_{\tau_k,I}) -(n-n_1)\times \frac{\partial \log(C^*_{\gamma_k,n_1})}{\partial \gamma_k} -  \phi^{(s)} \times \log(D^*_{\tau_k,\cI}) - \frac{\partial \log(E^*_{\phi,\gamma_k})}{ \partial \gamma_k}, \\
g''(\gamma_k)&=& -(n-n_1)\times \frac{\partial^2 \log(C^*_{\gamma_k,n_1}) }{\partial^2 \gamma_k}- \frac{\partial^2 \log(E^*_{\phi,\gamma_k}) }{\partial^2 \gamma_k},
\end{eqnarray}
where
\begin{eqnarray}
\nonumber
\log(B^*_{\tau_k,I})&=& \sum_{t=1}^{n_1+1} n_{\tau_k,t}^{0|1} \log (t),\\ \nonumber
\frac{\partial \log(C^*_{\gamma_k,n_1})}{\partial \gamma_k}&=&-\frac{\sum_{t=1}^{n_1+1} t^{-\gamma_k} \log (t)}{C^*_{\gamma_k,n_1}},\\\nonumber
\log(D^*_{\tau_k,\cI}) &=& d_{\tau}(\tau_k^1,\cI^{+(s)}),\\\nonumber
\frac{\partial \log(E^*_{\phi,\gamma_k})}{ \partial \gamma_k} &=& \sum_{t=2}^{n_1} \frac{t \cdot \phi^{(s)} \exp\{-t\phi^{(s)}\gamma_k \}}{1-\exp\{-t\phi^{(s)}\gamma_k \}} -\frac{\phi^{(s)}\cdot(n_1-1)}{1-\exp\{-\phi^{(s)}\gamma_k\}}\times \exp\{-\phi^{(s)}\gamma_k\},\\ \nonumber
\frac{\partial^2 \log(C^*_{\gamma_k,n_1}) }{\partial^2 \gamma_k}&=& \frac{1}{(C^*_{\gamma_k,n_1})^2}[C^*_{\gamma_k,n_1}\sum_{t=1}^{n_1+1}  t^{-\gamma_k}\log (t)^2 + (\sum_{t=1}^{n_1+1} t^{-\gamma_k} \log (t))^2], \\ \nonumber
\frac{\partial^2 \log(E^*_{\phi,\gamma_k}) }{\partial^2 \gamma_k}&=& \sum_{t=2}^{n_1} \frac{-t^2\phi^{(s)^2}\exp\{-t\phi^{(s)}\gamma_k \}}{[1-\exp\{-t\phi^{(s)}\gamma_k \}]^2} + \frac{\phi^{(s)^2}(n_1-1)\exp\{-\phi^{(s)}\gamma_k\}}{[1-\exp\{-\phi^{(s)}\gamma_k\}]^2}.
\end{eqnarray}
\subsection{Update $\phi$} \label{APsubphi}
We define $g(\phi)=\log(f(\phi | \bgamma^{(s+1)}, \cI^{(s)}))$. Thus, for $\phi \in [0,10]$
\begin{eqnarray}
g(\phi)=-\sum_{k=1}^m \left(\phi \cdot \gamma_k^{(s+1)}  \log(D^*_{\tau_k,\cI}) +\log(E^*_{\phi,\gamma_k^{(s+1)}})\right).
\end{eqnarray}
The update equation is given as
\begin{eqnarray}
\phi^{(s+1)}= \phi^{(s)} - \alpha_{\phi} \frac{g'(\phi)}{g''(\phi)},
\end{eqnarray}
where $\alpha_{\phi}$ is step length parameter which controls convergence of $g(\phi)$.
And
\begin{eqnarray}
g'(\phi)&=&-\sum_{k=1}^m \left(\gamma_k^{(s+1)}  \log(D^*_{\tau_k,\cI})+ \frac{\partial [\log(E^*_{\phi,\gamma_k^{(s+1)}})]}{\partial \phi}\right), \\
g''(\phi)&=& -\sum_{k=1}^m  \frac{\partial^2 [\log(E^*_{\phi,\gamma_k^{(s+1)}})]}{\partial^2 \phi},
\end{eqnarray}
where $$\frac{\partial [\log(E^*_{\phi,\gamma_k^{(s+1)}})]}{\partial \phi}= \left(\sum_{t=2}^{n_1} \frac{t\cdot \gamma_k^{(s+1)} \exp\{-t\phi \gamma_k^{(s+1)} \}}{1-\exp\{-t \phi \gamma_k^{(s+1)}\}   } \right) - (n_1-1)\times \frac{\gamma_k^{(s+1)}\exp\{-\phi\gamma_k^{(s+1)}\}}{1-\exp\{-\phi\gamma_k^{(s+1)}\}},$$
$$\frac{\partial^2 [\log(E^*_{\phi,\gamma_k^{(s+1)}})]}{\partial^2 \phi}= \left(\sum_{t=2}^{n_1}  \frac{-t^2 \gamma_k^{(s+1)^2} \exp\{-t \phi \gamma_k^{(s+1)}\}}{[1-\exp\{-t \phi \gamma_k^{(s+1)}\}]^2}\right)  +  \frac{\gamma_k^{(s+1)^2}(n_1-1)\exp\{- \phi \gamma_k^{(s+1)}\}}{[1-\exp\{- \phi \gamma_k^{(s+1)}\}]^2}.$$

\subsection{Update $\cI$} \label{APsubI}
We define $g(\cI)=\log(f(\cI|\bgamma^{(s+1)},\phi^{(s+1)}))$. Thus,
\begin{eqnarray}
g(\cI)&=&\log\left(\prod_{k=1}^m\frac{1}{A^*_{\tau_k,I}\times(B^*_{\tau_k,I})^{\gamma_k}\times(D^*_{\tau_k,\cI})^{\phi\cdot\gamma_k}}  \right) .
\end{eqnarray}
Given current estimation $\cI^{(s)}$, the proposal of a new estimate can be obtained by iteratively swapping the neighboring entities in $\cI^{(s)}$. To be noticeable that the entity whose ranking is $n_1$ could be randomly swapped with any background entity. The proposed estimation is denoted by $\cI^{(s+\frac{1}{2})}$. If $g(\cI^{(s+\frac{1}{2})}) > g(\cI^{(s)})$, assign $\cI^{(s+\frac{1}{2})}$ to $\cI^{(s+1)}$. Otherwise, keep generating proposed estimation $\cI^{(s+\frac{1}{2})}$ until $g(\cI^{(s+\frac{1}{2})}) > g(\cI^{(s)})$ or no new proposal can be generated for $\cI$.

\section{Details of Inferring PAMA-H}\label{AP:PAMA-H}

\subsection{The Gauss-Seidel Iterative Optimization}
Let $\btheta^{(s)}=(\cI^{(s)},\phi^{(s)},\psi^{(s)})$ be the maximizer obtained at cycle $s$.
As $\bgamma$ is treated as missing data, we adapt MCEM (\cite{tanner1990}) to implement optimization. Then E step is: define $\hat{Q}^{(s+1)}(\btheta\mid \btheta^{(s)}) = \frac{1}{m^{(s)}} \sum_{i=1}^{m^{(s)}}  \log f(\cI^{(s)},\phi^{(s)},\psi^{(s)},\bgamma_i\mid\tau_1,\cdots,\tau_m)$, where $\bgamma_i$ is a sample drawn from $f(\bgamma \mid \cI^{(s)},\phi^{(s)},\psi^{(s)})$ which is defined in \eqref{app:PAMA-HB:gamma}. And M step is: maximize $\hat{Q}^{(s+1)}(\btheta\mid \btheta^{(s)})$ with respect to $\btheta$ to obtain $\btheta^{(s+1)}=(\cI^{(s+1)},\phi^{(s+1)},\psi^{(s+1)})$. The  \textit{Gauss-Seidel} method is utilized to conduct this optimization. The detailed procedure of Gauss-Seidel method is stated below. 

{\bf Update $\cI$}. With objective function being $\hat{Q}^{(s+1)}(\btheta\mid \btheta^{(s)})$, the procedure of updating $\cI$ is same as that in Section \ref{AP:mle}.

{\bf Update $\phi$.} With objective function being $\hat{Q}^{(s+1)}(\btheta\mid \btheta^{(s)})$, the procedure of updating $\phi$ is same as that in Section \ref{AP:mle}.

{\bf Update $\psi$.} Define $g(\psi)$ as partial function of $\hat{Q}^{(s+1)}(\btheta\mid \btheta^{(s)})$, Newton-like method is adopted to update $\psi$. Then $g(\psi) = \frac{1}{m^{(s)}} \sum_{i=1}^{m^{(s)}}  \log f_{\psi}(\bgamma_i)$. Suppose $F_\psi(\gamma)$ is an exponential distribution $f(\gamma\mid \alpha)$. Then $g(\alpha) = \log(\alpha)  - \frac{\alpha}{m^{(s)}} \sum_{i=1}^{m^{(s)}}\sum_{k=1}^m \gamma_{ik} ]$. It is straightforward to conduct Newton-like method to obtain $\psi^{(s+1)}$.

\subsection{The Bayesian Inference}
Suppose the $F_\psi(\gamma)$ is an exponential distribution $f(\gamma\mid \alpha)$. More specifically, $$f(x\mid \alpha) = \alpha  \exp\{-\alpha x\}, \alpha>0.$$ Note that $\psi = \alpha$ here. As the conditional distribution of $\cI$ and $\phi$ are already known in (27) and (28), we just show the conditional distributions of $\psi$ and $\gamma_k$ as below,
\begin{eqnarray}
f(\alpha\mid \cI,\phi,\bgamma)&\propto&\pi(\alpha)\cdot\prod_{k=1}^m{\alpha} \exp\{-\alpha \gamma_k\},\\
\label{app:PAMA-HB:gamma}
f(\gamma_k\mid\cI,\alpha,\psi,\bgamma_{[-k]})&\propto&\frac{ \exp\{-\alpha \gamma_k\}}{(B^*_{\tau_k,I})^{\gamma_k}\times(C^*_{\gamma_k,n_1})^{n-n_1}\times (D^*_{\tau_k,\cI})^{\phi\cdot\gamma_k}\times E^*_{\phi,\gamma_k}}.
\end{eqnarray}

\subsection{PAMA-H versus PAMA in Numerical Experiments}

We simulate 500 independent data sets with hierarchy distribution being an exponential distribution. The true parameters are $\cI=(1,2,\cdots,n_1,{\bf0}_{n_0}), \phi = 0.6, n_1= 10$ and $\alpha = 1$ for the parameter in the exponential distribution. Four different configurations of $n,m$ are considered. The average recovery distances and coverages are displayed in Table \ref{Tab:PAMAH}. Figure \ref{fig:PAMA-Hgamma} compares the estimation of $\gamma_k$ obtained by different models and frameworks. Both Table \ref{Tab:PAMAH} and Figure \ref{fig:PAMA-Hgamma} indicate that PAMA model performs as well as PAMA-H model. 

\begin{table}[htp]
    \centering
    \begin{tabular}{cc|c|cc|cc} 
    \hline 
    \multicolumn{2}{c|}{Configuration}& &\multicolumn{2}{c|}{Bayesian Inference} &\multicolumn{2}{c}{{MLE}}  \\ \cline{1-7} 
     $n$ &$m$ & Metric &PAMA$_{HB}$&PAMA$_B$  & PAMA$_{HF}$ & PAMA$_{F}$\\  \cline{1-7} 
\multirow{2}{*}{100}&\multirow{2}{*}{10}&Recovery& {\bf 0.70} (3.83)& 0.98 (6.11) & 19.09 (26.56)& 17.22  (24.80)    \\ 
 && Coverage&{\bf 1.00}  (0.01)& {\bf1.00} (0.01)& 0.96 (0.05)&0.97   (0.05) \\ \cline{1-7}
  \multirow{2}{*}{100}&\multirow{2}{*}{20}&Recovery& 1.22 (6.71)& {\bf 0.84} (5.09)& 33.76 (29.96) &32.16  (29.28)    \\ 
 && Coverage& {\bf1.00} (0.01) & {\bf1.00} (0.01) & 0.93 (0.06)&0.93 (0.06)   \\\cline{1-7}
 \multirow{2}{*}{200}&\multirow{2}{*}{10}&Recovery& 4.15 (24.22)& {\bf 3.36} (19.94)& 57.23 (57.02)& 49.08   (56.00)   \\ 
&& Coverage&{\bf 1.00} (0.02) & {\bf1.00} (0.02)& 0.94  (0.06)&0.95 (0.06)   \\\cline{1-7}
\multirow{2}{*}{200}&\multirow{2}{*}{20}&Recovery& {\bf 1.51} (10.86)& {\bf 1.51} (10.97)& 92.59 (65.94) & 85.67  (64.35)    \\
&& Coverage&{\bf 1.00} (0.01)& {\bf 1.00} (0.01) & 0.91 (0.07)&{ 0.91 (0.07)}   \\ 
\hline
 \hline 
    \end{tabular}
    \caption{Average recovery distances and coverages of different methods based on 500 independent replicates under PAMA-H model with different configurations of $n$ and $m$. The numbers in brackets are corresponding standard deviation.}
    \label{Tab:PAMAH}
\end{table}

\begin{figure}
    \centering
    \includegraphics{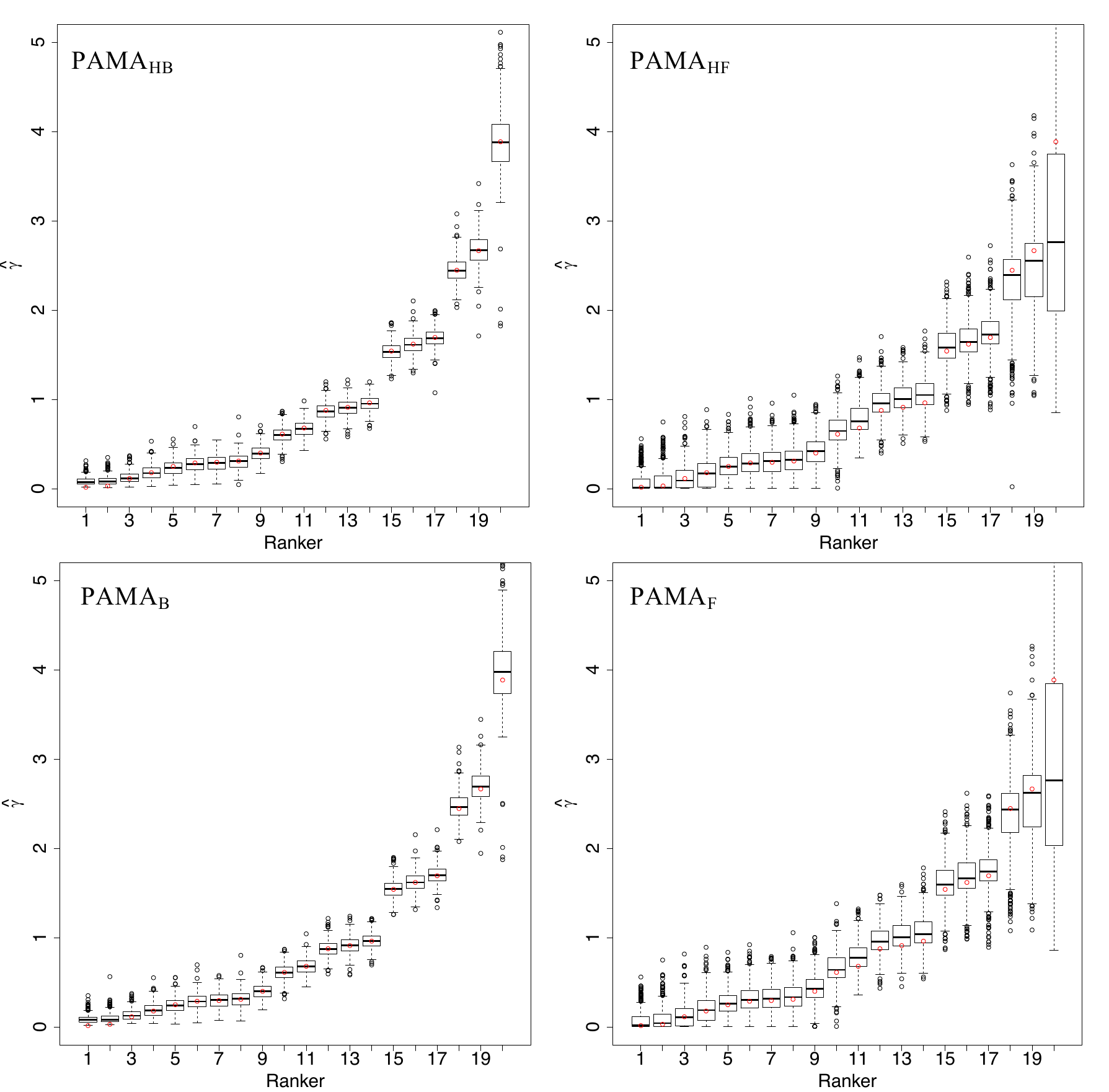}
    \caption{The boxplot of ${\hat{\gamma_k}}$ using different models and frameworks with $n=200$ and $m=20$.}
    \label{fig:PAMA-Hgamma}
\end{figure}

\section{Statistical Inference of the Covariate-Assisted Partition-Mallows Model}
\label{AP:Covariate-Assisted_Model}
\subsection{Bayesian Inference}

Due to the incorporation of $\bX$, the full conditional distributions may occur changes from the conditional distributions in Section 4.2. While the full conditional distributions of $\bgamma$ and $\phi$ keep the same as Equations (28) and (29), the full conditional distribution of $\cI$ changes to
\begin{eqnarray}
f(\cI\mid \cdot) &\propto& \exp\{\boldsymbol{\psi}^T\sum_{i:I_{i}=1}X_i\}\prod_{k=1}^m\frac{\mI\big(\tau_k^0 \in \mathcal{A}_{U_R}(\tau_k^{0\mid1})\big)}{A^*_{\tau_k,I}\times(B^*_{\tau_k,I})^{\gamma_k}\times(D^*_{\tau_k,\cI})^{\phi\cdot\gamma_k}}.
\end{eqnarray}

{In addition, the full conditional distribution of the $l^{th}$ element of  $\boldsymbol{\psi}$ is given as follows,}
\begin{eqnarray}
f(\boldsymbol{\psi}_l \mid \cdot) &\propto& \dfrac{\exp\{\boldsymbol{\psi}^T(\sum_{i:I_{i}>0} X_i)\}}{\prod_{i=1}^n(1+\exp\{\boldsymbol{\psi}^T X_i\})}, l=1,\cdots,p.
\end{eqnarray}
MH algorithm is also adopted to draw corresponding samples. Posterior point estimation can be calculated accordingly. 

\subsection{MLE}
Gauss-Seidel iterative optimization is adopted as well in optimizing covariate-assisted PAMA. Let $\btheta^{(s)}=(\cI^{(s)},\phi^{(s)},\bgamma^{(s)},\boldsymbol{\psi}^{(s)})$ be the maximizer obtained at cycle $s$. The detailed computation of $\btheta^{(s+1)}$ is given below. 
\subsubsection{Update of $\gamma_k$}
This is same as Section \ref{APsub:gamma}.
\subsubsection{Update of $\phi$}
This is same as Section \ref{APsubphi}.
\subsubsection{Update of $\cI$}
We define $g(\cI)=\log(f(\cI|\bgamma^{(s+1)},\phi^{(s+1)},\boldsymbol{\psi}^{(s)}))$, thus,
\begin{eqnarray}
g(\cI)&=&\log\left( P(\cI \mid \bX,\boldsymbol{\psi}^{(s)}) \times P(\tau_1,\cdots,\tau_m \mid \cI,\phi^{(s+1)},\bgamma^{(s+1)}) \right) 
\end{eqnarray}
Given current estimation $\cI^{(s)}$, the proposal of new estimation can be obtained by iteratively swapping the neighboring entities in $\cI^{(s)}$. To be noticeable that the entity whose ranking is $n_1$ could be randomly swapped with any background entity. The proposed estimation is denoted by $\cI^{(s+\frac{1}{2})}$. If $g(\cI^{(s+\frac{1}{2})}) > g(\cI^{(s)})$, assign $\cI^{(s+\frac{1}{2})}$ to $\cI^{(s+1)}$. Otherwise, keep generating proposed estimation $\cI^{(s+\frac{1}{2})}$ until $g(\cI^{(s+\frac{1}{2})}) > g(\cI^{(s)})$ or no new proposal can be generated for $\cI$. 
\subsubsection{Update of $\boldsymbol{\psi}$}
Define $g(\boldsymbol{\psi}) = \log(P(\cI^{(s+1)}\mid \bX,\boldsymbol{\psi}))$. Maximizing $g(\boldsymbol{\psi})$ with respect to $\boldsymbol{\psi}$ is actually a standard logistic regression problem. Therefore, $\boldsymbol{\psi}^{(s+1)} = {\arg\  max}_{\boldsymbol{\psi}} \log(P(\cI ^{(s+1)}\mid \bX,\boldsymbol{\psi}))$ by using standard statistical software. 

\section{Boxplots of $\bgamma$}\label{AP:boxplotwholegamma}
Figure \ref{Fig:wholegamma} demonstrates the boxplots of $\bgamma$ of all the setting with all combinations of $(n,m)$.

\begin{figure}[htp]
\centering
\includegraphics[width=0.95\linewidth]{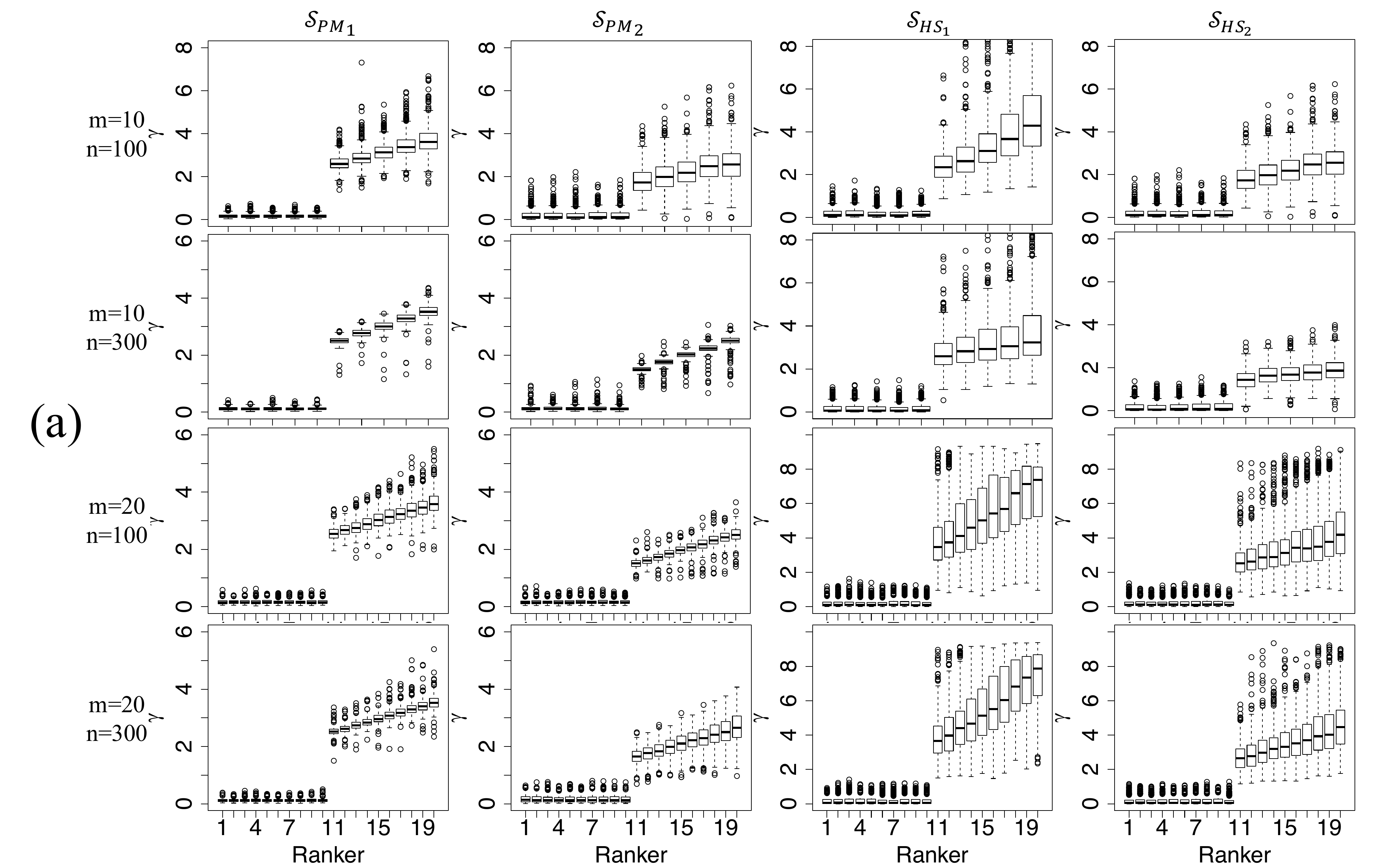}
\includegraphics[width=0.95\linewidth]{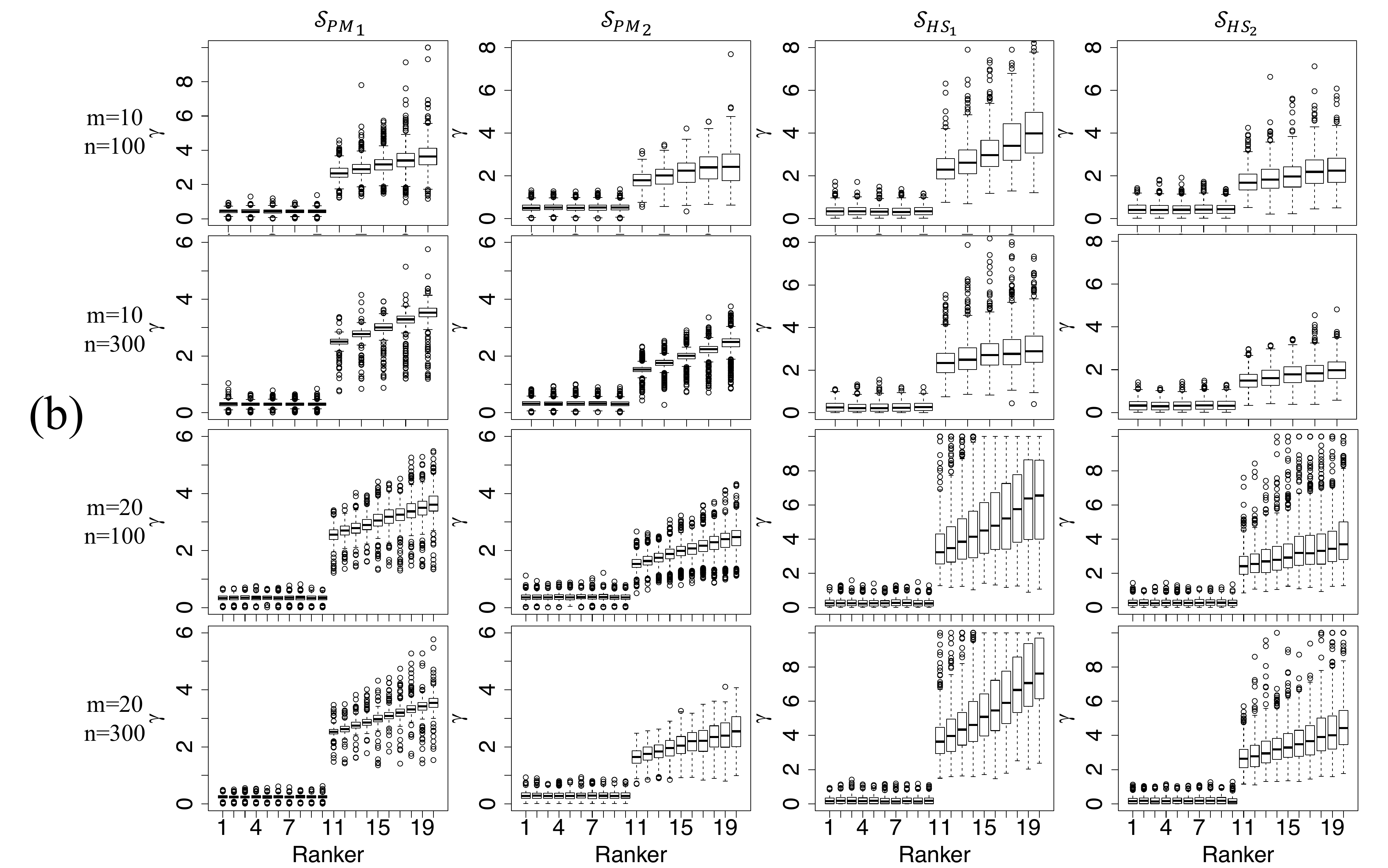}
\caption{(a) The boxplots of $\{\bar\gamma_k\}$ estimated by PAMA$_B$. (b) The boxplots of $\{\hat\gamma_k\}$ estimated by PAMA$_F$. Each row denotes a fixed combination of $m$ and $n$. Each column denotes a scenario setting. The boxplots are based on results from 500 independent replicates.}
\label{Fig:wholegamma}
\end{figure}

\section{Simulation Results}
\label{sec:SR}
Figure \ref{Fig:n100m20}, \ref{Fig:n300m10} and \ref{Fig:n300m20} present results from different methods for simulated data under various combinations of $n$ and $m$.

\begin{figure}[htp]
\centering
\includegraphics[width=\linewidth]{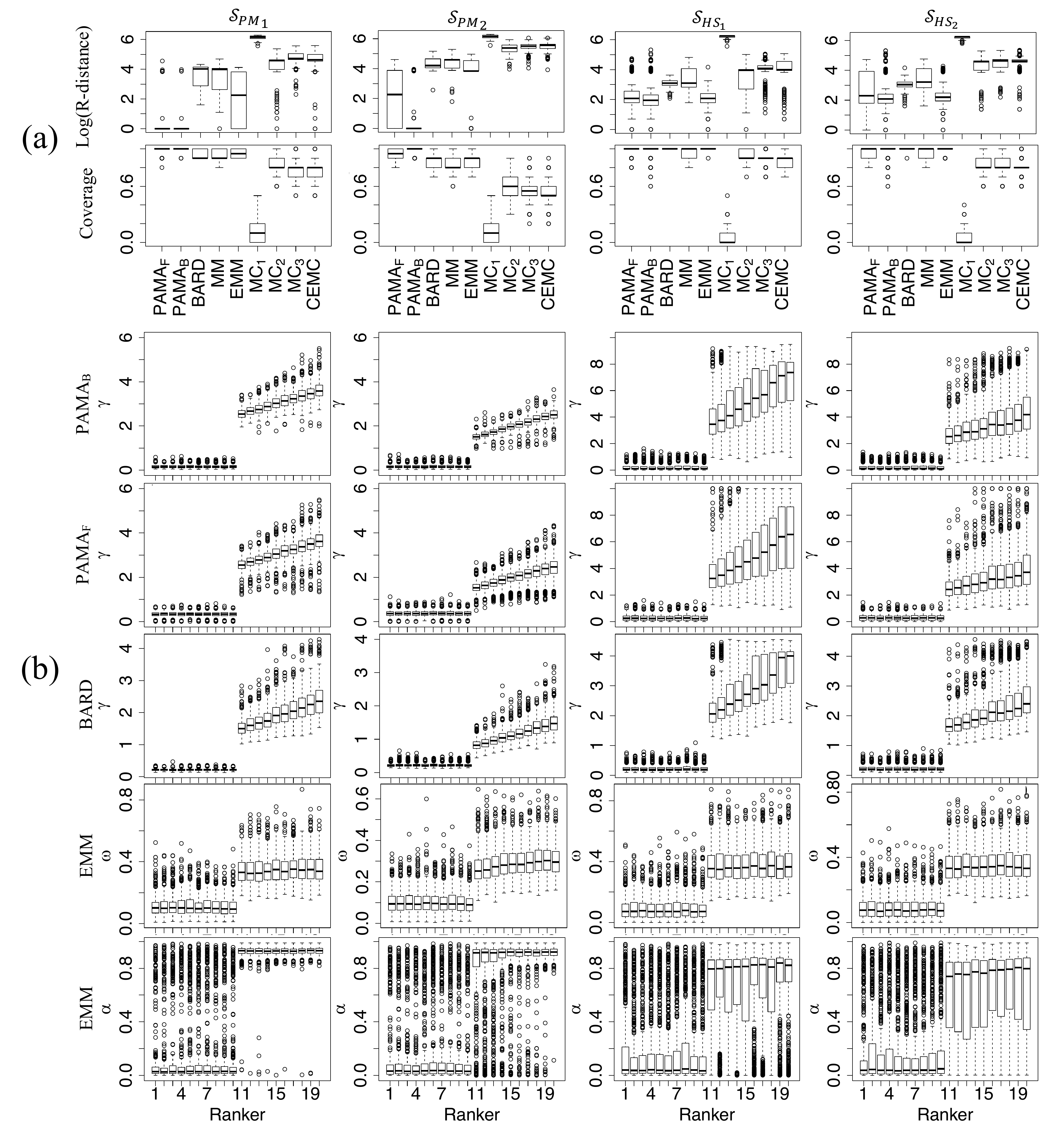} 
\caption{Results from different methods for simulated data under various scenarios with $m=20,n=100$, and $n_1=10$. The boxplots are based on 500 replications. (a) Recovery distances in log scale and coverage obtained from nine algorithms. (b) Quality parameters obtained by Partition-type models and EMM.}
\label{Fig:n100m20}
\end{figure}

\begin{figure}[htp]
\centering
\includegraphics[width=\linewidth]{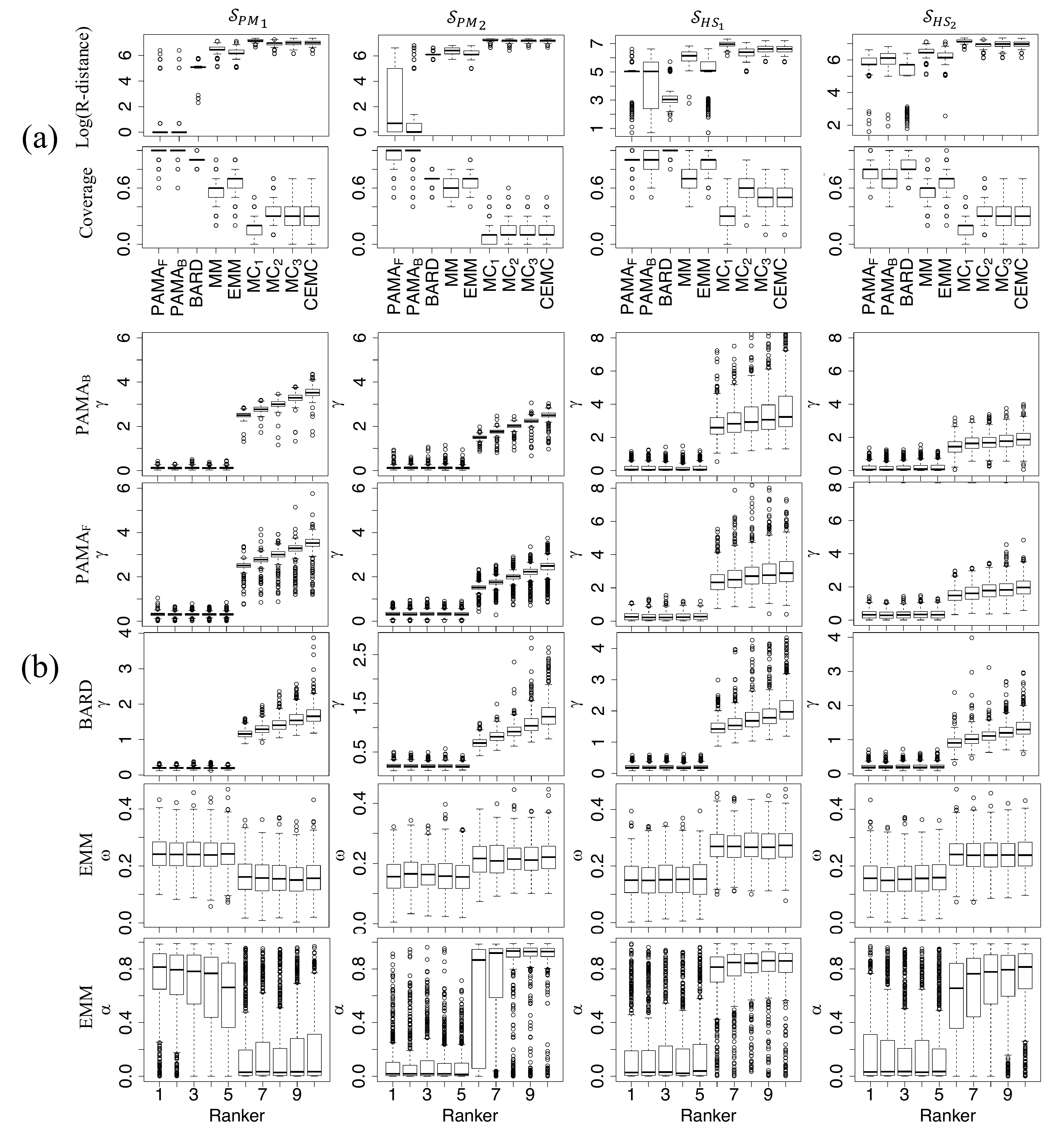} 
\caption{Results from different methods for simulated data under various scenarios with $m=10,n=300$, and $n_1=10$. The boxplots are based on 500 replications. (a) Recovery distances in log scale and coverage obtained from nine algorithms. (b) Quality parameters obtained by Partition-type models and EMM.}
\label{Fig:n300m10}
\end{figure}

\begin{figure}[htp]
\centering
\includegraphics[width=\linewidth]{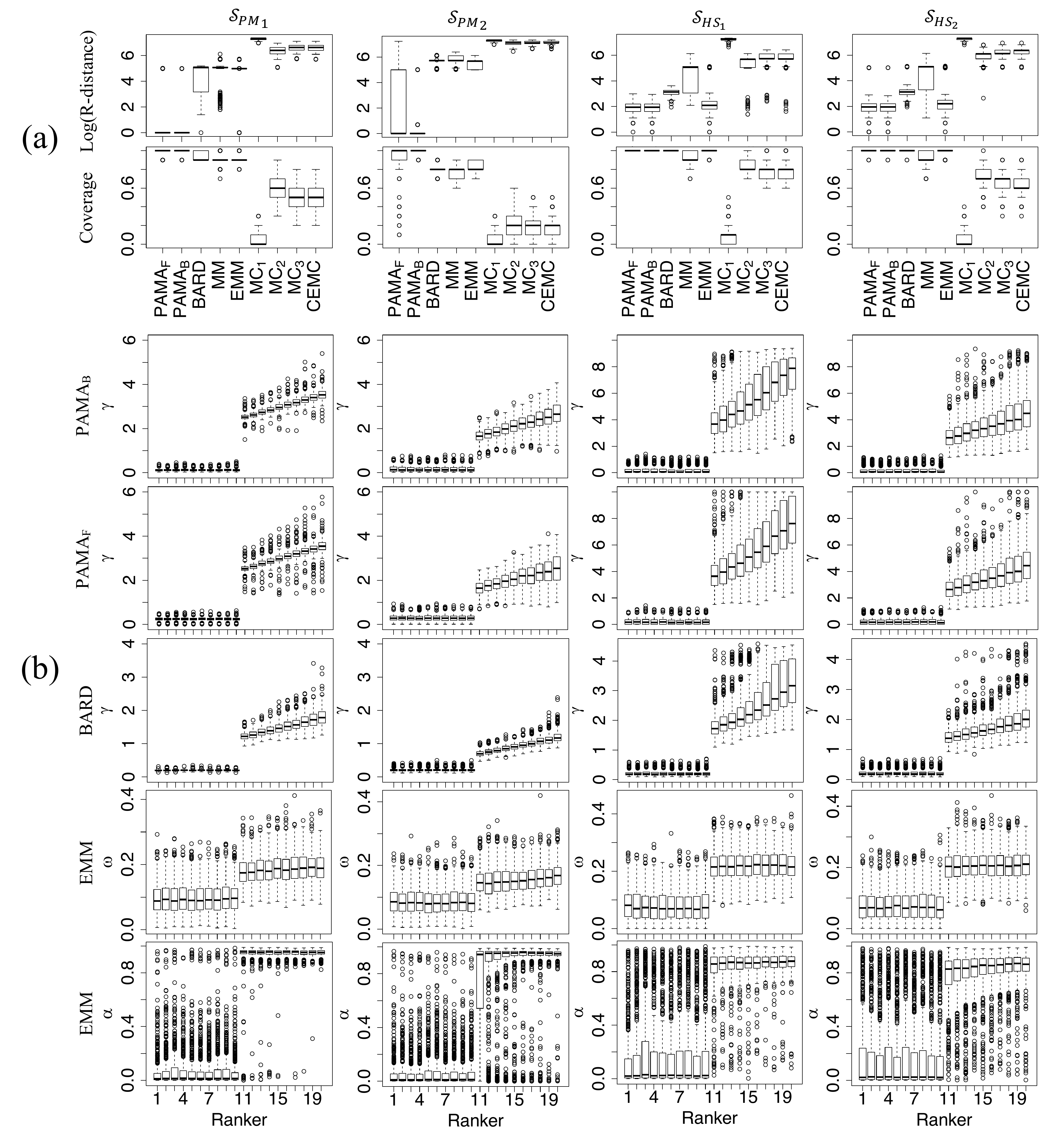} 
\caption{Results from different methods for simulated data under various scenarios with $m=20,n=300$, and $n_1=10$. The boxplots are based on 500 replications. (a) Recovery distances in log scale and coverage obtained from nine algorithms. (b) Quality parameters obtained by Partition-type models and EMM.}
\label{Fig:n300m20}
\end{figure}

\section{Robustness of $n_1$} \label{AP:robustness}
Figure \ref{fig:misspecifieddSM} shows boxplots of estimated $\cI$ for each mis-specified case (5, 12, 15). 
\begin{figure}
    \centering
    \includegraphics[width=0.8 \linewidth]{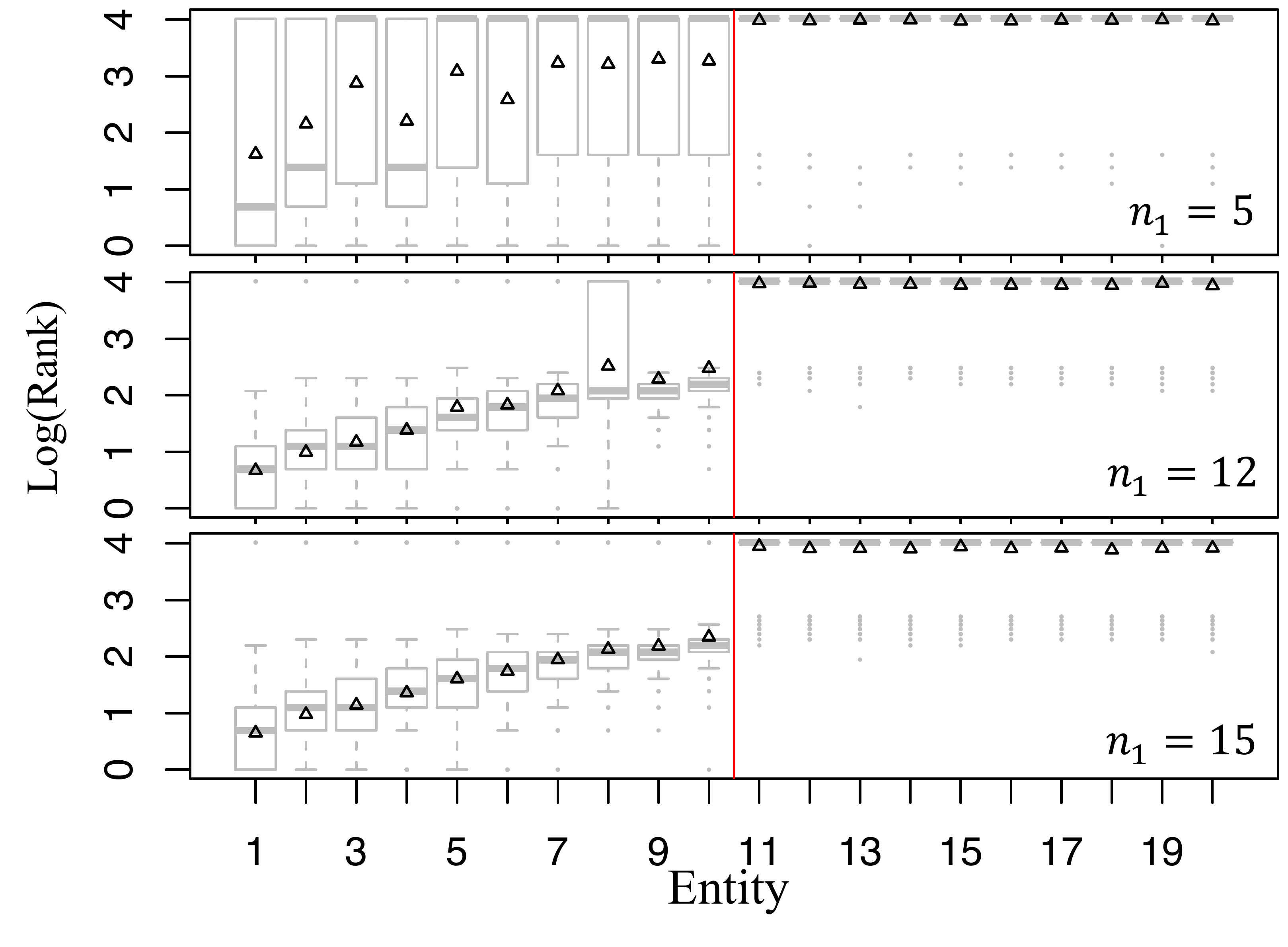}
    \caption{Boxplots of the estimated $\cI$ from  500 replications under the setting of $\cS_{HS_1}$ with $n_1$ being set as 5, 12 and 15, respectively. The true $n_1$ is $10$.
    The vertical lines separate relevant entities (left) from background ones.
    The Y-axis shows the logarithm of the entities' ranks. The rank of a background entity is replaced by their average $\frac{100+10+1}{2}$.
    The triangle denotes the mean rank of each entity.}
    \label{fig:misspecifieddSM}
\end{figure}

\bibliographystyle{apalike}
\bibliography{BARDMallows}

\begin{thebibliography}{}

\bibitem[Aslam and Montague, 2001]{aslam2001models}
Aslam, J.~A. and Montague, M. (2001).
\newblock Models for metasearch.
\newblock In {\em {Proceedings of the 24th annual international ACM SIGIR
  Conference on Research and Development in Information Retrieval}}, pages
  276--284. ACM.

\bibitem[Bader, 2011]{BADER20111099}
Bader, M. (2011).
\newblock The transposition median problem is {NP}-complete.
\newblock {\em Theoretical Computer Science}, 412(12):1099--1110.

\bibitem[Bhowmik and Ghosh, 2017]{Bhowmik2017}
Bhowmik, A. and Ghosh, J. (2017).
\newblock {LETOR} methods for unsupervised rank aggregation.
\newblock In {\em {Proceedings of the 26th International Conference on World
  Wide Web}}, pages 1331--1340. International World Wide Web Conferences
  Steering Committee.

\bibitem[Borda, 1781]{de1781memoire}
Borda, J.~C. (1781).
\newblock M{\'e}moire sur les {\'e}lections au scrutin.
\newblock {\em Histoire del' Acad{\'e}mie Royale des Sciences}.

\bibitem[Chen et~al., 2016]{chen2016drhp}
Chen, J., Long, R., Wang, X., Liu, B., and Chou, K. (2016).
\newblock {dRHP-PseRA}: detecting remote homology proteins using profile-based
  pseudo protein sequence and rank aggregation.
\newblock {\em Scientific Reports}, 6(32333).

\bibitem[{Chen} et~al., 2019]{fanjianqing2017Spectral}
{Chen}, Y., {Fan}, J., {Ma}, C., and {Wang}, K. (2019).
\newblock Spectral method and regularized mle are both optimal for top-{$K$}
  ranking.
\newblock {\em Annals of statistics}, 47(4):2204.

\bibitem[Chen and Suh, 2015]{chen2015spectral}
Chen, Y. and Suh, C. (2015).
\newblock Spectral {MLE}: Top-{$K$} rank aggregation from pairwise comparisons.
\newblock In {\em {International Conference on Machine Learning}}, pages
  371--380.

\bibitem[Deconde et~al., 2011]{Deconde2011Combining}
Deconde, R.~P., Hawley, S., Falcon, S., Clegg, N., Knudsen, B., and Etzioni, R.
  (2011).
\newblock Combining results of microarray experiments: a rank aggregation
  approach.
\newblock {\em Statistical Applications in Genetics \& Molecular Biology},
  5(1):1544--6115.

\bibitem[Deng et~al., 2014]{deng2014bayesian}
Deng, K., Han, S., Li, K.~J., and Liu, J.~S. (2014).
\newblock Bayesian aggregation of order-based rank data.
\newblock {\em Journal of the American Statistical Association},
  109(507):1023--1039.

\bibitem[Diaconis, 1988]{diaconis1988group}
Diaconis, P. (1988).
\newblock Group representations in probability and statistics.
\newblock {\em Lecture Notes-Monograph Series}, 11:1--192.

\bibitem[Diaconis and Graham, 1977]{diaconis1977spearman}
Diaconis, P. and Graham, R.~L. (1977).
\newblock Spearman's footrule as a measure of disarray.
\newblock {\em Journal of the Royal Statistical Society Series B
  (Methodological)}, 39(2):262--268.

\bibitem[Dwork et~al., 2001]{dwork2001rank}
Dwork, C., Kumar, R., Naor, M., and Sivakumar, D. (2001).
\newblock Rank aggregation methods for the web.
\newblock In {\em {Proceedings of the 10th International Conference on World
  Wide Web}}, pages 613--622. ACM.

\bibitem[Fagin et~al., 2003]{fagin2003efficient}
Fagin, R., Kumar, R., and Sivakumar, D. (2003).
\newblock Efficient similarity search and classification via rank aggregation.
\newblock In {\em {Proceedings of the 2003 ACM SIGMOD International Conference
  on Management of Data}}, pages 301--312. ACM.

\bibitem[Fligner and Verducci, 1986]{Fligner1986Distance}
Fligner, M.~A. and Verducci, J.~S. (1986).
\newblock Distance based ranking models.
\newblock {\em Journal of the Royal Statistical Society. Series B
  (Methodological)}, 48(3):359--369.

\bibitem[Freund et~al., 2003]{freund2003efficient}
Freund, Y., Iyer, R., Schapire, R.~E., and Singer, Y. (2003).
\newblock An efficient boosting algorithm for combining preferences.
\newblock {\em Journal of Machine Learning Research}, 4(Nov):933--969.

\bibitem[Hastings, 1970]{Hastings1970}
Hastings, W.~K. (1970).
\newblock {Monte Carlo sampling methods using Markov chains and their
  applications}.
\newblock {\em Biometrika}, 57(1):97--109.

\bibitem[Irurozki et~al., 2014]{Irurozki2014PerMallows}
Irurozki, E., Calvo, B., and Lozano, J.~A. (2014).
\newblock Permallows : An {R} package for {M}allows and generalized {M}allows
  models.
\newblock {\em Journal of Statistical Software}, 071(12):1--30.

\bibitem[Johnson et~al., 2020]{JohnsonS2019}
Johnson, S.~R., Henderson, D.~A., and Boys, R.~J. (2020).
\newblock Revealing subgroup structure in ranked data using a {B}ayesian
  {WAND}.
\newblock {\em Journal of the American Statistical Association},
  115(532):1888--1901.

\bibitem[Li et~al., 2020]{fan2019}
Li, H., Xu, M., Liu, J.~S., and Fan, X. (2020).
\newblock An extended mallows model for ranked data aggregation.
\newblock {\em Journal of the American Statistical Association},
  115(530):730--746.

\bibitem[Li et~al., 2021]{li2020bayesian}
Li, X., Yi, D., and Liu, J.~S. (2021).
\newblock Bayesian analysis of rank data with covariates and heterogeneous
  rankers.
\newblock {\em Statistical Science}, (In press).

\bibitem[Lin, 2010]{Lin2010Space}
Lin, S. (2010).
\newblock Space oriented rank-based data integration.
\newblock {\em Statistical Applications in Genetics \& Molecular Biology},
  9(1):Article20.

\bibitem[Lin and Ding, 2010]{2010Integration}
Lin, S. and Ding, J. (2010).
\newblock Integration of ranked lists via cross entropy monte carlo with
  applications to {mRNA} and {microRNA} studies.
\newblock {\em Biometrics}, 65(1):9--18.

\bibitem[Linas et~al., 2010]{baltrunas2010group}
Linas, B., Tadas, M., and Francesco, R. (2010).
\newblock Group recommendations with rank aggregation and collaborative
  filtering.
\newblock In {\em {Proceedings of the Fourth ACM Conference on Recommender
  Systems}}, pages 119--126. ACM.

\bibitem[Liu, 2008]{liu2008monte}
Liu, J.~S. (2008).
\newblock {\em Monte Carlo strategies in scientific computing}.
\newblock Springer Science \& Business Media.

\bibitem[Liu et~al., 2007]{liu2007supervised}
Liu, Y., Liu, T., Qin, T., Ma, Z., and Li, H. (2007).
\newblock Supervised rank aggregation.
\newblock In {\em {Proceedings of the 16th International Conference on World
  Wide Web}}, pages 481--490. ACM.

\bibitem[Luce, 1959]{luce1959}
Luce, R.~D. (1959).
\newblock {\em Individual choice behavior: a theoretical analysis}.
\newblock New York: Wiley.

\bibitem[Mallows, 1957]{mallows1957non}
Mallows, C.~L. (1957).
\newblock Non-null ranking models.
\newblock {\em Biometrika}, 44(1/2):114--130.

\bibitem[Plackett, 1975]{Plackett1975}
Plackett, R.~L. (1975).
\newblock The analysis of permutations.
\newblock {\em Journal of the Royal Statistical Society: Series C (Applied
  Statistics)}, 24(2):193--202.

\bibitem[Porello and Endriss, 2012]{porello2012ontology}
Porello, D. and Endriss, U. (2012).
\newblock Ontology merging as social choice: judgment aggregation under the
  open world assumption.
\newblock {\em Journal of Logic and Computation}, 24(6):1229--1249.

\bibitem[Rajkumar and Agarwal, 2014]{rajkumar2014statistical}
Rajkumar, A. and Agarwal, S. (2014).
\newblock A statistical convergence perspective of algorithms for rank
  aggregation from pairwise data.
\newblock In {\em {International Conference on Machine Learning}}, pages
  118--126.

\bibitem[Renda and Straccia, 2003]{renda2003web}
Renda, M.~E. and Straccia, U. (2003).
\newblock Web metasearch: rank vs. score based rank aggregation methods.
\newblock In {\em {Proceedings of the 2003 ACM Symposium on Applied
  Computing}}, pages 841--846. ACM.

\bibitem[Soufiani et~al., 2014]{soufiani2014statistical}
Soufiani, H.~A., Parkes, D.~C., and Xia, L. (2014).
\newblock A statistical decision-theoretic framework for social choice.
\newblock In {\em {Advances in Neural Information Processing Systems}}, pages
  3185--3193.

\bibitem[Tanner and Wong, 1987]{tanner1987}
Tanner, M.~A. and Wong, W.~H. (1987).
\newblock The calculation of posterior distributions by data augmentation.
\newblock {\em Journal of the American Statistical Association},
  82(398):528--540.

\bibitem[Thurstone, 1927]{Thurstone1927}
Thurstone, L.~L. (1927).
\newblock A law of comparative judgment.
\newblock {\em Psychological Review}, 34(4):273--286.

\bibitem[Wald, 1949]{awald1949}
Wald, A. (1949).
\newblock Note on the consistency of the maximum likelihood estimate.
\newblock {\em The Annals of Mathematical Statistics}, 20(4):595--601.

\bibitem[Wei and Tanner, 1990]{tanner1990}
Wei, G. C.~G. and Tanner, M.~A. (1990).
\newblock A {M}onte {C}arlo implementation of the {EM} algorithm and the poor
  man's data augmentation algorithms.
\newblock {\em Journal of the American Statistical Association},
  85(411):699--704.

\bibitem[Yang, 2018]{YANG2018281}
Yang, K.~H. (2018).
\newblock Chapter 7 - {S}tepping through finite element analysis.
\newblock In Yang, K.-H., editor, {\em {Basic Finite Element Method as Applied
  to Injury Biomechanics}}, pages 281--308. Academic Press.

\bibitem[Young, 1988]{young1988condorcet}
Young, H.~P. (1988).
\newblock Condorcet's theory of voting.
\newblock {\em American Political Science Review}, 82(4):1231--1244.

\bibitem[Young and Levenglick, 1978]{young1978consistent}
Young, H.~P. and Levenglick, A. (1978).
\newblock A consistent extension of condorcet’s election principle.
\newblock {\em SIAM Journal on Applied Mathematics}, 35(2):285--300.

\bibitem[Yu, 2000]{Yu2000Bayesian}
Yu, P. L.~H. (2000).
\newblock Bayesian analysis of order-statistics models for ranking data.
\newblock {\em Psychometrika}, 65(3):281--299.

\end{thebibliography}
\end{document}